\documentclass[letterpaper,twocolumn,american,showpacs,pr,aps,superscriptaddress,eqsecnum,nofootinbib]{revtex4}
\usepackage{amsfonts}%
\usepackage{amsmath}%
\makeatletter
\def\l@subsubsection#1#2{}
\makeatother
\usepackage{amssymb}%
\usepackage{graphicx}
\usepackage{color}
\usepackage{lipsum}
\usepackage{comment}
\usepackage{hyperref}
\providecommand{\U}[1]{\protect\rule{.1in}{.1in}}
\newtheorem{theorem}{Theorem}

\newtheorem{corollary}[theorem]{Corollary}

\newtheorem{definition}[theorem]{Definition}
\newtheorem{example}[theorem]{Example}

\newtheorem{proposition}[theorem]{Proposition}

\newenvironment{proof}[1][Proof]{\noindent\textbf{#1.} }{\ \rule{0.5em}{0.5em}}

\newcommand{\bes} {\begin{subequations}}
\newcommand{\ees} {\end{subequations}}
\newcommand{\bea} {\begin{eqnarray}}
\newcommand{\eea} {\end{eqnarray}}
\newcommand{\beq}{\begin{equation}}
\newcommand{\eeq}{\end{equation}}

\def\>{\rangle}
\def\<{\langle}
\def\Tr{\textrm{Tr}}

\renewcommand{\min}{\textrm{min}}

\newcommand{\ignore}[1]{}

\definecolor{nblue}{rgb}{0.2,0.2,0.7}
\definecolor{ngreen}{rgb}{0.2,0.6,0.2}
\definecolor{nred}{rgb}{0.7,0.2,0.2}
\definecolor{nblack}{rgb}{0,0,0}

\begin{document}

%\title{The resource theory of quantum coherence}
%\title{Resource theories of coherence}\
%\title{How to understand quantum coherence as a resource: a comparative study of three approaches}
%\title{How to understand quantum coherence as a resource : the importance of distinguishing speakable and unspeakable notions}
%\title{Distinguishing different notions of coherence as a resource}
%\title{Distinguishing speakable and unspeakable notions of quantum coherence is critical for understanding it as a resource}
%\title{Dinstinguishing speakable and unspeakable notions of quantum coherence and understanding them as resources}
%\title{Understanding quantum coherence as a resource: distinguishing speakable and unspeakable notions}
%\title{Understanding quantum coherence as a resource: speakable versus unspeakable coherence}
\title{How to quantify coherence: 
Distinguishing speakable and unspeakable notions}

\author{Iman Marvian}
\affiliation{Research Laboratory of Electronics, Massachusetts Institute of Technology, Cambridge, MA 02139}
\affiliation{Department of Physics and Astronomy, Center for Quantum Information Science and Technology, University of Southern California, Los Angeles, CA 90089}
\author{Robert W. Spekkens}
\affiliation{Perimeter Institute for Theoretical Physics, 31 Caroline St. N, Waterloo, \\
Ontario, Canada N2L 2Y5}

\date{\today}

\begin{abstract}
Quantum coherence is a critical resource for many operational tasks. Understanding how to quantify and manipulate it also promises to have applications for a diverse set of problems in theoretical physics.  For certain applications, however, one requires coherence between the eigenspaces of specific physical observables, such as energy, angular momentum, or photon number, and it makes a difference {\em which} eigenspaces appear in the superposition.  For others, there is a preferred set of subspaces relative to which coherence is deemed a resource, but it is irrelevant which of the subspaces appear in the superposition.  We term these two types of coherence unspeakable and speakable respectively. We argue that a useful approach to quantifying and characterizing {\em unspeakable} coherence is provided by the resource theory of asymmetry when the symmetry group is a group of translations, and we translate a number of prior results on asymmetry into the language of coherence.  
We also highlight some of the applications of this approach, for instance, in the context of quantum metrology, quantum speed limits, quantum thermodynamics, and NMR.  The question of how best to treat {\em speakable} coherence as a resource is also considered.  We review a popular approach in terms of operations that preserve the set of incoherent states, propose an alternative approach in terms of operations that are covariant under dephasing, and we outline the challenge of providing a physical justification for either approach. 
Finally, we note some mathematical connections that hold among the different approaches to quantifying coherence.

%\color{blue}
%There are two slightly different proposals in the literature for quantification and classification of coherence as a resource, and sometimes they have been confused with each other. The first proposal defines coherence as asymmetry relative to a translational symmetry, such as phase shift or time translation, while the second proposal defines coherence as the resource under \emph{incoherent operations}, operations which map incoherent states to incoherent states. Free operations in the first resource theory are a proper subset of the free operations in the second resource theory, and consequently, any measure of coherence in the second resource theory is also a  measure of coherence in the first resource theory, i.e. it is also a measure of asymmetry. On the other hand,  there are measures of asymmetry, such as the Wigner-Yanase-Dyson skew information, which can increase under incoherent operations, and so they are not measure of coherence in the second proposal.  In this paper we review these two resource theories of coherence, clarify their differences and discuss about their operational interpretations and applications in the contexts of quantum metrology, reference frames, quantum thermodynamics and quantum control. Furthermore, we provide a novel interpretation of quantum speed limits  in terms of the notion of coherence in the first proposal. 
%\color{black}
\end{abstract}

\maketitle

\tableofcontents

\section{Introduction and summary}

%Coherence is important. Blah Blah.(See also  \cite{yadin2014new})

%Resource theory approach  has seen a great deal of success in characterizing entanglement \cite{}, asymmetry \cite{gour2008resource,GMS09, MS11, MS_Short,   Modes, Marvian_thesis, Noether,  GConcurrence, skotiniotis2012alignment, narasimhachar2014phase}, and athermality \cite{}. 

\color{black}
%To take a resource-theoretic approach to a given quantum phenomenon is to try to and understand it 
Many properties of quantum states can be better understood by considering them as constituting a resource~\cite{coecke2014mathematical}.  The properties of entanglement \cite{horodecki2009quantum,Vedral:98}, asymmetry \cite{gour2008resource,GMS09, MS11, MS_Short,   Modes, Marvian_thesis, Noether,  GConcurrence, skotiniotis2012alignment, narasimhachar2014phase}, and athermality \cite{brandao2015second, brandao2013resource, gour2015resource, skrzypczyk2014work,  horodecki2013fundamental, aaberg2013truly} are good examples.  We are here concerned with the property of having coherence relative to some decomposition of the Hilbert space.  This property appears to 
%Quantum coherence is a property of quantum states that seems to 
be necessary for certain types of tasks, and as such it is natural to attempt to understand coherence from the resource-theoretic perspective.  This has led to some proposals for how to define a resource theory of coherence and in particular how to quantify coherence \cite{Coh_Plenio, Noether, marvian2015quantum, yadin2015quantum, yadin2015general}.  
%Why do we believe that it might be useful 

The following is a list of operational tasks for which quantum coherence seems to be a resource.
%What are the tasks to which a resource of coherence is put?  
%\begin{enumerate}
\begin{itemize}
\item 
%Quantum techniques for high precision metrology. 
Quantum metrology. An example is the task of estimating the phase shift on a field mode, used in quantum accelerometers and gravitometers.  Here one requires the ability to prepare and measure coherent superpositions of different occupation numbers of the mode. Another example is estimating the rotation of a quantum gyroscope about some axis, where one must be able to prepare and measure coherent superpositions of eigenstates of the angular momentum operator along that axis, a problem that is relevant for developing high precision measurements of magnetic field strength \cite{ajoy2012stable}. A third example is building high precision clocks, where one must be able to prepare and measure coherent superpositions of energy eigenstates \cite{QMetrology, Lloyd_Nature_Clock, giovannetti2011advances, schnabel2010quantum}. 
\item Reference frame alignment.
 %using quantum communication.  
 Examples include aligning distant gyroscopes, synchronizing distant clocks, and phase-locking distant phase references.  Each example requires communicating quantum states that carry the appropriate sort of information (orientation, time, and phase, for instance) and therefore, like the metrology examples, requires coherence relative to the appropriate eigenspaces \cite{QRF_BRS_07}.
%An example is the task of synchronizing distant clocks by the communication of quantum states carrying timing information, which requires a coherent superposition of energy eigenstates.   Another example is that of aligning distant gyroscopes by the communication of quantum states carrying directional information,  [citations]
\item 
%Quantum techniques for thermodynamic tasks. 
Thermodynamic tasks. An example is the task of extracting as much work as possible from a given quantum state given a bath at some fixed background temperature. This require states that are not in thermal equilibrium at the background temperature.  Resources here include not only those states having a nonthermal distribution of energy eigenstates, but also those that have coherence between the energy eigenspaces \cite{lostaglio2015quantum, lostaglio2015description, janzing2000thermodynamic, cwiklinski2015limitations}.  
%, which implies that latter sort of state is a thermodynamic resource. 
\item Computational, cryptographic and communication tasks.  
%Here, there is a computational basis in which preparations and measurements are implemented. 
For these sorts of tasks, it is well known that having access only to preparations and measurements that are all diagonal in some basis, hence incoherent, is not sufficient for achieving any quantum advantage.  So it is natural to seek to study such coherence as a resource.  
%\item Quantum key distribution. [citations]
%\end{enumerate}
\end{itemize}

%\color{red} [Athough quantum control gives a nice story about when the symmetric operations are easier than the asymmetric ones, it isn't clear how a coherent state is a resource in that context, so I've left it off the list for now.  Also, in what sense is coherence a resource in the study of quantum speed limits?  Is it a resource for rapid state evolution?]

For many resource theories, there are also  applications to problems in theoretical physics.  For example, while entanglement theory was originally developed through its role as a resource in operational tasks such as quantum teleportation \cite{bennett1993teleporting} and dense coding \cite{bennett1992communication}, the possibility of quantifying entanglement has since found applications in diverse problems, including the study of phase transitions, characterizing the ground states of many-body systems \cite{latorre2003ground, kitaev2006topological, vidal2003entanglement, marvian2013symmetry}, holography in quantum field theories \cite{ryu2006holographic}, and the black hole information-loss paradox \cite{almheiri2013black, verlinde2012black, hayden2007black}. 
%many-body physics, the study of phase transitions, and the black hole information loss paradox.  
Similarly, the possibility of quantifying coherence is expected to shed light on various problems in theoretical physics.  The following are a few examples.

\begin{itemize}
%\item{\bf Noether's theorem}
\item{Quantum speed limits.}  The Mandelstam-Tamm bound  \cite{QSL_MT}  and the Margolus-Levitin bound  \cite{Margolus:98} are upper bounds on the minimum time it takes for a system in some state to evolve to a (partially) distinguishable state.  This time is clearly related to the amount of coherence between energy eigenstates and therefore quantifying this coherence can shed light on quantum speed limits~\cite{marvian2015quantum, mondal2016quantum}. 
%It turns out that  the notion of coherence as translational asymmetry  provides a new interpretation of quantum speed limits, i.e., the Mandelstam-Tamm  \cite{QSL_MT} bound and the Margolus-Levitin bound  \cite{Margolus:98}. As explained in Ref.~\cite{Marvian2015}, the minimum time it takes for the system to evolve to a (partially) distinguishable state is the inverse of a particular measure of translational asymmetry (hence a particular measure of unspeakable coherence). Furthermore, the standard quantum speed limits can be interpreted as upper bounding this measure of coherence with another measure of coherence. 
\item{Magnetic resonance techniques.}  For such techniques, in particular in NMR,  if the system one is probing consists of many spins, then the large dimension of the Hilbert space together with constraints on the measurements are such that full tomography is not possible.  Still, one can obtain much useful information about the state by measuring the degree of coherence relative to  the quantization axis \cite{cappellaro2014implementation}. If the system is quantized along the $\hat{z}$ axis,  then  \emph{coherence of order} $q$  of the state $\rho$ is defined as the norm of the sum of the off-diagonal terms $\rho_{m_1m_2}|m_1\rangle\langle m_2|$ with $m_2-m_1=q$, where $|m\rangle$ is the eigenstate with eigenvalue $m$ of  $J_{\hat{z}}$, the magnetic moment in the $\hat{z}$ direction \cite{cappellaro2014implementation}\footnote{This concept is closely related to the idea of decomposition into modes of asymmetry, introduced in Refs.~\cite{Marvian_thesis, Modes} (See Sec.\ref{Modes} for a short review). In particular, 
the \emph{$q$th-quantum coherence component} in the language of \cite{cappellaro2014implementation} is the same as the \emph{mode $q$ component} in the language of \cite{Marvian_thesis, Modes}. Furthermore, the Frobenius norm of the \emph{$q$th-quantum coherence component}, which can be measured in NMR experiments, provides lower and upper bounds on a measure of coherence studied in  \cite{Marvian_thesis, Modes}, which quantifies the asymmetry of a state in mode $q$ (See Eq.(\ref{modeNMR})).}. Measuring the quantum coherence of different orders is relatively straightforward and has been useful in many NMR experiments, in particular, in the context of quantum information processing, as well as in simulations of many-body dynamics (See e.g. \cite{cappellaro2007simulations, cho2006decay, cappellaro2014implementation}).
\color{black}
\item{Coherence lengths.} The spatial extent over which a quantum state is coherent is an important concept in many-body physics~\cite{Frerot2015}, for instance, in the onset of Bose-Einstein condensation~\cite{ketterle1997coherence},  and in quantum biology, for instance, in excitation transport in photosynthetic complexes \cite{levi2015quantum, mohseni2008environment, rebentrost2009environment}. 
% The onset of Bose-Einstein condensation, for instance, is also related to a measure of such coherence \cite{ketterle1997coherence}.
\color{black}
\item{Order parameters.} Quantum phase transitions in the ground states of quantum many-body systems, such as a spin chain, can be studied in terms of the degree of coherence contained in local reductions of the state, such as single-spin or two-spin, density operators \cite{Karpat2014, malvezzi2016quantum}.
\item{Decoherence theory.}  It is well known that interaction of a system with its environment can lead to the loss of coherence relative to preferred subspaces that depend on the nature of the interaction \cite{joos2013decoherence}.  For instance, if an environment couples to the spatial degree of freedom of a system, then it will reduce the spatial extent over which the system exhibits coherence \cite{joos2013decoherence}.  Such decoherence plays a significant role in many accounts of the emergence of classicality.  Measures of coherence, therefore, can be used as a tool for studying such emergence.
%to formalize  Classical limit. Spatial decoherence clearly cares about whether two position eigenstates are contiguous or not.
%The order of eigenstates is important therefore it is unspeakable information. 
\end{itemize}
\color{black}

%Which task one is interested in may determine what sort of coherence is a resource. 
It is critical, however, to distinguish two types of coherence that arise in these various applications.
% that might serve as a resource.
The distinction can be explained as follows.
%To explain the distinction, consider the following pair of states 
Consider the states
\begin{equation}
|\psi\rangle=\frac{|0\rangle+|1\rangle}{\sqrt{2}}\ \ \text{and}\ \  |\phi\rangle=\frac{|0\rangle+|2\rangle}{\sqrt{2}}\ .
\end{equation}
If we are interested in quantum computation using qutrits and the elements of the set $\{| \l\rangle \}_{l \in \{0,1,2\}}$
%$\{ |0\rangle,|1\rangle,|2\rangle\}$
 are the computational basis states, then we would expect the states $|\psi\rangle$ and $|\phi\rangle$ to be equivalent resources because
 %both involve the same type of superposition and 
 the particular identities of the computational basis states appearing in the superposition are not relevant for any computational task.
 %involved are not relevant in a computation. 
  In this case, $l$ is simply an arbitrary label or \emph{flag} for different distinguishable pure states. 
  %It is merely a convenience to label the computational basis states by an integer; any set of distinguishable labels would suffice. 
If, on the other hand, we are considering a phase estimation task and the elements of the set $\{| \l\rangle \}_{l \in \{0,1,2\}}$ are eigenstates of the number operator, then there is a significant difference between the states $|\psi\rangle$ and $|\phi\rangle$.  For instance, $|\psi\rangle$ can detect a phase shift of $\pi$ while $|\phi\rangle$ cannot.  Conversely, if one's task is to estimate a very small phase shift, then $|\phi\rangle$ is a better resource than  $|\psi\rangle$, because the former becomes more orthogonal to itself than the latter under a small phase shift. Similarly, starting from the incoherent state  $|0\rangle$, to prepare a state close to $|\phi\rangle$
%$|\phi\rangle=(|0\rangle+|2\rangle)/{\sqrt{2}}$ 
one needs to have access to a phase reference with higher precision than the phase reference required to prepare a state close to $|\psi\rangle$
%$|\psi\rangle=(|0\rangle+|1\rangle)/{\sqrt{2}}$ 
\cite{Modes}.   
%Here, $l$ is not an arbitrary label.  Rather, its integer value is relevant for the task of phase estimation.   
Here, $l$ is not an arbitrary label, but an eigenvalue of the number operator, and its value is relevant for the task of phase estimation.

These two types of coherence pertain to two types of information which have been termed 
%This difference in types of coherence parallels a difference between types of information, namely, the difference between 
{\em speakable} and {\em unspeakable}~\cite{QRF_BRS_07, peres2002unspeakable}.   Speakable information is information for which the means of encoding is irrelevant.  This is exemplified by the fact that if one seeks to transmit a bit-string, it is irrelevant what degree of freedom one uses to encode the bits.  Unspeakable information is information which can only be encoded in certain degrees of freedom. Information about orientation, for instance, is unspeakable because it can only be transmitted using a system that transforms nontrivially under rotations.  Information about time is also unspeakable because it can only be transmitted by a system that transforms nontrivially under time-translations.  We shall therefore refer to the two types of coherence we have outlined above as speakable and unspeakable respectively.

\begin{comment}
%One might summarize 
This difference in types of coherence concerns 
%is perhaps best understood in terms of 
whether the labels of the superposed states are part of a metric space or not.  We shall therefore refer to coherence between states where the labels need not be part of a metric space as {\em speakable coherence}, while we refer to coherence between states where the labels are part of a metric space as {\em unspeakable coherence}.
\end{comment}

%It turns out that in many of the tasks for which coherence is thought to be a resource, it is {\em unspeakable coherence} that is at play.  Indeed, for the list of tasks we have provided above, in each case except for the last, it is unspeakable coherence that is relevant for the task.  
%For the list of tasks we have provided above, the relevant notion of coherence is the unspeakable one in all cases except for the last.  
For the list of operational tasks we have provided above, the relevant notion of coherence is the unspeakable one in all cases except for the last item. This is also the case for most of the physical applications listed above. 
 This is because from the point of view of speakable coherence, the eigenvalues of the observable that defines the preferred  subspaces are not relevant: the set of preferred subspaces is a set without any order. 
%This is because speakable coherence does not care about the eigenvalues of the observable that defines the preferred  subspaces, and from its point of view preferred subspaces are just a set without any order. 
%This is because from the point of view of speakable coherence, there is no order on the elements of the preferred basis, and therefore the eigenvalues of the observable which defines the preferred basis are completely irrelevant. 
%It follows that in contexts such as 
However, for the examples of quantum speed limits, coherence lengths and magnetic resonance, for instance, it is clear that the eigenvalues of the relevant observable, the Hamiltonian, position, and magnetic moment observables respectively,  has important physical meaning, and there is a natural order defined on the preferred subspaces. Therefore, the notion of unspeakable coherence seems to be the more appropriate one in these cases.
  \color{black}

\color{black}

%\color{red} [Should we note here that the relevant distinction among the two types of tasks is whether they deal with fungible/spekable information or nonfungible/unspeakable information?  Would it make sense to refer to the two notions of coherence as  {\em fungible} and {\em nonfungible} or {\em speakable} and {\em unspeakable} rather than unspeakable and speakable?  Regarding the latter terminology, is it appropriate to use the term ``unspeakable'' even if one is taking about, for instance, the integers modulo $n$ (i.e. when the symmetry group is a discrete cyclic group)?] \color{black}

Most recent work on coherence as a resource, however, considers only speakable coherence.
%a speakable notion of coherence.  
One might think, therefore, that there is work to be done in defining a resource theory of {\em unspeakable} coherence.
%for the unspeakable notion of coherence.   
 In fact, however, such a resource theory {\em already exists}.  It simply goes by another name: the resource theory of {\em asymmetry}.\footnote{In early work on the topic, which emphasized the role of asymmetric quantum states in defining quantum reference frames, the resource theory of asymmetry was termed the resource theory of {\em frameness}~\cite{gour2008resource}.}  To be precise, the resource of unspeakable coherence is nothing more than the resource of 
asymmetry
 %symmetry-breaking, or {\em asymmetry}, 
 relative to a group of translations.

\subsection{Resource-Theoretic approach to unspeakable coherence}

Consider the task of phase estimation as an example.  A state of some field mode has coherence relative to the eigenspaces of the number operator $N$ if and only if it is asymmetric (i.e., symmetry-breaking) relative to the group of phase-shifts generated by $N$, where such translations of the phase are represented by the group of unitaries $\{e^{-i N \theta}:\ \theta\in(0,2\pi]\}$.   This follows from the fact that if a state $\rho$ is symmetric under phase-shifts, that is, $\forall\theta\in(0,2\pi]: e^{-i N \theta} \rho e^{i N \theta}  = \rho$, then it must be block-diagonal with respect to the eigenspaces of $N$, while if it is not symmetric under phase-shifts, then it cannot have this form. 

%In a similar fashion, 
Other examples are treated in a similar fashion\footnote{Throughout this article, we use units where $\hbar=1$.}:
coherence relative to the eigenspaces of a Hamiltonian $H$ is simply asymmetry relative to the group of time-translations generated by this Hamiltonian, $\{e^{- i H t}:\ t\in\mathbb{R}\}$; coherence relative to the eigenspaces of the momentum operator $P$ is simply asymmetry relative to the group of spatial translations, $\{e^{- i P x}:\ x\in\mathbb{R} \}$; coherence relative to the eigenspaces of the angular momentum operator $J_z$ is simply asymmetry relative to the group of rotations around $\hat{z}$, $\{e^{-i J_z \theta}:\ \theta\in(0,2\pi]\}$.
%; and coherence relative to the eigenspaces of the number operator $N$ for some field mode is simply asymmetry relative to the group of phase-shifts for that mode, $\{e^{-i N \theta}:\ \theta\in(0,2\pi]\}$.

 Any resource theory must not only partition the states into those that are resources and those that can be freely prepared at no cost, it must also partition the {\em operations} into those that are resources and those that can be freely implemented at no cost.  The free set of operations is required to be closed under composition and convex combination \cite{coecke2014mathematical}. In entanglement theory, for instance, not only are the states partitioned into those that are unentangled, hence free, and those that are entangled, hence resources, but the operations are also partitioned into those that can be achieved by Local Operations and Classical Communications (LOCC), which are deemed to be free, and those which cannot, which are deemed to be resources.
 %in the entanglement theory, where unentangled states are \emph{free states}, Local Operations and Classical Communications (LOCC) form the set of \emph{free operations}. 
 %These are states and operations which are \emph{easy} or \emph{allowed} to prepare and implement under a practical or fundamental constraint.  

If one considers 
%the restriction on {\em operations} that characterizes 
each of the tasks for which unspeakable coherence is a resource, one sees that the freely-implementable operations are those that are covariant under translations, that is, for which first translating and then implementing the operation is equivalent to first implementing the operation and then  translating (See Def. \ref{defnTI}).
For instance, in the task of reference frame alignment, the set-up of the problem is that there are two parties, each of which has a local reference frame (e.g. a gyroscope, a clock, a phase reference), but the group element that relates these two frames (e.g. the rotation, the time translation, the phase shift) is unknown.   It is not difficult to show that the operations that one party can implement relative to the other party's reference frame are precisely those that are covariant under the group action \cite{QRF_BRS_07, MS11, MS_Short, Marvian_thesis}.  In fact, there are {\em many} ways of providing a physical justification of the translationally-covariant operations, and we shall review these at length further on.

These considerations imply that the problem of quantifying and classifying unspeakable coherence can be considered a special case of the resource theory of asymmetry where the group under consideration describes a translational symmetry.   (The resource theory of asymmetry is more general than this, however, because it is also capable of dealing with non-Abelian groups where asymmetry does not simply correspond to the existence of coherence between some preferred set of subspaces.)

%The notion that the resource of coherence should be understood as asymmetry relative to the action of a translational symmetry and that the free operations defining the resource theory are the translationally-covariant ones was first proposed in Ref.~\cite{GMS09} and developed in Refs.~\cite{Modes, Noether}  and has been adopted also in the study of quantum thermodynamics (see, e.g., \cite{lostaglio2015quantum,yang2015optimal}) and quantum speed limits \cite{marvian2015quantum}.

The notion that the resource of coherence should be understood as asymmetry relative to the action of a translational symmetry and that the free operations defining the resource theory are the translationally-covariant ones was first proposed in Ref.~\cite{Noether} and developed in Appendix A of Ref.~\cite{Modes} and in Ref.~\cite{marvian2015quantum} (See also \cite{yadin2015general}). 
%(1) in *both* the modes paper (appendix A) and the QSL paper, we note that coherence can be understood as asymmetry, 
%(2) In the Noether paper, we note that coherence can be understood as asymmetry and we highlight several applications of unspeakable coherence.} and developed in Refs.~\cite{Modes, Noether}. 

This connection implies that most questions about unspeakable coherence as a resource find their answers in prior work on the resource theory of asymmetry.  It suffices to specialize known results to the particular translational symmetry of interest.  One of the goals of this article is to explicitly translate some of these known results from the language of asymmetry to the language of coherence, to describe the measures of coherence that result, and to review some of the applications. 
%the surprising breadth of applications of the resource theory of translational asymmetry.

This approach to coherence has already been applied to shed light on the various applications of unspeakable coherence outlined above: quantum metrology \cite{Modes,Noether}, aligning reference frames \cite{QRF_BRS_07,MS11}, quantum thermodynamics \cite{lostaglio2015quantum,yang2015optimal,Aberg2014, janzing2000thermodynamic, cwiklinski2015limitations}, quantum speed limits \cite{marvian2015quantum, mondal2016quantum}.
%The method  of mode decompositions has been used in Refs.~\cite{lostaglio2015quantum} and \cite{lostaglio2015description} to study coherence in the context of quantum thermodynamics.  Also note that  \emph{Catalytic Coherence} discussed in \cite{Aberg2} uses the notion of coherence as translational asymmetry.
Furthermore, it was shown in Ref.~\cite{Noether} that for symmetic open-system dynamics, measures of asymmetry are monotonically nonincreasing, thereby yielding a significant generalization of Noether's theorem. 
 Translated into the language of coherence, this result states that for open system dynamics that is translationally-covariant, every measure of coherence which is derived within the translational-covariance approach to coherence provides a monotone of the dynamics.  Such measures, therefore, provide a powerful new tool for studying decoherence.
% For open system dynamics that have a translational symmetry, therefore, measures of coherence which are based on translational asymmetry, also play the role of monotones in this generalized Noether's theorem. [Connect back to decoherence theory.]
% \color{red} 
%[It might be good to go into more detail about these applications after the "measures of coherence" section, but I'm not so interested in writing such a section.] 
%\color{black}

% to tasks that use coherence. 

\subsection{Resource-Theoretic approaches to speakable coherence}

The second topic we address in this work is whether and how one can develop a resource theory of {\em speakable} coherence. 

When one considers recent work on quantifying coherence from the perspective of the 
%unspeakable/speakable
speakable/unspeakable distinction, it is clear that it concerns itself only with the 
speakable notion (See e.g. \cite{Coh_Plenio, winter2015operational, streltsov2015genuine, girolami2015witnessing, ma2015converting, yadin2015quantum, streltsov2015measuring, napoli2016robustness, piani2016robustness}).
%speakable notion. 
%Most recent work on quantifying coherence concerns itself  is best characterized as attempting to develop a resource theory of speakable coherence.  
Most of this work builds on a proposal by Baumgratz, Cramer and Plenio (BCP) \cite{Coh_Plenio}.  
The set of free operations in the BCP approach, called \emph{incoherent operations}, is defined based on the Kraus decomposition of quantum operations, and is closely related to another set which is called \emph{incoherence-perserving}. These are operations which take every incoherent state to an incoherent state.

Because it concerns speakable coherence, this approach is only appropriate for tasks concerning speakable information.
%{\em speakable} notion of coherence, its potential applications are rather limited in scope.  In particular,
% it is not the appropriate approach if one is interested in tasks wherein the relevant notion of coherence is a unspeakable notion. 
%they do not extend to the tasks for which the resource is unspeakable coherence.   
Nonetheless, if one is content to accept that a resource theory of speakable coherence has a more limited scope of applications than one might have na\"{i}vely expected, the question arises of whether the BCP approach is the right way to define the resource theory of speakable coherence.

%One criticism of the BCP proposal is that it places {\em no constraint} on the sorts of measurements that can be implemented, in the following sense: for {\em any} POVM, it is possible to find a measurement that realizes it which is considered free in the BCP proposal.  But just as everyone agrees that the free states in the resource theory of coherence should be restricted to those that have no coherence between the elements of the preferred basis, one would expect that the free {\em effects} in a resource theory of coherence should be similarly restricted, contrary to what occurs in the BCP proposal. 

%The more important criticism, however, concerns the lack of a physical justification for the proposal. 
%the set of incoherent operations as the set of free operations. 

\subsection{Criticism of resource-theoretic approaches to speakable coherence}

As we noted earlier, to take a resource-theoretic approach to any given property one must first of all make a proposal for which set of operations can be freely implemented.  But a given proposal for how to do so is only expected to have physical relevance if it can be provided with a {\em physical justification}, that is, if one can provide a restriction on {\em experimental capabilities} that yields all and only the operations in the free set that is proposed. For instance, in entanglement theory, the restriction on experimental capabilities that yields all and only the LOCC operations between two parties is the absence of any quantum channel between the two parties. 

%Our main criticism of the BCP proposal is that no physical justification has been provided for the incoherent operations being the ones that can be freely implemented.

%Establishing such a physical justification is critical to establishing that a given definition of a resource theory has physical relevance [citation]. 

%The BCP proposal for how to treat coherence resource-theoretically involves taking the free set of operations to be the incoherent operations. 
Despite the amount of attention that the BCP approach has received, no one has yet described a physical justification for the set of incoherent operations or the set of incoherence-preserving operations.  The property of taking incoherent states to incoherent states is certainly a mathematically well-defined constraint; whether there is an {\em experimental} constraint that corresponds to this property is the question of interest here. 
 Of course one can imagine physical scenarios in which preparing coherent states is hard, for instance, because of the challenge of isolating one's systems from environmental decoherence.  But this does not justify the claim that the set of incoherent operations or the set of incoherence-preserving operations is the natural one to study; 
 %to justify the study of a resource theory in which these sets are free, 
to do so one needs to argue that \emph{all} of the operations in a given set can be easily implemented in that physical scenario. 
%, at the presence of decoherence and environment effects. 
However, it is not clear whether such a justification can be found.
%as we discuss in the following, it is not clear whether there is a physical scenario that picks out the incoherent operations or the incoherence-preserving operations.
%The fact that no such justification has yet been provided is certainly a deficiency of the BCP approach. 
%it turns out that each of these sets include operations which do not seem to be implementable in these situations. \color{red}Indeed there are good reasons to %believe that it {\em cannot} be justified by a restriction on experimental capabilities.  [Rob, given that our arguments are not conclusive, I think we  should say something weaker.] 

For one, because the free states in the resource theory of coherence should be restricted to those that have no coherence between the preferred subspaces, one would expect that the free measurements in a resource theory of coherence should be similarly restricted. However, we show that the BCP proposal 
%based on incoherent operations 
places {\em no constraint} on the sorts of measurements that can be implemented, in the following sense: for {\em any} POVM, it is possible to find a measurement that realizes it which is considered free in the BCP approach.  Therefore, to justify the BCP approach to the resource theory of speakable coherence, one needs to argue that there are physical or experimental constraints  which lead to a significant restriction on state-preparations and transformations, but no restriction on the possibilities for discriminating states.

\color{black}
%operations, whereas they do not limit the distinguishability of states at all. 

%, contrary to what occurs in this proposal. 

%Our main criticism of the BCP proposal is that there are good reasons to believe that it {\em cannot} be justified by any such restriction on experimental capabilities. %This is contrary to what one would expect.   
%Second, and more importantly, 
%For another, 
%But there are good reasons to believe that it {\em cannot} be justified by any such restriction on experimental capabilities.

For another, it turns out that even if one finds physical scenarios in which the set of free unitaries is the set of
 incoherent unitaries  (as defined by BCP), this still does not justify the set of incoherent or incoherence-preserving  operations as the set of free operations for a resource theory of speakable coherence.   As we will show, a general incoherence-preserving (incoherent) operation cannot be implemented using only  incoherent states, incoherent unitaries and incoherent measurements. Using a more technical language, this means that  incoherence-preserving (incoherent)  operations do not admit dilation using only incoherent resources (at least, not in a straightforward way, when we treat all systems even-handedly). 
 
 This lack of dilation for the set of free operations is not necessarily a problem in its own right
 \footnote{In entanglement theory, the free operations are the LOCC operations, which do not all admit of a dilation in terms of the free unitaries (because the only unitaries in the LOCC set are tensor products of local unitaries, which do not support any sort of communication, classical or quantum, between the two parties), and yet this is not considered a problem with entanglement theory because there is a natural physical restriction that yields LOCC as the set of free operations, namely, the fact that classical channels between parties are technologically easy to implement, while quantum channels are not. }.
% Strictly speaking, one may argue that even in the entanglement theory the set of LOCC does not have dilation in terms of unitaries in this set: the only unitaries in this set are the tensor products of local unitaries, which do not allow any sort of communication, classical or quantum, between the two parties. the set of LOCC operations admits of a natural physical justification as all and the only operations which can be implemented without having a quantum channel between the two parties.}. 
 One can imagine physical scenarios where the set of operations that we can implement in a \emph{controlled} fashion using the free unitaries and free states, are smaller than the set of all free operations we can implement on the systems of interest, because the latter operations might result from \emph{uncontrolled} interaction of these systems with an environment over which the experimenter has limited control, such as a thermal bath.   The lack of a dilation that is even-handed in its treatment of systems implies that this is the only sort of avenue open for providing a physical justification of the incoherence-preserving or incoherent operations.
% However, lack of dilation implies that finding physical justifications for the study of incoherent unitaries does not automatically justify the set of  incoherence-preserving or incoherent operations. 
 % Indeed, 
Furthermore, we show that there are other natural proposals for the set of free operations for the resource theory of speakable coherence that share precisely the same set of free unitaries, namely, the incoherent unitaries. Therefore, even if one accepts that the set of incoherent unitaries is the appropriate set of free unitaries for a resource theory of speakable coherence, this does not resolve the question of which of the many sets of free operations consistent with this choice one should use to define the resource theory.
 % it is not clear why one should take one of these sets of operations as the set of free operations, and not the other.

%Therefore, to justify the approach based on incoherent (incoherence-preserving) operations  to the resource theory of speakable coherence, one needs to provide some restriction that picks out all and only these operations.  The possibilities for doing so are discussed in more detail in Sec.~\ref{Sec:CriticismIP}.

Our concerns about the suitability of the incoherence-preserving operations in a resource theory of coherence are bolstered by comparing them to the {\em non-entangling} operations in entanglement theory\footnote{One can also define a set of operations that plays a role in entanglement theory which is parallel to the role played by {\em incoherent} operations in the resource theory of coherence, namely, those bipartite operations for which there is a Kraus decomposition where each term is non-entangling.}.  The non-entangling operations are those that map unentangled states to unentangled states \cite{brandao2008entanglement}.   Like the incoherence-preserving operations, therefore, they are the largest set that maps the free states to the free states. 
The non-entangling operations are a strictly larger set than the LOCC operations because they include nonlocal operations such as swapping systems between the two parties, and because they allow the implementation of arbitrary POVM measurements on the bipartite system.  It is difficult to imagine any restriction on experimental capabilities that yields all and only the non-entangling operations.
%, unlike the LOCC operations which arise naturally from the restriction of having no quantum channel.
Indeed, it is widely acknowledged that the LOCC operations---which {\em do} arise from a natural restriction, having classical channels but not quantum channels---is the physically interesting set,  while  the non-entangling operations are studied primarily as a mathematical technique for making inferences about LOCC.  Incoherence-preserving and incoherent operations may ultimately have a similarly subservient role to play in the resource theory of coherence.
\color{black}

\color{black}

\subsection{A proposal for the resource theory of speakable coherence}

In the absence of a physical justification for the BCP approach, the question arises of whether an alternative choice of the set of free operations might be more suited to a resource-theoretic treatment of {\em speakable coherence}.  Once the question is raised, a natural alternative for the set of free operations immediately suggests itself, namely, those  that are {\em covariant under dephasing}, that is, 
%We here propose an alternative choice of the set of free operations for defining the resource theory of speakable coherence. 
%We propose a  resource theory which does not have some drawbacks of the BCP proposal (In particular, it does not allow arbitrary measurements). 
%Specifically, we propose the set of operations that are {\em covariant under dephasing}, that is, 
those that  commute with the operation that achieves complete dephasing relative to the preferred subspaces.
%  This dephasing operation removes all of the off-block-diagonal components of the density operator while leaving the block-diagonal components unchanged.  
We call this the {\em dephasing-covariance approach to coherence}.  A variant of this proposal has recently been considered in Ref.~\cite{yadin2014new}\footnote{Ref.~\cite{yadin2014new} came to our attention as we were preparing this manuscript.}.

This proposal does not have  one of the counter-intuitive features of the BCP approach that was outlined above: the set of free measurements includes only the POVMs whose elements are incoherent, as one would expect.  Nonetheless, it is still not clear whether the dephasing-covariance approach to coherence has much physical relevance because it is still unclear whether there is any restriction on experimental capabilities that picks out all and only the dephasing-covariant operations.  (In particular, it is unclear whether every dephasing-covariant operation admits of a dilation in terms of incoherent states, incoherent measurements and dephasing-covariant unitaries.)  We do not settle the issue here.

\subsection{Relation between different approaches}

In addition to providing a characterization and assessment of both the dephasing-covariance approach and the BCP approach for treating speakable coherence as a resource, we explore the mathematical relation between the free set of operations that each adopt.  In particular, we show that the dephasing-covariant operations relative to a choice of preferred subspaces are a strict subset of the incoherent (incoherence-preserving) operations relative to the same choice.  This implies that any measure of coherence  in the incoherent (incoherence-preserving) approach is also a measure of coherence in the dephasing-covariance approach.

%relation between them.   
%The same relations hold for composite systems, as long as the dephasing operation is taken to act independently on different systems. 

We also compare the translational-covariance approach to coherence with the dephasing-covariance approach (and, via the connection noted above, with the incoherent and incoherence-perserving  approaches).

Any given translational symmetry defines a decomposition of the Hilbert space via the joint eigenspaces of the generators of this symmetry.  Thus, for any given translational symmetry, one can consider 
% Any given translational symmetry and its corresponding generators one can consider the decomposition of the Hilbert space to the eigensubspaces of these generators, and  the
 the sets of dephasing-covariant, incoherent, and incoherence-preserving operations defined relative to this decomposition of the Hilbert space.
% the sets of dephasing-covariant, incoherent, and incoherence-preserving operations which are defined based on the decomposition of the Hilbert space to the eigensubspacses of these generator(s). 
 For instance, the incoherence-preserving operations defined  by a given translational symmetry is the set of operations under which any state which is incoherent with respect to the eigenspaces of its generators, is mapped to a state which is still incoherent relative to these eigenspaces. 

We show that the set of translationally-covariant operations is a strict subset of the dephasing-covariant operations and thus also a strict subset of the incoherent (incoherence-preserving) operations.  This implies that any measure of coherence in the dephasing-covariance proposal is also a measure of coherence in the translational-covariance proposal.  We also show that this inclusion relation is strict. 

Given these inclusion relations, the question arises of whether the measures of coherence that have been identified recently as valid in the BCP approach were already identified in prior work on the resource theory of asymmetry.  We show that this is indeed the case for most such measures of coherence.\footnote{Note that the present work expands on some of the comparisons between the BCP approach to coherence and the one based on translational asymmetry made in Ref.~\cite{marvian2015quantum}.} In addition to noting these relations, we discuss two general techniques for deriving measures of coherence, one that infers them from measures of information and the other that appeals to a certain kind of  decomposition of operators into so-called {\em modes of asymmetry}.

%\Ref.~\cite{Marvian2015} already noted some of the comparisons between 
%the BCP approach to coherence with the one based on translational covariance that we was already made in Ref.~\cite{Marvian2015}.

% ne must use a collective dephasing operation  if the translation group is defined such that different components of a composite system receive independent translations.

%\red{[This paragraph does not belong to this section, because we do not use this concepts in introduction and summary, we should move it to preliminaries ]
%\color{blue} I agree we should move this to preliminaries. \color{black}

\subsection{The choice of preferred subspaces}

The notion of the state of a system being {\em coherent} is only meaningful relative to a choice of decomposition of the Hilbert space of the system into subspaces.  The latter must be dictated by physical considerations, which is to say, {\em operational criteria}.   This is because, from a purely mathematical point of view,  any state is coherent in some basis and incoherent in another basis. If $|\pm\rangle \equiv  ( |0\rangle\pm |1\rangle )/\sqrt{2}$, then the state $|+\rangle$ is coherent if one judges relative to the $\{|0\rangle, |1\rangle\}$ basis, but by the same token, $|0\rangle$ is deemed coherent if one judges relative to the $\{|+\rangle, |-\rangle\}$ basis.  One consequently has no alternative but to appeal to physical considerations in defining the notion of coherence\footnote{The fact that the direct-sum decomposition of Hilbert space that defines the resource of coherence is determined by the physical problem of interest is completely analogous to how the factorization of the Hilbert space of a composite system that defines the resource of entanglement is so determined.}.  
%Any two systems under consideration might be entangled in degrees of freedom over which we do not have control, or of whose existence we may not be aware.  Any such entanglement cannot be used as a resource  to perform tasks, and  cannot be observed experimentally. 

Furthermore, physical considerations often dictate that the relevant notion of coherence is relative to a decomposition of the Hilbert space into subspaces that are not 1-dimensional.  A few examples serve to illustrate this.
%how the preferred subspaces may be higher-dimensional. 

In the context of decoherence theory, environmental decoherence does not always pick out 1-dimensional subspaces of the system Hilbert space.  The dimensions of the decohering subspaces depend on which degree of freedom of the system couples to the environment and generally this is a degenerate observable.  Indeed, this fact, i.e., the existence of \emph{Decoherence Free Subspaces} with dimension larger than one, has been exploited to protect quantum information against decoherence \cite{Lidar:2003fk, Zanardi:97c}. 

Another common example is where the notion of coherence that is of interest is coherence relative to the eigenspaces of some particular physical observable, such as the system's Hamiltonian (as happens when the coherence is a resource for building a quantum clock) or the photon number operator in a particular mode (as happens when the coherence is a resource for phase estimation). Even if the notion of coherence of interest is the speakable one, {\em some} physical degree of freedom must be used to encode the coherence, and practical considerations might dictate the use of a particular physical observable.  And in all such cases, there is a priori no reason that the physical observable should be nondegenerate.  

A final example is if there are degrees of freedom over which the experimenter has no control.  In this case, coherence in that degree of freedom is neither observable nor usable.  

Thus, if the physically relevant observable is $L$, and $l$ labels its eigenspaces while $\alpha$ is a degeneracy index, the state $|l,\alpha_1\rangle+|l,\alpha_2\rangle$ is an {\em incoherent} state in the resource theory insofar as it has no coherence {\em between} the eigenspaces of $L$\footnote{This is analogous to how entanglement {\em between} laboratories can be a resource while entanglement {\em within} a given laboratory is not.}.  Coherence {\em within} an eigenspace of $L$ might be made to be a resource as well, but this requires the degeneracy to be broken, for instance, by introducing another physical observable which picks out a basis of that eigenspace.

We conclude that any resource theory of coherence should be able to quantify and characterize  coherence, not only with respect to 1-dimensional subspaces, but also  with respect to 
%a decomposition of the Hilbert space into 
subspaces of arbitrary dimension.
%any subspace decomposition of the Hilbert space. 
As we will see in the following, the resource theory of unspeakable coherence based on translationally-covariant operations has this capability.  In the case of speakable coherence, we define dephasing-covariant and incoherence-preserving operations to incorporate this possibility,  
%decompositions into subspaces of arbitrary dimension 
and we generalize the definition of incoherent operations in BCP~\cite{Coh_Plenio}, which assumed 1-dimensional subspaces, to do so as well.
% (because the original definition in BCP \cite{Coh_Plenio} assumed 1-dimensional subspaces).\color{black}

%we assume this generality is allowed for incoherent operations, as well as  dephasing-covariant and incoherence-preserving operations,  even though, the original proposal of BCP \cite{Coh_Plenio} for incoherent operations is concerned only with coherence relative to 1-dimensional subspaces. 

%This sort of physical questions should be meaningful and quantifiable in any reasonable resource theory of  coherence. 

%As we will see in the following, the resource theory of unspeakable coherence based on translationally-covariant operations has this capability. In the case of speakable coherence, we assume this generality is allowed for incoherent operations, as well as  dephasing-covariant and incoherence-preserving operations,  even though, the original proposal of BCP \cite{Coh_Plenio} for incoherent operations is concerned only with coherence relative to 1-dimensional subspaces. 

\subsection{Composite systems}
How should the resource theory of coherence be defined on composite systems? In particular, how should we define the set of free states and free operations in this case? For instance, suppose we are interested in quantifying coherence with respect to the energy eigenbasis, which is relevant, for instance, in the context of thermodynamics and clock synchronization.  Consider two non-interacting systems having identical Hamiltonians, $H_A$ and $H_B$, with energy eigenbases $\{|E\rangle_A\}$ and $\{|E\rangle_B\}$ respectively.  One can then imagine  two different ways of defining coherence on the composite system $AB$.  Definition (1): coherence is defined relative to the products of eigenspaces of the single system Hamiltonians. 
%, i.e. $H_A$ and $H_B$. 
In this case, a joint state of systems $A$ and $B$  is coherent if it contains coherence with respect to either of these two Hamiltonians. In particular, in this approach, the state $(|E_1\rangle_A|E_2\rangle_B+|E_2\rangle_A|E_1\rangle_B)/\sqrt{2}$ is considered a resource. 
Definition (2): coherence is defined relative to the eigenspaces of the total Hamiltonian, that is, of $H_A\otimes I_B+I_A\otimes H_B$, where $ I_{A}$ and $I_{B}$ are,  respectively,  the identity operators on systems $A$ and $B$.
%(2) Alternatively, one may define coherence on the joint system in terms of the total Hamiltonian $H_A\otimes I_B+I_A\otimes H_B$, where $ I_{A}$ and $I_{B}$ are,  respectively,  the identity operators on systems $A$ and $B$.
 In this case, states which do not contain coherence relative to the total Hamiltonian, such as $(|E_1\rangle_A|E_2\rangle_B+|E_2\rangle_A|E_1\rangle_B)/\sqrt{2}$, are deemed to be incoherent. 

As the set of preferred subspaces and incoherent states on a single system should be chosen based on physical considerations, the set of incoherent states on composite systems should also be defined in a similar fashion. It turns out that each of the above definitions can be relevant in some physical scenarios. For instance, in the scenarios where the two subsystems cannot exchange energy (for instance, because they are held by two distant parties) then approach (1) is relevant, and entangled states such as $(|E_1\rangle_A|E_2\rangle_B+|E_2\rangle_A|E_1\rangle_B)/\sqrt{2}$ are resources. On the other hand, in the scenarios where we can easily apply operations that allow  energy exchange between the two subsystems, then the relevant observable is the total energy, and not the energy of the individual subsystems. Therefore, in this situation  approach (2) is the relevant one, and  entangled states such as $(|E_1\rangle_A|E_2\rangle_B+|E_2\rangle_A|E_1\rangle_B)/\sqrt{2}$ are not resources.  The fact that the resource of coherence is only defined relative to a choice of basis which depends on the physical scenario is precisely analogous to how the resource of entanglement is only defined relative to a choice of factorization of the Hilbert space which depends on the physical scenario.   (For instance, in the distant laboratories paradigm, entanglement between laboratories is a resource, while entanglement between systems in the same laboratory is not.) 
%Note that this is precisely analogous to entanglement theory, where certain entangled states, such as for two systems in the same laboratory, are considered free. 

\subsection{Outline}

The article is organized as follows.  Sec.~\ref{prelim} covers preliminary material, including a discussion of certain features that are common to the various different proposals for a resource theory of coherence,  what counts as a physical justification of a proposal for the set of free operations, and the definition of a measure of coherence. 
Sec.~\ref{TCcoherence} presents the resource theory of unspeakable coherence that one obtains by taking the free operations to be those that are translationally covariant.  In particular, various different characterizations and physical justifications of the free operations are provided.  Sec.~\ref{DCcoherence} presents the proposal for speakable coherence based on dephasing-covariant operations, together with a discussion of the relation to the translationally-covariant operations and physical justifications.  In Sec.~\ref{BCP}, we review the BCP proposal for speakable coherence, which is defined by the incoherent operations, as well as a related proposal, defined by the set of incoherence-preserving operations.  The relation to the dephasing-covariance appraoch is considered, as well as possibilities for a physical justification.  Finally, in Sec.~\ref{Sec:examples}, we consider measures of coherence within the various approaches, and in Sec.~\ref{Discussion} we provide some concluding remarks.

\section{Preliminaries}\label{prelim}

Any resource theory is specified by a set of \emph{free states} and a set of \emph{free operations}. These are states and operations which are \emph{easy} or \emph{allowed} to prepare and implement under a practical or fundamental constraint.  

%\subsection{Free states and free operations}
\subsection{Free states}

The notion of coherence is only defined relative to a preferred decomposition of the Hilbert space into subspaces. This preferred decomposition is determined based on practical restrictions or physical considerations, although in some cases a preferred decomposition may be considered as a purely mathematical exercise. 
For a system with Hilbert space $\mathcal{H}$, we denote the preferred subspaces by $\{ \mathcal{H}_l \}_l$, so that $\mathcal{H} =\bigoplus_l \mathcal{H}_l$.  Here, the index $l$ may be discrete or continuous.   We denote the projectors onto these subspaces by $\{\Pi_l\}_l$.   

The free states, which are termed \emph{incoherent states}, are those states which are block-diagonal relative to the preferred subspaces, 
\beq\label{incoherent_state}
\rho=\sum_l p_l \Pi_l\ .
\eeq

An alternative way of characterizing the set of free states is via map that dephases between the preferred subspaces.  This dephasing map has the form
%the {\em dephasing map}.  Dephasing relative to the subspaces defined by $\{\Pi_l\}_l$, equivalently, relative to the eigenspaces of $L$, is achieved by the map
\begin{equation}\label{Dephasing}
\mathcal{D}(\cdot) \equiv \sum_{l} \Pi_l (\cdot) \Pi_l \ .
\end{equation}
As a superoperator acting on the vector space of  operators, $\mathcal{D}$ is a projector, and hence idempotent, $\mathcal{D}^2=\mathcal{D}$.  In fact, it projects onto the subspace of operators that are block-diagonal relative to the decomposition $\{ \mathcal{H}_l \}_l$, so that the set of incoherent states can be characterized as those that are invariant under $\mathcal{D}$, 
\begin{equation}
%\rho \in \mathcal{I} \; \leftrightarrow \; 
\mathcal{D}(\rho)=\rho\ .
\end{equation}

Note that for any choice of preferred subspaces, the set of incoherent states is closed under convex combinations. We will denote this set by $\mathcal{I}$.
%The subspace spanned by  incoherent states is denoted by $\mathcal{I}$.

%In many cases of interest, the set of preferred subspaces is specified as the set of eigenspaces of an observable. If $L$ denotes a Hermitian operator whose eigenspaces are precisely the preferred subspaces, that is, $L= \sum_l  \lambda_l \Pi_l$, where the $\lambda_l$ are all distinct, then we can equivalently characterize the incoherent states as those which commute with $L$,
%\beq\label{incoherent_state2}
%[\rho,L]=0\ .
%\eeq

\subsection{Free measurements}

If a system survives a quantum measurement, then the outcome of the measurement provides the ability to predict the outcomes of future measurements on the system.  To do so, one must specify the state update map associated to the measurement.  The von Neumann projection postulate is an example. This is a specification of the measurement's {\em predictive} aspect.  Whether the system survives or not, every quantum measurement also allows one to make retrodictions about earlier interventions of the system.  
%Measurements provide information both for prediction  in quantum theory have a predictive and a retrodictive aspect.   
We here focus only on the retrodictive aspect of a measurement, which in quantum theory is represented by a postive-operator-valued measure (POVM).   An element of a POVM, denoted $E$, satisfies $I \ge E \ge 0$, and wlll be termed  an {\em effect}.

We will call an effect {\em incoherent} if it is block-diagonal relative to the preferred subspaces.  Although one might have expected all proposals for the resource theory of coherence to allow as free only those POVMs made up of incoherent effects, we will see that this is not the case for the proposal based on incoherence-preserving  or incoherent \color{black}    operations.
%, and we will argue that this fact is a deficiency of the proposal. 

\subsection{Free operations}\label{constraintsonfreeops}

We turn now to the set of free operations. 
We consider not only those operations wherein the input and output spaces are the same (i.e., transformations of a system) but also those where they may be different (in which case the operation involves adding or taking away some or all of the system).  In particular, the free operations from a trivial input space to a nontrivial output space specify which \emph{preparations} of the output system can be freely implemented, so that a specification of the free operations implies a specification of the free states.

The minimal property that the set of free operations should have is to be {\em incoherence-preserving}, which is to say that each free operation takes every incoherent state on the input space to an incoherent state on the output space.
%The minimal property that the set of free operations should have is
%\begin{itemize}
%\item\label{property2}{Incoherence-preserving:} Every free operation takes incoherent states on its input space to incoherent states on its output space.
%\end{itemize}
Note that the incoherence-preserving property implies that for the special case where the operation is a state preparation, the free set corresponds to the incoherent states. \footnote{This is because we can think of the input in a state preparation as a 1-dimensional system, which is necessarily in an incoherent state.}

%Note that because the input and output spaces may be different, specifying the free operations from a trivial input space to a nontrivial output space specifies which \emph{preparations} of the output system can be freely implemented, that is, which states are free for that output space.  

All proposals we consider here are such that every free operation is incoherence-preserving. 
%have these minimal properties.  
Nonetheless, different proposals %We will see, however, that different proposals 
for how to treat coherence as a resource differ in their choice of the set of free operations, subject to this constraint. 
%assume different sets of free operations.

Note that we here use the term \emph{quantum operation} to refer to a trace-nonincreasing completely positive linear map.  If the operation is trace-preserving, we will refer to it as a \emph{quantum channel}.
%will refer to a trace-nonincreasing completely positive linear map as a \emph{quantum operation}, while one that is trace-preserving will be called a \emph{quantum channel}.

%{\bf Consistency condition of Stinespring dilation as a criterion of operational meaningfulness}

%\subsection{Dilation of the free operations}
%\subsection{Probing physical justification via dilation} 
\subsection{Physical justification of the free operations through dilation} \label{dilation}

%We now formulate a desideratum for the set of free operations that is {\em not} shared by all three proposals.
% to be {\em physically justified}.

It is widely believed that physical systems undergoing closed-system dynamics evolve according to a unitary map. 
%fundamentall evolution of physical systems are unitary evolutions. 
In this view, the only circumstance in which a nonunitary map is used to describe the evolution of a system's state is when the system is known to undergo open-system dynamics, that is, when it interacts with some auxiliary system (perhaps its environment) via a unitary map, but one chooses to not describe the auxiliary system, by marginalizing or post-selecting on it. 
It is straigthforward to show that in any situation wherein a system interacts unitarily with the auxiliary system and one subsequently marginalizes or post-selects on the latter (through a partial trace or a partial trace with some measurement effect respectively), the effective evolution of the system's state is always described by a completely positive trace-nonincreasing map (what we are here calling a quantum operation).   %One expects, therefore, that for {\em open-system} dynamics, i.e., wherein the system interacts with an environment, {\em nonunitary} evolution is possible but it arises from a unitary interaction of that system with its environment, followed by a partial trace operation. 
The Stinespring dilation theorem \cite{nielsen2000quantum} guarantees that the vice-versa is also true: {\em every} quantum operation on the system can be achieved in this fashion. 

For a given triple of state of the auxiliary system, effect measured thereon, and unitary coupling of the system to its auxiliary, we will term the effective quantum operation on the system that it defines the {\em marginal operation} on the system.  For a given quantum operation on the system, we will term any triple of state of the auxiliary, effect on the auxiliary and unitary coupling of system to auxiliary that yields the operation as a marginal, a {\em dilation} of that operation.

% {\em all} nonunitary operations on a system can be implemented in this fashion. 
%supports this expectation: every nonunitary operation on a system can be achieved by unitarily coupling that system to an environment and tracing over the environment. 
%So, in the physical world any non-unitary quantum operation can be realized only using unitary evolutions. In general, it is guaranteed by the Stinespring dilation theorem that  any quantum operation can be realized by coupling the system to an ancilla and applying a unitary on the joint system. 

%A resource theory is defined by specifying a set of operations that are deemed to be freely implementable.  Included among these are states and unitaries.  
%Included among these are states and unitaries.  
%A natural consistency condition, therefore, is to 

%For a resource theory, it is similarly natural to demand that any nonunitary operation that is included in the set of free operations should be achievable using free states, free effects, and free unitaries.  That is, it should be achievable by coupling the system to an environment via a unitary operation that is free, using an initial state of the environment that is free.  If this condition is not met, then it implies that a nonunitary operation that is supposed to be free in fact must consume some resource, either in the state of the environment or in the unitary coupling.

In the context of resource theories, one way to define the free set of operations on a system is by specifying the free states and free effects on the auxiliary system, as well as the free unitaries that couple the system to the auxiliary, and then defining the free set of operations on the system as all and only those that can be obtained as a marginal of these.   If a proposal for the free set of operations is {\em not} defined in this way, then one {\em can and should} ask whether it admits of such a definition or not.  In other words, one should seek to determine whether the free set of operations in a given proposal can be understood as those that admit of a dilation in terms of the free states and effects on, and free unitary couplings with, the auxilary system. We refer to such dilations as {\em free dilations}.

% For instance, in entanglement theory any bi-partite LOCC operation has free dilation, in the sense that it can be implemented by  local unitaries on the systems of interests and local auxiliary systems, which do not couple systems at different local points (labs) to each other.

We shall here ask this question of various proposals for how to choose the free set of operations in a resource theory of coherence.  If, for any given proposal, one finds that the free set of operations on the system of interest includes an operation that does {\em not} admit of a free dilation, then this may imply that some nontrivial resource on the composite of system and auxiliary must be consumed in order to realize the operation.

We will show that in cases where one considers a translation group that acts collectively on all physical systems, the translationally-covariant operations have free dilation.  We will also show that the set of incoherence-preserving operations and the set of incoherent operations does not have this property, at least not if we treat all systems even-handedly.  For the set of dephasing covariant operations, the question remains open.

\color{black}

\subsection{Measures of coherence}

In any resource theory, a measure of the resource is a function from states to real numbers which defines a partial order on the set of states.
%, and intuitively quantify the amount of resource in a given state. 
%There are different properties which can be considered for a measure of resource. However, the most essential one
The essential property that any such function must have is to be a monotone (i.e, to be monotonically nonincreasing)  under the free operations\footnote{See Ref.~\cite{gour2008resource} for an operational justification of monotonicity, that is, an account of why monotonicity is required if a measure of a resource is to characterize the degree of success achievable in some operational task.}.  We are therefore going to use the following definition of a measure of coherence:
\begin{definition}
A function $f$ from states to real numbers is a measure of coherence according to a given proposal if 
%\textbf{A}, \textbf{B} or \textbf{C} if it satisfies the followings \\
\noindent (i) For  any trace-preserving quantum operation $\mathcal{E}$ which is free according to the proposal,
%in the corresponding resource theory,
 it holds that $f(\mathcal{E}(\cdot))\le f(\cdot)$.\\
\noindent (ii) For any incoherent state $\rho\in\mathcal{I}$, it holds that  $f(\rho)=0$.
\end{definition}
Because any incoherent state can be mapped to any other incoherent state via a free trace-preserving operation, condition (i) implies that the value of the function $f$ must be the same for all incoherent states. Condition (ii) merely expresses a choice of convention for this value: that all incoherent states should be assigned measure zero.  Of course, given any function satisfying condition (i), one can define a shift of this function which satisfies condition (ii).

\begin{comment}
We now discuss the operational motivation for monotonicity.   

The motivation for demanding that a measure of the resource be monotonically non-increasing under the allowed operations is that it is a necessary condition if the measure is to have operational significance. This is an important point that is worth making precise.\newline
We shall say that a measure of a resource is operational if and only if it quantifies the optimal figure of merit for some task that requires the resource for its implementation. Specifically, we imagine a task that is described entirely operationally (i.e. in terms of empirically observable consequences) and a figure of merit that quantifies the degree of success achieved by every possible protocol for implementing the task (under the restriction that defines the resource theory). Success might be measured in terms of the probability of achieving some outcome, or the yield of some other resource, etc. The key point is that any processing of the resource (consistent with the restriction that defines the resource theory) cannot increase an operational measure of that resource because the definition of an operational measure already incorporates an optimization over protocols and thus an optimization over all such processings.

\end{comment}

It is worth noting that any measure of a resource $f$ is constant on states that are connected by a free unitary operation.  That is, if $\mathcal{U}$ is a free unitary operation, then any resource measure $f$ must satisfy  
$$f(\rho) = f(\mathcal{U} (\rho)).$$
The proof is simply that if $\rho$ and $\sigma$ are connected by a free unitary, then state conversion in both directions are possible under the free operations,  $\rho \to \sigma$ and  $\sigma \to \rho$, which in turn implies that $f(\rho) \ge f(\sigma)$ and $f(\rho) \le f(\sigma)$.
%the fact that $\rho \to \sigma$ implies $f(\rho) \ge f(\sigma)$, while the fact that $\sigma \to \rho$ implies $f(\sigma)\ge f(\rho)$.

%[Placeholder: Provide an equivalent definition of measures of coherence in terms of incoherent states, free unitaries and monotonicity under partial trace.  Note that the quantum speed of evolution satisfies this property.]

We distinguish the three resource-theoretic approaches to coherence that we consider in this article 
by the set of free operations that define them: translationally-covariant, dephasing-covariant and incoherence-preserving operations.  A measure of coherence within a given approach is also defined relative to the set of free operations within that approach.  Therefore, we refer to measures of coherence within the different approaches as measures of {\em TC-coherence}, {\em DC-coherence}, and {\em IP-coherence} respectively. 
In Sec. \ref{Sec:examples},  we provide a list of examples for each type.

\section{Coherence via translationally-covariant operations}\label{TCcoherence}

We begin by demonstrating that if one is interested in an {\em unspeakable} notion of coherence, then coherence can be understood as asymmetry relative to a symmetry group of translations.   In this approach, the coherence is defined based on a given observable $L$, such as the Hamiltonian, the linear momentum,  or the angular momentum. Then, to characterize coherence relative to the eigenbasis of $L$, we consider the asymmetry relative to the group of translations generated  by $L$, i.e.,  the group of unitaries 
\begin{equation}
U_{L, x} \equiv e^{-i x L} : x\in \mathbb{R}\ .
\end{equation}
The superoperator representation of the translation $x \in \mathbb{R}$ is then
%Define the superoperator representation of the translation $x \in \mathbb{R}$ generated by the observable $L$,
%\begin{equation}
%\mathcal{U}_{L, x}(\cdot) \equiv e^{-i x L}(\cdot)  e^{i x L}.
%\end{equation}
\begin{eqnarray}
\mathcal{U}_{L, x}(\cdot) &\equiv& U_{L, x}(\cdot)  U^{\dag}_{L, x}\nonumber\\
&=& e^{-i x L}(\cdot)  e^{i x L}.
\end{eqnarray}
Note that this group has often a natural physical interpretation. For instance, if $L$ is the Hamiltonian, then it generates the group of time translations, and if  $L$ is the component of angular momentum in some direction, then it generates the group of rotations about this direction.\footnote{Note that, depending on the spectrum of $L$, this group could be isomorphic to U(1) or to $\mathbb{R}$  (under addition).}

%Suppose the spectral decomposition of the generator $L$ is
%\beq
%L=\sum_l \lambda_l \Pi_l
%\eeq 
%where the $\lambda_l$ are all distinct, so that the $\{\Pi_l\}_l$ are the projectors onto the eigenspaces of $L$.

%unspeakable coherence is a special case of asymmetry, for the case of translational symmetries, i.e. the group of unitaries $\{ e^{-i x L}: x\in\mathbb{R}\}$,  where the observable $L$ is the generator of the symmetry with eigenvalues $l$. In this approach,  $\{P_l\}_l$ are projectors onto the eigenspaces of an observable  $L=\sum_l \lambda_l P_l$ where the $\lambda_l$ are nondegenerate.

%According to the  resource theory of asymmetry,
The free states are taken to be those that are  \emph{translationally-invariant} or {\em translationally-symmetric},
 %the free states are those that are invariant under all translations, termed the \emph{translationally-invariant} states,
  %or \emph{translationally-symmetric} states, 
%\begin{align}
%\forall x\in \mathbb{R} :\  e^{- i L x}\rho\ e^{i L x}=\rho.
%\end{align}
%Defining the superoperator representation of the group as ???, the translationally-invariant states can equivalently be characterized by
%or equivalently,
\begin{align}\label{fourthdefnincoherentstates}
\forall x\in \mathbb{R} :\  \mathcal{U}_{L,x} (\rho)  &=\rho.
\end{align}
   %$\{ e^{-i x L}: x\in\mathbb{R}\}$. 
One can easily see that the set of translationally-invariant states coincides with the set of states that are incoherent with respect to the eigenspaces of $L$, % $\{P_l\}_l$, 
 i.e., 
\begin{equation}\label{equiv}
\forall x\in \mathbb{R} : \mathcal{U}_{L,x}(\rho)=\rho\ \ \Longleftrightarrow \ \  [L,\rho]=0.
%\forall x\in \mathbb{R} : e^{- i L x}\rho\ e^{i L x}=\rho\ \ \Longleftrightarrow \ \  [L,\rho]=0.
%\ \ \Longleftrightarrow\  \rho\in\mathcal{I} .
\end{equation}
Therefore,  in the translational-covariance approach to coherence the preferred subspaces relative to which coherence is a resource are the eigenspaces of the generator $L$. 
%one wishes to define coherence are those corresponding to the projectors $\{P_l\}_l$, then a state is incoherent if and only if it is translationally-invariant.
%Therefore, if the preferred subspaces relative to which one wishes to define coherence are those corresponding to the projectors $\{P_l\}_l$, then a state is incoherent if and only if it is translationally-invariant.

%In this section, we consider defining the resource of coherence as nothing other than the resource of translational-{\em non}invariance or translational asymmetry.   We must therefore review the way in which the resource theory of translational asymmetry is defined. 
%, namely, the translationally-covariant operations.

\color{black}

%\subsection{Definition of free operations: Translationally-Invariant operations}
\subsection{Free operations as translationally-covariant operations}

For a given choice of symmetry transformations, the resource theory of asymmetry is defined by taking the set of free operations to be those that are covariant relative to the symmetry transformations.   We here particularize this definition to the case of a translational symmetry, and provide several ways of characterizing this set.

%\color{red} [In the folowing, would it be better to move all the proofs to the appendices?  ] \color{black}

%\subsubsection{Characterization via covariance properties}
\subsubsection{Definition of translationally-covariant operations}

 \begin{definition}\label{defnTI}
 We say that a quantum operation  $\mathcal{E}$ is {\em translationally-covariant relative to the translational symmetry generated by $L$} if
 %free if it is covariant with respect to a translation of $x$ generated by the observable $L$, for all $x\in \mathbb{R}$,
\begin{align}
\forall x\in \mathbb{R} :\  \mathcal{U}_{L,x} \circ \mathcal{E}  &=\mathcal{E} \circ  \mathcal{U}_{L,x}.\label{TI_opo}
\end{align}
\end{definition}

Note that condition \eqref{TI_opo} is equivalent to 
\begin{align}
\forall x\in \mathbb{R} :\  \mathcal{U}_{L,x} \circ \mathcal{E}  \circ  \mathcal{U}^{\dag}_{L,x} &=\mathcal{E}.\label{TI_opo2}
\end{align}
If the input and output spaces of the map $\mathcal{E}$ are distinct, then  the generator $L$ may be different on the input and output spaces. For instance, in the case where  $L$  corresponds to the  angular momentum  in a certain direction, then this observable may have different representations on the input and output spaces.  For simplicity, we do not indicate such differences in our notation.
\color{black}
A preparation of the state $\rho$ is an operation with a trivial input space and translational-covariance in this case implies translational-invariance of $\rho$, confirming that translationally-invariant states are the free states in this approach.

%Both desiderata for the free set outlined in Sec.~\ref{constraintsonfreeops} are satisfied by translationally-covariant quantum operations.  The closure property is easily verified.
%Translationally-covariant quantum operations are easily verified to have property \ref{property1} of Sec.~\ref{constraintsonfreeops}.  
%To see that properties \ref{property2} and \ref{property3} also hold, it is useful to  note that a state is incoherent relative to the eigenspaces of $L$ if and only if it is invariant under all translations generated by $L$, that is, the condition $[\rho,L]=0$ is equivalent to
%\begin{align}\label{fourthdefnincoherentstates}
%\forall x\in \mathbb{R} :\  \mathcal{U}_{L,x} (\rho)  &=\rho.
%\end{align}
%To see that property \ref{property2} holds, 
%To see that the incoherence-preserving property holds,
Sec.~\ref{constraintsonfreeops} articulated a minimal constraint on the free operations, that they should be incoherence-preserving. Translationally-covariant operations have this property because 
%are easily seen to have the property of being incoherence-preserving, as the desideratum of  note that 
 from Eq.~\eqref{fourthdefnincoherentstates} one can deduce that for a translationally-covariant operation $\mathcal{E}$, and an incoherent state $\rho$ for input,
\begin{align}
\forall x\in \mathbb{R} :\  \mathcal{U}_{L,x} (\mathcal{E}(\rho))  &=\mathcal{E} (\mathcal{U}_{L,x}(\rho)) = \mathcal{E}(\rho),
\end{align}
which implies that $\mathcal{E}(\rho)$ is translationally-invariant, hence incoherent. 
%satisfies Eq.~\eqref{fourthdefnincoherentstates}. 
 Therefore incoherent states are mapped to incoherent states.
%, where the set of incoherent states relative to the eigenspaces of $L$ are those that commute with $L$, implying that if 
%$[\rho,L]=0$ then $[\mathcal{E}(\rho),L]=0$, so that incoherent states are mapped to incoherent states.

If one thinks of incoherence as translational symmetry, then the incoherence-preserving property formalizes the simple intuition known as {\em Curie's principle}: If the initial state does not break the translational symmetry and the evolution does not break the translational symmetry either, then the final state cannot break the translational symmetry.

%Finally, we show that property \ref{property3} also holds. If the input space is trivial, so that the map $\mathcal{E}$ corresponds to the preparation of a state $\rho$ on the output space, then Eq.~\eqref{TI_opo} reduces to Eq.~\eqref{fourthdefnincoherentstates} which, as just noted, is equivalent to $[\rho,L]=0$ and therefore Eq.~\eqref{TI_opo} specializes to the definition of an incoherent state on the output space.

\subsubsection{Translationally-covariant measurements}

If an operation $\mathcal{E}$ has a trivial {\em output} space, so that it corresponds to tracing with a measurement effect $E$ on the input space, that is,  $\mathcal{E}(\cdot) = {\rm Tr}(E \cdot)$, then Eq.~\eqref{TI_opo} reduces to 
\begin{align}
\forall x\in \mathbb{R} :\ {\rm Tr}(E (\cdot)) &= {\rm Tr}(E \; \mathcal{U}_{L	,x}(\cdot)) \nonumber\\
&= {\rm Tr}(\mathcal{U}^{\dag}_{L,x}(E) \;(\cdot)) ,
\end{align}
which in turn implies 
\begin{align}
\forall x\in \mathbb{R} :\  \mathcal{U}^{\dag}_{L,x} (E)  &=E,
\end{align}
i.e., the effect $E$ is translationally-invariant. This condition is equivalent to $[E,L]=0$, so that the effect $E$ is incoherent with respect to the eigenspaces of $L$.

\begin{proposition}\label{propTC:meas}
%The set of POVMs associated to translationally-covariant measurements  are all and only those whose elements are all block-diagonal relative to the eigenspaces of $L$.
A POVM is translationally-covariant if and only if all of its effects are incoherent 
%elements are block-diagonal 
relative to the eigenspaces of $L$.
\end{proposition}

%\color{red} [I was going to talk about translationally-covariant unitary maps as well, but they do not always correspond to those for which the unitary operator commutes with the generator, so I think it might be best not to mention it here.]

\subsubsection{Translationally-covariant unitary operations}

Finally, if $\mathcal{V}$ is a unitary translationally-covariant operation, that is, $\mathcal{V}(\cdot) = V(\cdot)V^{\dag}$ for some unitary operator $V$ (in which case the input and output spaces are necessarily the same), then Eq.~\eqref{TI_opo} reduces to 
\begin{align}
\forall x\in \mathbb{R} : U_{L,x} V (\cdot) V^{\dag} U^{\dag}_{L,x}  &=  V U_{L,x}(\cdot) U^{\dag}_{L,x} V^{\dag},
%\nonumber\\
%&= {\rm Tr}(\mathcal{U}^{\dag}_{L,x}(E) \;(\cdot)) ,
\end{align}
which implies 
\begin{align}
\forall x\in \mathbb{R} :\  U_{L,x} V = e^{i\phi(x)} V U_{L,x},
\end{align}
for some phase $\phi(x)$. Taking the traces of both sides, we find that in finite-dimensional Hilbert spaces, this condition can hold if and only if $e^{i\phi(x)}=1$,
 %[Was this the point that Giulio mentioned before? Should we we cite him?], 
that is, if and only if
\begin{align}
\forall x\in \mathbb{R} :\  [U_{L,x} ,V]=0,
\end{align}
which is equivalent to $[V,L]=0$, so that the unitary operator $V$ is also block-diagonal with respect to the eigenspaces of $L$.  If $\{ \mathcal{H}_{\lambda} \}_{\lambda}$ denotes the set of eigenspaces of $L$, $\{ \Pi_{\lambda} \}_{\lambda}$ the projectors onto these, and $\{ V_{\lambda} \}_{\lambda}$ an arbitrary set of unitaries {\em within} each such subspace, then any such unitary $V$ can be written as
\beq\label{unit352}
V=  \sum_{\lambda} V_{\lambda} \Pi_{\lambda}\ .
\eeq

%If the spectrum of $L$ has infinite cardinality, then the system must have an infinite-dimensional Hilbert space, and 
If the Hilbert space is infinite-dimensional, on the other hand, then the characterization above need not apply.  Indeed, in this case, there are translationally-covariant unitaries that need not map every eigenspace of $L$ to itself.  For instance, if the generator is a charge operator with integer eigenvalues, $Q= \sum_{q = -\infty}^{\infty} q |q\rangle \langle q|$, then the unitary $V_{\Delta q}$ that applies a rigid shift of the charge by an integer $\Delta q$, that is, 
$$V_{\Delta q} \equiv \sum_{q = -\infty}^{\infty} |q + \Delta q \rangle \langle q|,$$
defines a unitary operation $\mathcal{V}_{\Delta q}(\cdot) = V_{\Delta q}(\cdot)V_{\Delta q}^{\dag}$ that is covariant relative to the group of shifts of the phase conjugate to charge, $\{ \mathcal{U}_{Q,x}: x\in \mathbb{R} \}$ where $\mathcal{U}_{Q,x}(\cdot) = e^{-ixQ}(\cdot)e^{ixQ}$.
%corresponds to a translationally-covariant unitary operation.
As another example, if the system is a particle in one dimension, then the unitary operation that boosts the momentum by $\Delta p$ is translationally-covariant relative to the group of spatial translations.  This is because the unitary operation associated with a boost by $\Delta p$, $e^{-i\Delta p X} (\cdot) e^{i\Delta p X}$ where $X$ is the position operator, 
and the unitary operation associated with a translation by $\Delta x$, $e^{-i\Delta x P} (\cdot) e^{i\Delta x P}$ where $P$ is the momentum operator, commute with one another for all $\Delta x, \Delta p \in \mathbb{R}$.

%\color{red} [Make the definition more general, to allow for distinct input and output spaces, then show that it particularizes to the translationally invariant states, unitaries and POVMs that we define in the beginning.] \color{black}

\subsubsection{Characterization via Stinespring dilation}\label{TCStinespring}

We show that every translationally-covariant operation on a system can arise by coupling the system to 
an ancilla in an incoherent (translationally-invariant) state, subjecting the composite to a translationally-covariant unitary, and post-selecting on the outcome of a measurement on the ancilla which is assoicated to an incoherent (translationally-invariant) effect.  Such an implementation is termed a translationally-covariant {\em dilation} of the operation.

%We show that every translationally-covariant channel on a system can arise by coupling the system to an ancilla in a translationally-invariant state, subjecting the composite to a translationally-covariant unitary, and taking a partial trace.  
%Such an implementation is termed a translationally-covariant {\em dilation} of the channel. \color{black}  For a translationally-covariant {\em operation} (which, unlike a channel, need not be trace-preserving), a translationally-covariant dilation is similar but allows for a measurement followed by a post-selection of one outcome in place of the partial trace. \color{black}

To make sense of the notion of a translationally-covariant dilation, however, one needs to specify not only the representation of the translation group on the system and ancilla individually, but on the composite of system and ancilla as well. 
%To discuss a dilation of a nonunitary operation on some system to a unitary operation on the composite of that system and  an auxiliary system, we must first address the question of how the translational symmetry is represented on the composite.  
Recall that we allow the operation $\mathcal{E}$ to have different input and output spaces, so that to make sense of a translationally-covariant dilation, we must also specify the representation of the translation group on the output versions of the system and ancilla.  

Some notation is helful here. We denote the Hilbert spaces corresponding to the input and output of the map $\mathcal{E}$ by $\mathcal{H}_{\rm s}$ and $\mathcal{H}_{\rm s'}$ respectively.  Denoting the Hilbert space of the ancilla by $\mathcal{H}_{\rm a}$, the composite Hilbert space of system and ancilla is $\mathcal{H}_{\rm s} \otimes \mathcal{H}_{\rm a}$.  We denote the subsystem that is complementary to $\mathcal{H}_{\rm s'}$ by $\mathcal{H}_{\rm a'}$ (this is the subsystem over which one traces), so that $\mathcal{H}_{\rm s} \otimes \mathcal{H}_{\rm a} = \mathcal{H}_{\rm s'} \otimes \mathcal{H}_{\rm a'}$.

In the physical situations to which TC-coherence applies---which we will discuss at length in Sec.~\ref{JustificationsTC}---one can always choose an ancilla system such that translation is represented {\em collectively} on the composite of system and ancilla.  Specifically, if $L_{\rm s}$ is the generator of translations on $\mathcal{H}_s$ and $L_{\rm a}$ is the generator of translations on $\mathcal{H}_a$, then the generator of translations on the composite $\mathcal{H}_{\rm s} \otimes \mathcal{H}_{\rm a}$ is $L = L_{\rm s} \otimes I_{\rm a} +  I_{\rm s} \otimes L_{\rm a} $. Similarly, we have $L = L_{\rm s'} \otimes I_{\rm a'} +  I_{\rm s'} \otimes L_{\rm a'}$.  It follows that the translation operation on the composite is collective on the factorization $\mathcal{H}_{\rm s} \otimes \mathcal{H}_{\rm a}$, that is, $\mathcal{U}_{L,x} = \mathcal{U}_{L_s,x} \otimes \mathcal{U}_{L_a,x}$ and on the factorization $\mathcal{H}_{\rm s'} \otimes \mathcal{H}_{\rm a'}$, that is, $\mathcal{U}_{L,x}= \mathcal{U}_{L_{s'},x} \otimes \mathcal{U}_{L_{a'},x}$.  In the discussion below, we use $L$ to denote the generator of translations, regardless of the system it is acting upon.

%\begin{definition}
\begin{proposition}\label{propTCStinespring}
A quantum operation  $\mathcal{E}$ is {\em translationally-covariant} 
%relative to the translational symmetry generated by $L$
 if and only if it can be implemented by coupling the system $\mathcal{H}_{\rm s}$ to an ancilla $\mathcal{H}_{\rm a}$ prepared in an incoherent state $\sigma$ via a translationally-covariant unitary quantum operation $\mathcal{V}$, and then post-selecting on the outcome of a measurement on the ancilla $\mathcal{H}_{\rm a'}$ which is associated with an incoherent effect $E$.
 %to an taking a partial trace with an incoherent effect $E$ on an ancilla $\mathcal{H}_{\rm a'}$. 
 %tracing over a subsystem (denoted $\mathcal{H}_{\rm a'}$ and not necessarily equal to $\mathcal{H}_{\rm a}$).
   Formally, the condition is that for all quantum states $\rho$,
\begin{equation}\label{TIStinespring}
\mathcal{E}(\rho) = {\rm tr}_{\rm a'} (E \mathcal{V}( \rho \otimes \sigma)),
\end{equation}
%where $\mathcal{D}_{L_{\rm a}}(\sigma)= \sigma$ and $\mathcal{D}_{L_{\rm a'}}(E)= E$ 
where $[L_{\rm a},\sigma]=0$ and $[L_{\rm a'}, E]=0$ 
and where 
%$\mathcal{V}(\cdot) =V (\cdot) V^{\dag}$ for some unitary operator $V$ where 
$\mathcal{V} \circ \mathcal{U}_{L,x} = \mathcal{U}_{L,x} \circ \mathcal{V}$
%$[\mathcal{V},\mathcal{U}_{L,x}]=0$
 for all $x\in \mathbb{R}$.
%$[V,L]=0$.
%\end{definition}
\end{proposition}

\color{black}
\begin{proof}
The proof that any operation of the form of Eq.~\eqref{TIStinespring} is translationally-covariant is as follows:
\bes
\begin{align}
&\mathcal{E}\big(\mathcal{U}_{L_{\rm s},x}(\rho)\big)\\
&=\text{tr}_{\rm a'} \Big(E \mathcal{V} \Big(\mathcal{U}_{L_{\rm s},x}(\rho)\otimes\sigma\Big) \Big)\ \\ 
&=\text{tr}_{\rm a'}\Big(E\; \mathcal{V}  \Big( \mathcal{U}_{L_{\rm s},x}(\rho)\otimes \mathcal{U}_{L_{\rm a},x}(\sigma)\Big) \Big)\\ 
%&=\text{tr}_{\rm a'}\Big(E\; \mathcal{V}  \Big(\mathcal{U}_{L,x}(\rho\otimes \sigma)\Big) \Big)\\ 
%&=\text{tr}_{\rm a'} \Big(E\; \mathcal{U}_{L,x} \Big( \mathcal{V} [\rho\otimes\sigma] \Big)\Big)\\ 
&=\text{tr}_{\rm a'} \Big(E\; \mathcal{U}_{L_{\rm s'},x}\otimes \mathcal{U}_{L_{\rm a'},x} \Big( \mathcal{V} [\rho\otimes\sigma] \Big)\Big)\\ 
&= \mathcal{U}_{L_{\rm s'},x} \Big( \text{tr}_{\rm a'} \Big(\mathcal{U}^{\dag}_{L_{\rm a'},x} (E) \; \Big( \mathcal{V} [\rho\otimes\sigma] \Big)\Big)\Big)\\ 
&=\mathcal{U}_{L_{\rm s'},x} \Big( \text{tr}_{\rm a'} \big(E\;  \mathcal{V} [\rho \otimes\sigma] \big)\Big)\\ &=\mathcal{U}_{L_{\rm s'},x}(\mathcal{E}(\rho)),
\end{align}
\ees
where in the second equality, we have used that fact that $\sigma$ is an incoherent  state;  in the third equality, we have used the fact that $\mathcal{V}$ is a translationally-covariant operation; in the fifth equality, we have used the fact that $E$ is an incoherent effect; and in the last equality, we have used Eq.~\eqref{TIStinespring}.
\color{black}

For the converse implication, we refer the reader to the result on the form of the Stinespring dilation for group-covariant quantum operations by Keyl and Werner \cite{keyl1999optimal}. 
 \end{proof}

\subsubsection{Characterization via modes of translational asymmetry}\label{Modes}

We begin by introducing some technical machinery.

Denote the preferred subspaces of $\mathcal{H}$ relative to which coherence is evaluated, that is, the eigenspaces of $L$, by $\{ \mathcal{H}_{\lambda_n} \}$, where $\{\lambda_n\}_n$ is the set of eigenvalues of $L$ (these may be discrete or continuous). Let the set of \emph{modes} $\Omega $ be the set of the gaps between all eigenvalues, i.e. $\{\lambda_n-\lambda_m\}_{n,m}$. In the case where $L$ is the system Hamiltonian, each element of $\Omega$ can be interpreted as  a \emph{frequency} of the system. 
% $\{ \Pi_n \}$.  
%If these subspaces are not 1-dimensional, then we can 
%We define an orthonormal basis for each, indexed by $\alpha$, so that, $$\mathcal{H}_n \equiv {\rm span}\{ |n,\alpha\rangle  \}_{\alpha}.$$  

Elements of the set $\Omega$ label different modes in the system. For any $\omega \in \Omega$, define the superoperator
\beq 
\mathcal{P}^{(\omega)}= \lim_{x_0\rightarrow \infty} \frac{1}{2x_0} \int_{-x_0}^{x_0}{\rm d}x\  e^{-i\omega x}\  \mathcal{U}_{L,x} \ ,
\eeq
where $\mathcal{U}_{L, x}(\cdot)=U_{L, x}(\cdot)  U^{\dag}_{L, x} = e^{-i x L}(\cdot)  e^{i x L}$. 
%Roughly speaking, 
This superoperator is the projector that erases all the terms in the input operator except those which connect eigenstates of $L$ whose eigenvalues are different by $\omega$, i.e., all except those which are of the form $ |\lambda_n,\alpha \rangle\langle\lambda_n+\omega,\beta|$, where $ |\lambda_n, \alpha \rangle$ and $ |\lambda_n+\omega,\beta \rangle$ are eigenstates of $L$ with eigenvalues $\lambda_n$ and $\lambda_n+\omega$ respectively. \color{black} One can easily show that
\bes\label{modes-main}
\begin{align}
\sum_{\omega\in \Omega} \mathcal{P}^{(\omega)}&=\mathcal{I}_\text{id}\\
 \mathcal{P}^{(\omega)}\circ \mathcal{P}^{(\omega')}&=\delta_{\omega,\omega'}  \mathcal{P}^{(\omega)}\\
 \mathcal{P}^{(0)}&=\mathcal{D}\\  \mathcal{U}_{L,x}\circ \mathcal{P}^{(\omega)}&=e^{i \omega x}\mathcal{P}^{(\omega)}\ ,
\end{align}
\ees
where $\mathcal{I}_\text{id}$ is the identity superoperator, and $\delta_{\omega,\omega'}$ is the Kronecher delta.

The set of superoperators $\{\mathcal{P}^{(\omega)}: \omega\in\Omega \}$ are a complete set of projectors to different subspaces of the operator space $\mathcal{B}$. It can be easily shown that these subspaces are orthogonal according to the Hilbert-Schmidt inner product, defined by $(X,Y) \equiv {\rm Tr}( X^{\dag} Y)$ for arbitrary pair of operators $X,Y\in\mathcal{B}$.  Therefore, the operator space $\mathcal{B}$ can be decomposed into a direct sum of operator subspaces, $\mathcal{B} = \bigoplus_{\omega\in\Omega} \mathcal{B}^{(\omega)}$, where each $\mathcal{B}^{(\omega)}$ is the image of $\mathcal{P}^{(\omega)}$.

Note that any operator in the operator subspace $\mathcal{B}^{(\omega)}$ transforms distinctively under translations,
\begin{eqnarray}\label{modektransf}
A\in  \mathcal{B}^{(\omega)} \Longrightarrow \mathcal{U}_{L,x} ( A) = e^{i\omega x} A.
\end{eqnarray}
We refer to $\mathcal{B}^{(\omega)}$ as the ``mode $\omega$'' operator subspace.  For any operator $A$, the component of that operator in the operator subspace $\mathcal{B}^{(\omega)}$, denoted $$A^{(\omega)} \equiv \mathcal{P}^{(\omega)} (A),$$ is termed the ``mode $\omega$ component of $A$''.

Clearly, every incoherent (i.e. translationally symmetric) state lies entirely within the mode 0 operator subspace, while a coherent (i.e. translationally asymmetric) state has a component in at least one mode $\omega$ operator subspace with $\omega\ne 0$. %Equivalently, every translationally symmetric state is in the mode 0 operator subspace, while every translationally asymmetric state has a component in at least one mode $k$ operator subspace with $k\ne 0$.  

Operator subspaces associated with distinct $\omega$ values have been called ``modes of asymmetry'' in Ref.~\cite{Modes}, where the decomposition of states, operations and  measurements into their different modes was shown to constitute a powerful tool in the resource theory of asymmetry.   

\begin{example} Consider the special case where $J_{\hat{z}}$ is the angular operator in the $\hat{z}$ direction. For simplicity, assume that $J_{\hat{z}}$ is non-degenerate and let $\{|m\rangle \}_{m}$ be its orthonormal eigenbasis, where $|m\rangle$ is the eigenstate of $J_{\hat{z}}$  with eigenvalue $m$. Since the eigenvalues of the angular momentum operator are all separated by integers, it follows the  set of modes $\Omega$ is a subset of the integers, $\Omega \subseteq \mathbb{Z}$. Then, for each integer $k\in \Omega$ we have
\begin{equation}
\mathcal{B}^{(k)} \equiv {\rm span} \{ |m \rangle\langle m+k|\}_{m} .
\end{equation}
Furthermore, the mode $k$ component of any operator $A$ is given by
\beq
A^{(k)}=\sum_{m} |m\rangle\langle m+k|\  \langle m|A|m+k\rangle\ .
\eeq
The mode $k$ of the density operator $\rho$ corresponds to \emph{coherence of order $k$} in the context of magnetic resonance techniques \cite{cappellaro2014implementation}.
\end{example}

%In the resource theory of coherence, by contrast, the only important distinction is between the operator subspace with $k=0$, and the direct sum of operator subspaces with $k\ne 0$.  

With these notions in hand, we can provide the mode-based characterization of the translationally-covariant operations.
% alternative definition of the free operations. 

%\begin{definition}\label{defnmodes}
\begin{proposition}\label{defnmodes}
A quantum operation  $\mathcal{E}$ is translationally covariant relative to the generator $L$ if and only if it preserves the modes of asymmetry associated to $L$, that is, if and only if the mode $\omega\in \Omega$ component of the input state is mapped to the mode $\omega\in\Omega$ component of the output state. Formally, the condition is that whenever $\mathcal{E}(\rho)=\sigma$, we have $\mathcal{E}(\rho^{(\omega)})=\sigma^{(\omega)}$ where $\rho^{(\omega)} \equiv \mathcal{P}^{(\omega)} (\rho)$ and $\sigma^{(\omega)} \equiv \mathcal{P}^{(\omega)} (\sigma)$.  Note that $\sigma$ is only a normalized state if $\mathcal{E}$ is a channel (i.e. trace-preserving) and is otherwise subnormalized. 
%\end{definition}
\end{proposition}
The proof follows immediately from properties listed in Eq.(\ref{modes-main}) (See \cite{Modes} for further discussion).

\color{black}

%So, to summarize we have
%\bes
%\begin{align}
%&\rho^{(k)}=\mathcal{M}^{(k)}(\rho) ,\ \ \   \sigma^{(k)}=\mathcal{M}^{(k)}(\sigma) \ ,\\
%&\rho=\sum_{k\in \mathbb{Z}} \rho^{(k)},\ \ \ \sigma=\sum_{k\in \mathbb{Z}} \sigma^{(k)}\\
%&\mathcal{E}_\text{sym}(\rho^{(k)})=\sigma^{(k)}\ , \ \ \ \ \ \ \ \  \forall k\in \mathbb{Z} .
%\end{align}
%\ees
%This means that under symmetric operations the mode $k$ component of the input state is mapped to the mode $k$ component of the output state.

%\subsubsection{Definition in terms of Kraus decomposition}
\subsubsection{Characterization via Kraus decomposition}

%\begin{definition}\label{defnKraus}
\begin{proposition}\label{defnKraus}
A quantum operation $\mathcal{E}$ is translationally covariant if and only if it admits of a Kraus decomposition of the form
\begin{equation}
%\mathcal{E}(\cdot) = \sum_{\omega\in\Omega,\alpha} K_{\omega,\alpha} (\cdot) K_{\omega,\alpha}^{\dag},
\mathcal{E}(\cdot) = \sum_{\omega,\alpha} K_{\omega,\alpha} (\cdot) K_{\omega,\alpha}^{\dag},
\end{equation}
where the elements of the set $\{ K_{\omega,\alpha}\}_{\alpha}$  are all mode $\omega$ operators.
\end{proposition}

To see that any quantum operation with such a Kraus decomposition is translationally covariant, we note that 
\begin{equation}
\mathcal{U}_{L,x} \circ \mathcal{E}  \circ  \mathcal{U}^{\dag}_{L,x} =\sum_{\omega\in\Omega,\alpha} \mathcal{U}_{L,x}(K_{\omega,\alpha}) (\cdot) (\mathcal{U}_{L,x}(K_{\omega,\alpha}))^{\dag}
%\mathcal{E}(\cdot) = \sum_{k,\alpha} K_{k,\alpha} (\cdot) K_{k,\alpha}^{\dag},
\end{equation}
and then use the fact that $K_{\omega,\alpha} \in \mathcal{B}^{(\omega)}$ and Eq.~\eqref{modektransf} to infer that $\mathcal{U}_{L,x}(K_{\omega,\alpha}) = e^{i\omega x} K_{\omega,\alpha}$, which in turn implies 
\begin{equation}
\forall x\in \mathbb{R}:\mathcal{U}_{L,x} \circ \mathcal{E}  \circ  \mathcal{U}^{\dag}_{L,x} =\sum_{\omega,\alpha} K_{\omega,\alpha} (\cdot) K_{\omega,\alpha}^{\dag}= \mathcal{E},
\end{equation}
which, from Eq.~\eqref{TI_opo2}, simply asserts the translational covariance of $\mathcal{E}$.

%The proof that this definition is equivalent to definition ?? proceeds as follows.  
%First, note that if $\mathcal{E}$ is a translationally covariant operation, then 
%\begin{equation}
%\mathcal{T}(g)\circ\mathcal{E}\circ\mathcal{T}(g^{-1}) = \mathcal{E}.
%\end{equation}

%The proof that every translationally-covariant operation has a Kraus decomposition of the form specified is more involved.  It can be inferred from a result in Ref.~\cite{gour2008resource} which characterizes the Kraus decomposition of any group-covariant operation, by specializing the result to the case of a translation group.

The proof that every translationally-covariant operation has a Kraus decomposition of the form specified can be inferred from a result in Ref.~\cite{gour2008resource} which characterizes the Kraus decomposition of any group-covariant operation, by specializing the result to the case of a translation group.  Alternatively, it can be inferred from the Stinespring dilation, Proposition \ref{propTCStinespring}, using a \color{black} slight generalization (from channels to all operations) of the \color{black} argument provided in Appendix A.1 of Ref.~\cite{marvian2015quantum}. 
%\color{red} [Do these proofs cover only the case of channels? ]\color{black}

\color{black}
Proposition \ref{defnKraus} also implies:
\begin{corollary}
A quantum operation $\mathcal{E}$ is translationally-covariant if and only if it admits of a Kraus decomposition every term of which is translationally-covariant. 
\end{corollary}
It suffices to note that each term of the Kraus decomposition specified in proposition \ref{defnKraus} is translationally-covariant.  It follows that if the different terms in this Kraus decomposition correspond to the different outcomes of a measurement, then even one who  post-selected on a particular outcome would describe the resulting operation as translationally-covariant.
\color{black}

%All unitary operations within a given subspace $\mathcal{H}_{n}$ are U(1)-invariant. However, there are more U(1)-invariant unitaries besides these, specifically, nontrivial unitaries on the subspace $\mathcal{H}^{\prime }=$ $\mathrm{span}\{\left\vert n\right\rangle \}$. Because unitary operations have a single Kraus operator, they are irreducible U(1)-invariant operations. However, the only way in which a single Kraus operator $K$ can be unitary is if $k=0$ in Eq.~(\ref{eq:krausforU1}), i.e. the operation does not allow shifts in the number, and $|c_{n}|=1,$ so that $K$ must have the form $\sum_{n}e^{i\chi_{n}}|n\rangle \langle n|.$ All told, the unitary operations that are U(1)-invariant have the effect of merely changing the relative phases of the $\left\vert n\right\rangle$.

%\subsection{Why the restriction to translationally-covariant operations arises naturally}
\subsection{Physical justifications for the restriction to translationally-covariant operations}\label{JustificationsTC}
%\label{NaturalnessTC}
%\subsection{Does the restriction to dephasing-covariant operations arise naturally?}

As noted in the introduction, it is critical that any definition of the restricted set of operations in a resource theory must be justifiable operationally.  
%There must be some operational constraint that distinguishes the restricted operations such that all and only these operations are \emph{free}. 
In this section, we discuss different physical scenarios in which the set of translationally-covariant  operations are naturally distinguished as the set of easy or freely-available operations.

\subsubsection{Fundamental or effective symmetries of Hamiltonians}
%{\bf Fundamental or effective symmetries of Hamiltonians.} 

If, for a set of systems, the Hamiltonians one can access are symmetric and the states and measurements that one can implement are also symmetric, then for any given system only symmetric operations are possible\footnote{This is the easy half of the dilation theorem in Prop.~\ref{TIStinespring}.}.  Such symmetry constraints can sometimes be understood to be consequences of fundamental or effective symmetries in the problem.

A constraint of translational symmetry on the Hamiltonian is fundamental if it arises from a fundamental symmetry of nature, such as a symmetry of space-time.  It is effective if it arises from practical constraints, for instance, if one is interested in time scales or energy scales for which a symmetry-breaking term in the Hamiltonian becomes negligible.  A translational symmetry constraint on the states and measurements can sometimes arise as a consequence of this symmetry of the Hamiltonian.   For instance, if the only states that one can freely prepare are those that are thermal, then given that thermal states depend only on the Hamiltonian, any fundamental or effective symmetry of the Hamiltonian is inherited by the thermal states.

\color{black}
\subsubsection{Lack of shared reference frames}
%{\bf Lack of shared reference frames.}
\label{Sec:RF}

The most natural experimental restriction that leads to translationally-covariant operations is when one lacks access to any reference frame relative to which the translations can be defined.  Such a lack of access can arise in a few ways.

For a pair of separated parties, each party may have a local reference frame, but no information about the relation between the two reference frames. For instance, a pair of parties may each have access to a Caretesian reference frame (or clock, or phase reference), but not know what rotation (or time-translation or phase-shift) relates one to the other.  
%This kind of restriction
% is particularly important in  quantum optics experiments, where establishing  a shared phase reference between two distant parties requires extra resources. 

Under this kind of restriction, each party essentially lacks access to the reference frame of the other.  It has been shown that this lack of a shared reference frame implies that the only operations that one party can implement, relative to the reference frame of the other, are those that are group-covariant (see Refs.~\cite{QRF_BRS_07,MS11}). For instance, if two parties lack of a shared phase reference, then the only operations whose descriptions they can agree on are phase-covariant operations. 

It is also possible that the reference frame that one requires cannot even be prepared locally, due to technological limitations.  For instance, only after the experimental realization of Bose-Einstein condensation in atomic systems~\cite{anderson1995observation, davis1995bose}, was it possible to prepare a system that could serve as a reference frame for the phase conjugate to atom number.

\subsubsection{Metrology and phase estimation}\label{infodual}
%{\bf Metrology and phase estimation}

%So one cannot justify the restriction to translationally-covariant operations by lack of access to a reference frame. 
%The reason that the restriction to translationally-covariant operations is relevant to quantum metrology is, in fact, rather subtle.

Unspeakable coherence is the main resource for quantum metrology, and in particular phase estimation.  In this context, a state is a resource to the extent that it allows one to estimate an unknown translation applied to the state (such as a phase-shift, a rotation, or evolution for some time interval). 
Suppose one prepares a system in the state $\rho$ prior to it being subjected to a unitary translation $\mathcal{U}_{L,x}$, where $x$ is unknown.  In this case, one knows that the state after the translation is an element of the ensemble $\{ \mathcal{U}_{L,x}(\rho)\}_x$ and the task is to estimate $x$.    Clearly, if $\rho$ is invariant under translations (i.e., incoherent), then it is useless for the estimation task.  In this sense, translationally asymmetry, and hence coherence relative to the eigenspaces of the generator of translations, is a necessary resource for metrology.

Furthermore, as we show in the following,  in this context, the set of translationally-covariant operations has also a simple and natural interpretation.
Suppose one is interested in determining which of two states, $\rho$ and $\sigma$, is the better resource for the task of estimating an unknown translation.  To do so, one must determine which of the two encodings, $x\rightarrow  \mathcal{U}_{L,x}(\rho)$ and $x\rightarrow  \mathcal{U}_{L,x}(\sigma)$, carries more information about $x$. But suppose there exists a quantum operation $\mathcal{E}$ such that for all $x$,  it transforms $\mathcal{U}_{L,x}(\rho)$ to  $\mathcal{U}_{L,x}(\sigma)$, i.e.,
\beq\label{trans}
\forall x:\ \ \ \ \mathcal{E}(\mathcal{U}_{L,x}(\rho))=\mathcal{U}_{L,x}(\sigma)\ .
\eeq
Here, the quantum operation $\mathcal{E}$ can be thought as an information processing which we perform on the state before performing the measurement which yields the value of $x$. If such a quantum operation exists,  then we can be sure that the state $\rho$ is more useful than $\sigma$ for this metrological task. Because any information that we can obtain using the state $\sigma$, we can also obtain if we use the state $\rho$. 

It turns out that any such information processing $\mathcal{E}$ can be chosen to be translationally covariant with respect to translation $U_{L,x}$, i.e. 
\begin{proposition}\label{prop14}
For any given pair of states $\rho$ and $\sigma$ the following statements are equivalent:\\
\noindent (i) There exists a translationally-covariant quantum operation $\mathcal{E}$ 
%(i.e. satisfying Eq.~(\ref{TI_opo})) 
such that $\mathcal{E}(\rho)=\sigma$. \\
\noindent (ii) There exists a  quantum operation $\mathcal{E}$ such that $\mathcal{E}\left( \mathcal{U}_{L,x}(\rho)\right)= \mathcal{U}_{L,x}(\sigma)$, for all $x\in \mathbb{R}$. 
\end{proposition}
This is the specialization to the case of a translational symmetry group of a similar proposition for an arbitrary symmetry group the proof of which is presented in Ref.~\cite{MS11}, where we have also presented a version of this duality for pure states and unitaries, its  interpretation in terms of reference frames, and some of its applications. 

Statement (ii) in proposition~\ref{prop14} concerns the relative quality of $\rho$ and $\sigma$ as resources for metrology, while statement (i) concerns the relative quality of  $\rho$ and $\sigma$ within the resource theory defined by the restriction to translationally-covariant operations.  The partial order of quantum states under translationally-covariant operations, therefore, determines their relative worth as resources for metrology. Note that in this context, translationally-covariant operations are \emph{all and only} the operations that are relevant.

It follows from proposition \ref{prop14} that any function which quantifies  the  performance of states in this metrological task should be a measure of unspeakable coherence.

%But by information processing we cannot increase the measure of information. So, the amount of information we can extract about the unknown translation $x$ from the output state $U(x)\sigma U^{\dag}(x)$ cannot be larger than the amount of information we can extract from the input state $U(x)\rho U^{\dag}(x)$. We can quantify this amount information using the standard tools of information theory, and in this way we can indeed quantify the amount of asymmetry of the state. 

%\subsubsection{Coherence in quantum thermodynamics}
%\subsubsection{Thermal operations and quantum thermodynamics}
%\subsubsection{Time-translational-covariant operations in thermodynamics}
%{\bf Time-translational-covariant operations in thermodynamics.}
%{\bf Thermodynamics.}
\subsubsection{Thermodynamics}

The resource theory of {\em athermality} seeks to understand states deviating from thermal equilibrium as a resource \cite{brandao2015second, brandao2013resource, gour2015resource, skrzypczyk2014work, horodecki2013fundamental, aaberg2013truly, janzing2000thermodynamic, cwiklinski2015limitations} . The free operations defining the theory, termed {\em thermal operations} are all and only those that can be achieved using thermal states, unitaries that commute with the free Hamiltonian, and the partial trace operation.  (The restriction on unitaries is motivated by the fact that were one to allow more general unitaries, one could increase the energy of a system, thereby allowing thermodynamic work to be done for free.)

Noting that: (i) if a unitary commutes with the free Hamiltonian, then it is covariant under time-translations, and (ii) because thermal states are defined in terms of the free Hamiltonian, they are symmetric under time-translations, it follows from the dilation theorem for translationally-covariant operations (Prop.~\ref{propTCStinespring}) that the restriction to thermal operations implies a restriction to time-translation-covariant operations.

%(Note also that any state with coherence between the energy eigenspaces, i.e. asymmetric relative to time-translations, is an athermal state, which is a resource for thermodynamic tasks)

% Therefore, thermal operations are included in the set of time-translation-covariant operations.  Any state with coherence between the energy eigenspaces, which is to say one that is asymmetric relative to time-translations, is therefore 

%The method  of mode decompositions has been used in Refs.~\cite{lostaglio2015quantum} and \cite{lostaglio2015description} to study coherence in the context of quantum thermodynamics.  Also note that  \emph{Catalytic Coherence} discussed in \cite{Aberg2} uses the notion of coherence as translational asymmetry.

\color{black}

%\subsubsection{Time-translational-covariant operations as the easy operations in control theory}
%\subsubsection{Time-translational-covariant operations in control theory}
%{\bf Time-translational-covariant operations in control theory.}
%{\bf Control theory.}
\subsubsection{Control theory}

Suppose we are trying to prepare a quantum system in a desired state by applying a sequence of control pulses to the system.  Then, there is an important distinction between the pulses which commute with the system Hamiltonian $H$, and hence are invariant under time translations, and those which are not. Namely, to apply the pulses which do not commute with the system Hamiltonian, we need be careful about the timing of the pulses, and  also the duration that the pulse is acting on the system.

To see this, first assume that the pulses are applied instantaneously, i.e., the width of the pulse is sufficiently small that the intrinsic evolution of the system generated by the Hamiltonian $H$ during the pulse is negligible. Then, if instead of applying the control unitary $V$ at the exact time $t$, we apply it at time $t+\Delta t$,  the effect of applying this pulse would be equivalent to applying  the pulse $e^{i H\Delta t} V e^{-i H\Delta t}$  at time $t$, instead of the desired pulse $V$. If $V$ does not commute with the Hamiltonian $H$, then in general  $V$ and $e^{i H\Delta t} V e^{-i H\Delta t}$ are different unitaries, and so the final state  is different from the desired state. 
%So, in the case of  pulses which do not commute with the Hamiltonian, the exact timing of applying the pulse could be a very sensitive factor.  

Furthermore, if the control pulse $V$ commutes with the system Hamiltonian $H$, then dealing with the nonzero width of the pulse is much easier and we do not need to be worried about the  intrinsic evolution of the system during the pulse, as we now demonstrate.   In general, to apply a control unitary $V$ we need to apply a control field to the system. The effect of this control field  can be described by a term $H_\text{cont}(t)$  which is added to the system Hamiltonian $H$. Then, to implement a control unitary $V$ which does not commute with the Hamiltonian $H$, we need to apply a control field $H_\text{cont}(t)$ which does not  commute with the Hamiltonian $H$. In this case, the width of the control pulse, i.e., the duration over which we apply the control field $H_\text{cont}(t)$, becomes an important parameter. In practice, in many situations we need to choose  the control field $H_\text{cont}(t)$ to be strong enough so that the evolution of the system during the pulse width is negligible. On the other hand, if the control field $H_\text{cont}(t)$ commutes with the system Hamiltonian, then the effect of finite width can be easily taken into account, and so we do not need to apply strong fields to the system.  

It follows that in this context,  operations which are covariant under time translations are easy to  implement, because for this type of operation, there is no sensistivity to the exact timing and the width of the control pulses. So it is natural to consider the operations that are covariant under time-translations as the set of freely-implementable operations, and this again leads us to treat coherence as translational asymmetry.

%\subsubsection{Quantum Speed Limits}

\color{blue}
% It turns out that  the notion of coherence as translational asymmetry  provides a new interpretation of quantum speed limits, i.e., the Mandelstam-Tamm  \cite{QSL_MT} bound and the Margolus-Levitin bound  \cite{Margolus:98}. As explained in Ref.~\cite{Marvian2015}, the minimum time it takes for the system to evolve to a (partially) distinguishable state is the inverse of a particular measure of translational asymmetry (hence a particular measure of unspeakable coherence). Furthermore, the standard quantum speed limits can be interpreted as upper bounding this measure of coherence with another measure of coherence. 
 %See Sec.(\ref{Sec:QSL}) for more details. 
 \color{black}

%\subsection{Generalization to the case where translations are not collective on all systems}
\subsection{Covariance with respect to independent translations}

%As noted in Ref.~\cite{coecke2014mathematical}, to specify a resource theory it is not sufficient to specify the free states and operations on a single system, one must specify also what the free states and operations are for a composite system.

%As noted in Sec.~\ref{}, 
It can happen that the set of all systems is partitioned into subsets and that the action of the translation group is only collective for those systems within a given subset, while it is independent for different subsets.  Suppose that the subsets are labelled by $\alpha$ and that for the set of systems of type $\alpha$, the generator of collective translations on these systems is denoted 
%perator whose distinct eigenspaces are the preferred subspaces for those systems is 
$L^{(\alpha)}$.   Consider the group element consisting of a translation by $x_{\alpha} \in \mathbb{R}$ for all the systems of type $\alpha$.  We label this group element by the independent translation parameters $(x_1, x_2, \dots, x_A) \in \mathbb{R}^A$ where $A$ denotes the number of different types of system.  
The unitary representation of this group element is 
\begin{equation}
U_{\{L^{(\alpha)}, x_{\alpha}\}} \equiv \bigotimes_{\alpha} e^{-i x_{\alpha} L^{(\alpha)}}.
\end{equation}
The superoperator representation of this group element is then
\begin{eqnarray}
\mathcal{U}_{\{L^{(\alpha)}, x_{\alpha}\}}(\cdot) &\equiv& U_{\{L^{(\alpha)}, x_{\alpha}\}}  (\cdot) U^\dag_{\{L^{(\alpha)}, x_{\alpha}\}}  %\nonumber\\
%&=& \Big( \bigotimes_{\alpha} e^{-i x_{\alpha} L^{(\alpha)}}\Big) (\cdot) \Big( \bigotimes_{\alpha} e^{-i x_{\alpha} L^{(\alpha)}}\Big).
\end{eqnarray}

In this case, the set of free states are translationally-invariant relative to translations generated by the set of generators $\{L^{(\alpha)}\}$.  These are the states that are block-diagonal relative to the distinct joint eigenspaces of $\{L^{(\alpha)}\}$. 

The free operations are those that are translationally-covariant relative to the set of generators $\{L^{(\alpha)}\}$,  that is, $\forall (x_1, x_2, \dots, x_A) \in \mathbb{R}^A\ ,$
 %free if it is covariant with respect to a translation of $x$ generated by the observable $L$, for all $x\in \mathbb{R}$,
\begin{align}
\mathcal{U}_{\{L^{(\alpha)}, x_{\alpha}\}} \circ \mathcal{E}  &=\mathcal{E} \circ  \mathcal{U}_{\{L^{(\alpha)}, x_{\alpha}\}}.\label{TCcomposite}
\end{align}

%defined by Eq.~\eqref{} but where $\mathcal{U}_{L,x}$ is replaced with $\mathcal{U}_{\{L^{(\alpha)}, x_{\alpha}\}}$.  

Indeed, all of the results expressed in this section can be generalized by substituting the translations $x\in \mathbb{R}$ with $(x_1, x_2, \dots, x_A) \in \mathbb{R}^A$, the superoperator $\mathcal{U}_{L,x}$ with $\mathcal{U}_{\{L^{(\alpha)}, x_{\alpha}\}}$, and the eigenspaces of $L$ with the joint eigenspaces of $\{L^{(\alpha)}\}$.

%This freedom in the choice of the translation group will allow us, in Sec.~\ref{}, to compare translationally-covariant operations to dephasing-covariant and incoherence-preserving operations for the choices of preferred subspaces that are preferred by the latter approaches.

% If the translation group does not act collectively on all systems of interest,  but rather acts only collectively within certain subsets and independently between these subsets (as in the example presented above),

\color{black}
%\section{speakable coherence as noncovariance under dephasing}
\section{Coherence via dephasing-covariant operations}\label{DCcoherence}

Much recent work seeking to quantify coherence as a resource has considered \emph{speakable} coherence.
%Most recent work has been on this notion.  
The article of Baumgratz, Cramer and Plenio~\cite{Coh_Plenio} (BCP) provides one such proposal, 
%based on operations that preserve incoherence, 
which has been taken up by most other authors who have sought to characterize coherence as a resource.  Nonetheless, we postpone our discussion of the BCP proposal to Sec.~\ref{BCP} and instead begin our discussion of speakable coherence with a very different proposal, based on operations that are dephasing-covariant. 
 %A variant of this proposal has recently been considered in Ref.~\cite{yadin2014new}\footnote{Ref.~\cite{yadin2014new} came to our attention as we were preparing this manuscript.}. 
 We here assess the dephasing-covariance approach and compare it
 %the approach to coherence based on dephasing-covariant operations 
 to the translational-covariance approach 
 %one based on translationally-covariant operations
  discussed in the last section.
 %The discussion of the BCP proposal is taken up in Sec.~\ref{BCP}, where we compare it the one we introduce in this section, and provide some criticisms.
%In the next section, we will discuss the differences betwee our proposal and that of BCP and we will provide criticisms of the latter.

%\subsection{Our proposal}

\subsection{Free operations as dephasing-covariant operations}

\subsubsection{Definition of dephasing-covariant operations}

As before, suppose that the preferred subspaces relative to which coherence is to be quantified are $\{ \mathcal{H}_l \}_l$ and are associated with the projectors $\{ \Pi_l \}_l$.
%, and let $L$ denote a Hermitian operator whose eigenspaces are $\{ \Pi_l \}_l$.

\begin{definition}
We say that a quantum operation $\mathcal{E}$ is \emph{dephasing-covariant} relative to the preferred subspaces
 %(the eigenspaces of $L$)
  if it commutes with the associated dephasing operation, $\mathcal{D}$ of Eq.~\eqref{Dephasing}, i.e., if
%i.e. if it satisfies the following 
\beq \label{DefDC}
\mathcal{E} \circ\mathcal{D}=\mathcal{D}\circ \mathcal{E}\ .
\eeq
\end{definition}

%In the following, we will often write $\mathcal{D}_L$ as simply $\mathcal{D}$ when no confusion is likely to occur.
%\color{red} [Add $L$ subscript to all that follows?] \color{black}
%\begin{definition}
%We say that a quantum operation $\mathcal{E}$ is \emph{dephasing-covariant} if it commutes with the dephasing operation, i.e., if
%\beq \label{DefDC}
%\mathcal{E} \circ\mathcal{D}=\mathcal{D}\circ \mathcal{E}\ .
%\eeq
%\end{definition}

Note that if the input and output spaces of the map $\mathcal{E}$ are distinct, then the dephasing map is different on the input and output spaces, but we do not indicate this difference in our notation.

Dephasing-covariant quantum operations are easily seen to be incoherence-preserving. 
%verified to have property \ref{property1} of Sec.~\ref{constraintsonfreeops}.  Property \ref{property2} is inferred as follows: 
It suffices to note that if $\mathcal{E}$ is dephasing-covariant, then for any incoherent state $\rho\in\mathcal{I}$,
\begin{equation}\label{IPofDC}
\mathcal{E}(\rho)=\mathcal{E}(\mathcal{D}(\rho))=\mathcal{D}\left(\mathcal{E}(\rho)\right),
\end{equation}
and therefore $\mathcal{E}(\rho)$ is invariant under dephasing and hence incoherent.
%Dephasing-covariant quantum operations cannot generate coherence, i.e., if $\mathcal{E}$ is dephasing-covariant, then for any incoherent state $\rho\in\mathcal{I}$,
%\beq
%\mathcal{E}(\rho)=\mathcal{E}(\mathcal{D}(\rho))=\mathcal{D}\left(\mathcal{E}(\rho)\right),
%\eeq
%and therefore $\mathcal{E}(\rho)\in\mathcal{I}$. The set of dephasing-covariant quantum operations is closed under composition of operations and under convex combinations. 
%Dephasing-covariant quantum operations are easily verified to have properties \ref{property1} and \ref{property2}.  For property \ref{property1}, it suffices to note that every incoherent state on a system is invariant under the dephasing map on that system.   

%Finally, if the input space is trivial, so that the map $\mathcal{E}$ corresponds to the preparation of a state $\rho$ on the output space, then Eq.~\eqref{DefDC} reduces to 
%\begin{align}
%\rho  &=\mathcal{D} (\rho)
%\end{align}
%which implies that $\rho$ is an incoherent state. So the definition of free operation, when specialized to states, yields the definition of an incoherent state, which is property \ref{property3}.

\subsubsection{Dephasing-covariant measurements}
If the output space is trivial, so that the map corresponds to tracing with a measurement effect $E$ on the input space, that is,  $\mathcal{E}(\cdot) = {\rm Tr}(E \cdot)$, then Eq.~\eqref{DefDC} reduces to 
\begin{align}
{\rm Tr}(E (\cdot)) &= {\rm Tr}(E \mathcal{D}(\cdot)) \nonumber\\
&= {\rm Tr}(\mathcal{D}(E) \;(\cdot)) ,
\end{align}
where we have used the fact that $\mathcal{D}$ is self-adjoint relative to the Hilbert-Schmidt inner product, and this in turn implies 
\begin{align}
\mathcal{D} (E)  &=E,
\end{align}
where we have used the fact that the set of all quantum states form a basis of the operator space. 
Thus $E$ is an incoherent effect, i.e., it is block-diagonal with respect to the preferred subspaces.

%It follows that the dephasing-covariant measurements are those corresponding to POVMs all of whose elements are incoherent (i.e. block-diagonal relative to the preferred subspaces).  
\begin{proposition}\label{prop:meas}
%The set of POVMs considered to be free (i.e. included among the dephasing-covariant operations) are all and only those that are diagonal in the preferred basis.  
%considered to be free (i.e. included among the dephasing-covariant operations)
%The set of POVMs associated to dephasing-covariant measurements are all and only those whose elements are all block-diagonal relative to the preferred subspaces.
% The dephasing-covariant set of POVMs contains those all of whose elements are block-diagonal relative to the preferred subspaces.
 A POVM is dephasing-covariant if and only if all of its effects are incoherent.
 %elements are block-diagonal relative to the preferred subspaces.
\end{proposition}

Comparing to proposition \ref{propTC:meas}, we see that a POVM is dephasing-covariant if and only if it is translationally-covariant.
%the POVM associated to a measurement is dephasing-covariant if and only if it is translationally-covariant.\color{black}

\subsubsection{Dephasing-covariant unitary operations}

Because the dephasing-covariant operations are all incoherence-preserving, the set of unitary dephasing-covariant operations are included within the set of unitary incoherence-preserving operations.  As it turns out,  the two sets are in fact equivalent.  
%We postpone the proof until Sec.~\ref{relationIPDC} (see proposition \ref{coincidenceunitaries}).  
We postpone the proof until Sec.~\ref{relationIPDC}, Proposition \ref{coincidenceunitaries}, where we also present the general form of such unitaries.

%In Sec.~\ref{}, we prove a characterization of the unitary incoherence-preserving operations, which therefore yields a characterization of the unitary dephasing-covariant operations, which can be summarized as follows.   If $\pi$ denotes a dimension-preserving permutation of the preferred subspaces, and $V_l$ denotes an arbitrary unitary {\em within} the $\mathcal{H}_l$ subspace, then $\mathcal{V}(\cdot)= V(\cdot)V^{\dag}$ is incoherence-preserving if and only if 
%\beq\label{unit35}
%V=  \sum_{l,\alpha_l} V_{\pi(l)} |\pi(l), \alpha_{\pi(l)} \rangle \langle l,\alpha_{l} |.
%\eeq
\color{black}

\subsubsection{Considerations regarding the existence of a free dilation}
%\subsubsection{Sufficient condition on Stinespring dilation}
%\subsubsection{Marginal maps of dephasing-covariant unitaries with incoherent ancillas}
\color{black}

\label{StinespringForDC}

By analogy with the considerations of Sec.~\ref{TCStinespring}, in order to discuss the possibility of dilating a dephasing-covariant operation with the use of an ancilla in an incoherent state and a dephasing-covariant unitary on the composite of system and ancilla, one needs to specify not only the preferred subspaces (relative to which coherence is defined) on the system and ancilla individually, but on the composite of system and ancilla as well.  When the input and output spaces differ this needs to be specified on the outputs as well, as discussed in Sec.~\ref{TCStinespring}.  

%\color{cyan} To distinguish the preferred subspaces for the system and for the system-ancilla composite, we denote the former $\{ \mathcal{H}^{(s)}_l \}_l$ and the latter $\{ \mathcal{H}^{(sa)}_{l'} \}_{l'}$.  In the following, we make only a single assumption about the relation between these, namely, that 
%\begin{equation}
%\forall l', \exists l : {\rm tr}_a \Pi^{(sa)}_{l'} = \Pi^{(s)}_{l}.
%\end{equation}
%\color{black}

Recall that the system input and output spaces are denoted $\mathcal{H}_{s}$ and $\mathcal{H}_{s'}$, the ancilla input and output spaces are denoted $\mathcal{H}_{a}$ and $\mathcal{H}_{a'}$, and the composite of system and ancilla is $\mathcal{H}_{sa} = \mathcal{H}_{s}\otimes \mathcal{H}_{a}= \mathcal{H}_{s'}\otimes \mathcal{H}_{a'}$.
We also denote the associated sets of incoherent states and dephasing maps with the subscripts $s$, $a$ and $sa$ (or $s'a'$). 

%For the resource theory to be consistent, any state that is free for the composite must have a reduction to each component that is also considered free: for all $\rho_{sa} \in \mathcal{I}_{sa}$, ${\rm tr}_a \rho_{sa}  \in \mathcal{I}_{s}$ and  ${\rm tr}_s \rho_{sa} \in \mathcal{I}_{a}$. Consequently, 
%\begin{eqnarray}
%{\rm tr}_a \mathcal{D}_{sa} = \mathcal{D}_{s}\otimes {\rm tr}_a,\label{TrDsa}\\
%{\rm tr}_s \mathcal{D}_{sa} =  {\rm tr}_s \otimes \mathcal{D}_{a}.
%\end{eqnarray}
%Similarly, the tensor product of free states on the components must also be a free state on the composite:
%\begin{equation}\label{IsIaIsa}
%\mathcal{I}_{s} \otimes \mathcal{I}_a \subseteq \mathcal{I}_{sa},
%\end{equation}
%\color{blue} which in turn implies that (?)
%\begin{eqnarray}
%\mathcal{D}_{sa} \circ ( \mathcal{D}_{s}\otimes \mathcal{I}_{a}) =  \mathcal{D}_{sa}\circ ( \mathcal{I}_{s}\otimes \mathcal{D}_{a}) \label{DsaDsDa}.
%\end{eqnarray}
%\color{cyan}
%A similar result holds for effects.

We here assume that the preferred subspaces for the composite are just the tensor products of those for the system and for the ancilla, so that 
\begin{eqnarray}\label{DsaDstimesDa}
\mathcal{D}_{sa} = \mathcal{D}_{s}\otimes \mathcal{D}_{a}.
\end{eqnarray}

\begin{proposition}\label{Prop:StinespringForDC}
A quantum operation  $\mathcal{E}$ is {\em dephasing-covariant} 
if it can be implemented by coupling the system to an ancilla in a state $\sigma$ that is incoherent, via a unitary quantum operation $\mathcal{V}$ that is dephasing-covariant, and then post-selecting on a measurement outcome associated to an incoherent effect $E$.
 Suppose $\mathcal{E}$ can be implemented as a dilation of the form
\begin{equation}
\mathcal{E}(\rho) = {\rm tr}_{\rm a'} (E\; \mathcal{V}( \rho \otimes \sigma)),
\end{equation}
where $\rho$ is a state on $\mathcal{H}_{\rm s}$, $\sigma$ is a state on $\mathcal{H}_{\rm a}$, $\mathcal{V}(\cdot) =V (\cdot) V^{\dag}$ for some unitary operator $V$ on $\mathcal{H}_{\rm sa}$ , and $E$ is an effect on $\mathcal{H}_{\rm a'}$.  Then formally, $\mathcal{E}$ is dephasing-covariant if there is such a dilation where $\mathcal{D}_{\rm a}(\sigma)= \sigma$, $\mathcal{D}_{\rm a'}(E)= E$ and  $\mathcal{V}\circ \mathcal{D}_{\rm sa}= \mathcal{D}_{\rm sa}\circ \mathcal{V}$.
\end{proposition}

\begin{proof}
The proof is as follows:
\begin{align}
\mathcal{E}\big(\mathcal{D}_{\rm s}(\rho)\big)&=\text{tr}_{\rm a'} \Big( E\; \mathcal{V} [\mathcal{D}_{\rm s}(\rho)\otimes\sigma] \Big)\ \\ &=\text{tr}_{\rm a'}\Big(  E\; \mathcal{V}  [\mathcal{D}_{\rm sa}(\rho\otimes \sigma)] \Big)\\ 
%&=\text{tr}_{\rm a'} \Big( \mathcal{D}_{\rm sa} \Big( \mathcal{V} [\rho\otimes\sigma] \Big)\Big)\\ 
&=\text{tr}_{\rm a'} \Big( E\; \mathcal{D}_{\rm s'a'} \Big( \mathcal{V} [\rho\otimes\sigma] \Big)\Big)\\ 
&=\text{tr}_{\rm a'} \Big(  \mathcal{D}_{\rm s'a'} \Big( E\; \mathcal{V} [\rho\otimes\sigma] \Big)\Big)\\ 
&=\mathcal{D}_{\rm s'} \Big( \text{tr}_{\rm a'} \big( E\; \mathcal{V} [\rho \otimes\sigma] \big)\Big)\\ &=\mathcal{D}_{\rm s'}(\mathcal{E}(\rho)),
\end{align}
where in the second line, we have used the fact that $\mathcal{D}_{\rm a}(\sigma)= \sigma$ together with Eq.~\eqref{DsaDstimesDa};
%that fact that $\sigma$ is an incoherent  state together with the assumption about dephasing on composite systems noted in Eq.~\eqref{Dsa}
in the third line, we have used the fact that $\mathcal{V}$ is a dephasing-covariant operation; in the fourth line, we have used the fact that $\mathcal{D}_{\rm a'}(E)= E$ and Eq.~\eqref{DsaDstimesDa}; in the fifth line, we have used Eq.~\eqref{DsaDstimesDa} again; and in the sixth line, we have used Eq.~\eqref{int}.
% and the cyclic property of the trace.
\end{proof}
\color{black}

It is an open question whether \emph{every} dephasing-covariant operation on a system can be implemented in this fashion for a suitable choice of ancilla. 

\subsubsection{Characterization via diagonal and off-diagonal modes}

A useful way of distinguishing incoherent states from coherent states is by considering their representation as vectors in $\mathcal{B}$, the space of linear operators on $\mathcal{H}$. 
The dephasing operations $\mathcal{D}$ is a projector on the operator space,  i.e., it satisfies $\mathcal{D}\circ \mathcal{D}=\mathcal{D}$, and it  induces a direct sum decomposition on  the operator space $\mathcal{B}$ as $\mathcal{B} = \mathcal{B}^{\rm diag} \oplus \mathcal{B}^{\rm offd}$, where $ \mathcal{B}^{\rm diag}$ and $\mathcal{B}^{\rm offd}$ are respectively the image and the kernel of $\mathcal{D}$.
For an arbitrary operator $A$, we define the diagonal component of $A$ to be
\begin{equation}
A^{\rm diag} \equiv \mathcal{D}(A),
\end{equation}
and the off-diagonal component of $A$ to be
\begin{equation}
A^{\rm offd} \equiv A-\mathcal{D}(A)= A-A^{\rm diag}.
% (\mathcal{I}-\mathcal{D})(A),
\end{equation}
Clearly, all incoherent states lie entirely within $\mathcal{B}^{\rm diag}$, while every coherent state has some nontrivial component in $\mathcal{B}^{\rm offd}$.

Then, the fact that dephasing-covariant operations by definition commute with the dephasing operation $\mathcal{D}$, immediately implies that these operations are block-diagonal with respect to this decomposition of the operator space $\mathcal{B}$.

It is useful to consider how a dephasing-covariant operation $\mathcal{E}$ is represented as a matrix on the operator space $\mathcal{B}$.  
% The formal distinction between incoherence-preserving  operations and dephasing-covariant operations can be explained using their matrix representations: For a quantum operation $\mathcal{E}$ 
If $\{X_i\}_i$ is an orthonormal basis (with respect to the Hilbert Schmidt inner product) for the space of operators $\mathcal{B}$, then $\mathcal{E}$ can be represented by the matrix elements $\mathcal{E}_{ij}=\Tr(X^\dag_i\mathcal{E}(X_j))$. 
% Let $\mathcal{B}_\mathcal{I}$ be the subspace of the operator space $\mathcal{B}$ that is spanned by the set of incoherent states $\mathcal{I}$, and let  $\mathcal{B}^\perp_\mathcal{I}$ be the orthogonal subspace. In other words,  $\mathcal{B}^\perp_\mathcal{I}$  is  the subspace spanned by operators which have only off-diagonal elements relative to the standard basis.  
$\mathcal{E}$ is dephasing-covariant iff its matrix representation has the following form relative to the decomposition $\mathcal{B} = \mathcal{B}^{\rm diag} \oplus \mathcal{B}^{\rm offd}$,
\begin{align}\label{twobytwoDC}
 \bordermatrix{~ & \mathcal{B}^{\rm diag} & \mathcal{B}^{\rm offd} \cr
               \mathcal{B}^{\rm diag}  & A  & 0 \cr
               \mathcal{B}^{\rm offd} & 0 &  B \cr}\ ,
\end{align}
where ${A}$ and ${B}$ are matrices.\color{black}

%In other words, we have the following mode-based characterization.
Alternatively, the mode-based characterization of dephasing-covariant operations can be given as follows.
\begin{proposition}\label{modecharDCops}
A quantum operation  $\mathcal{E}$ is dephasing-covariant relative to a preferred set of subspaces if and only if it preserves the diagonal and off-diagonal modes.  %That is, 
% , that is, the diagonal component of the input state is mapped only to the diagonal component of the output state, and the off-diagonal component of the input state is mapped only to the off-diagonal component of the output state. 
Formally, the condition is that  whenever $\mathcal{E}(\rho)=\sigma$, we have $\mathcal{E}(\rho^{\rm diag})=\sigma^{\rm diag}$  and $\mathcal{E}(\rho^{\rm offd})=\sigma^{\rm offd}$ .
%\end{definition}
\end{proposition}

\begin{comment}
\subsubsection{Characterization via Kraus decomposition}

\color{red} [The following proposition is tentative; if we cannot prove it, we will leave it out of the article.] \color{black}
%\begin{definition}\label{defnKraus}
\begin{proposition}\label{defnKrausDC}
A quantum operation $\mathcal{E}$ is dephasing-covariant if and only if it admits of a Kraus decomposition of the form
\begin{equation}
\mathcal{E}(\cdot) = \sum_{X\in {\{\rm diag,offd}\}} \sum_{\alpha} K^{X}_{\alpha} (\cdot) (K^{X}_{\alpha})^{\dag},
\end{equation}
where the elements of the set $\{ K^{\rm diag}_{\alpha}\}_{\alpha}$  are diagonal operators and the elements of the set $\{ K^{\rm offd}_{\alpha}\}_{\alpha}$ are off-diagonal operators.
\end{proposition}

\begin{proof}
To see that every dephasing-covariant quantum operation has such a Kraus decomposition, [...]
To see that any quantum operation with such a Kraus decomposition is dephasing-covariant, [...]
%\color{red}
%[Iman can you provide the proofs?  I cannot see how to prove it at the moment.]
%\color{black}
\end{proof}
\end{comment}

%\subsection{Does the restriction to dephasing-covariant operations arise naturally?}
\subsection{Physical justification for the restriction to dephasing-covariant operations?}

%We just need to find a reason for the unitaries to be restricted to those that permute the basis elements and for the states to all be incoherent.    

%There are many physical circumstances in which the only freely-implementable preparations of the ancilla systems correspond to the incoherent states, such as if they are unavoidably subject to environmental decoherence that induces dephasing between the preferred subspaces. For experimental operations that are not preparations, however, it is less clear whether the restriction to dephasing-covariant operations can arise as a natural constraint.

We noted that whether \emph{every} dephasing-covariant operation admits of a dilation in terms of an incoherent ancilla state, an incoherent effect on the ancilla, and a dephasing-covariant unitary on the system-ancilla composite 
%(for some choice of the preferred subspaces on the composite)
%(either for collective or independent dephasing)
 is currently an open question.  
%If its answer turns out to be negative, then this would rule against understanding the dephasing-covariant operations as those arising from an experimental restriction
%counts as evidence against there being any natural constraint on experimental operations associated to the dephasing-covariant operations 
%(this is the situation for the BCP approach, as we will see in Sec.~\ref{BCP}).  

If its answer is positive, then the problem of finding a physical justification for the dephasing-covariant operations reduces to finding a physical scenario wherein the only free states and effects on the ancilla are incoherent and the only free unitaries on the system-ancilla composite are those that are dephasing-covariant.  It is not obvious how to justify the latter constraint in particular. However, even if the answer is negative, there remains the possibility that one can find a physical justification for the free set of operations being the set of dephasing-covariant operations.  This is the same possibility that remained for justifying the incoherence-preserving or incoherent operations as the free set, namely, by allowing that different systems may not be treated even-handedly.

%On the other hand, even if its answer is negative, there might still be a physical justification of the set of dephasing-covariant operations as the free set.   The challenge is to find {\em some} restriction on the system-ancilla couplings and the states and measurements  on  the ancillas such that one obtains all and only the dephasing-covariant operations on the system.

Overall, therefore, it is at present unclear whether a restriction to dephasing-covariant operations arises from a natural experimental restriction.

%\subsection{Translational covariance versus dephasing covariance}
%asymmetry versus coherence as dephasing noncovariance} 
%\subsection{Relation to unspeakable notion of coherence}
\subsection{Relation of dephasing-covariant operations to translationally-covariant operations}
\subsubsection{Relation between the sets of free operations}

 Here we study the relation between the dephasing-covariant operations and the translationally-covariant operations {\em for the same choice of the preferred subspaces.}  In the dephasing-covariance approach to coherence, one must begin with a choice of preferred subspaces relative to which dephasing occurs.  If one is given a translational symmetry, then one can choose these subspaces to be the eigenspaces of the generator of that symmetry group (or, if the symmetry group incorporates independent  commuting  translations and therefore multiple commuting  generators, then as the joint eigenspaces of the generators).  Conversely, if a set of preferred subspaces is given, one can always construct a Hermitian operator that has these subspaces as eigenspaces with distinct eigenvalues and consider this to be a generator of translations. 
Because a given choice of preferred subspaces might only be physically justified in {\em one} of the two approaches, the comparison we are making here is best understood as probing the mathematical relation between the two approaches to quantifying coherence as a resource.
 
% For any given translational symmetry and its corresponding generator(s) one can consider  the set of dephasing-covariant operations,  defined based on the decomposition of the Hilbert space to the eigen-subspaces of these generator(s). 

To understand this connection it is useful to note that  the dephasing operation relative to the eigenspaces of $L$ 
%denoted $\mathcal{D}_L$ and defined by $\mathcal{D}_L \equiv \sum_l \Pi_l (\cdot) \Pi_l$ in Eq.~\eqref{Dephasing}, 
can be realized by applying a random translation %$e^{-i x L}$
 to the system, that is, a translation $\mathcal{U}_{L,x}$ where $x$ is chosen uniformly at random,
\beq\label{int}
\mathcal{D}(\cdot)=
%\equiv\sum_{i=1}^M P_i (\cdot) P_i=
\lim_{x_0\rightarrow \infty} \frac{1}{2x_0}\ \int_{-x_0}^{x_0}  dx\ \  e^{-i x L}(\cdot)  e^{i x L}\ .
\eeq
%This identity will be used in several of our proofs.

\begin{comment}
\begin{proposition}\label{prop1}
On an elementary system, consider the set of quantum operations that are translationally-covariant with respect to a generator $L$, 
\beq
TC_{L} \equiv \{ \mathcal{E} : \forall x \in \mathbb{R}, \;\mathcal{U}_{L,x} \circ \mathcal{E} = \mathcal{E} \circ \mathcal{U}_{L,x}  \}
\eeq
%are a proper subset of
and the set of quantum operations that are dephasing-covariant with respect to the eigenspaces of $L$,
\beq
DC_{L} \equiv \{ \mathcal{E} :  \mathcal{D}_{L} \circ \mathcal{E} = \mathcal{E} \circ \mathcal{D}_{L}  \}
\eeq
The former is a proper subset of the latter,
\beq
\text{TC}_{L}\  \subset  \ \text{DC}_{L}.
\eeq
\end{proposition}
\end{comment}

It is also useful to note the connection between the two approaches from the perspective of mode decompositions. 
The diagonal mode relative to the eigenspaces of $L$ corresponds to the $\omega=0$ mode of translational asymmetry relative to the generator $L$, 
\beq
\mathcal{B}^{\rm diag} = \mathcal{B}^{(0)},
\eeq
while the off-diagonal mode relative to the eigenspaces of $L$ corresponds to the direct sum of the $\omega\ne 0$ modes of translational asymmetry relative to the generator $L$,
\beq\label{relateBoffdtoBk}
\mathcal{B}^{\rm offd} = \bigoplus_{\omega\ne 0} \mathcal{B}^{(\omega)}.
\eeq
Intuitively then, to  choose the dephasing-covariant operations as the free set of operations is to disregard the distinction between the different nonzero modes.

\color{black}

Denote the the set of quantum operations that are translationally-covariant with respect to a generator $L$ by TC$_L$ and the set of quantum operations that are dephasing-covariant with respect to the eigenspaces of $L$ (which we will denote $\mathcal{S}(L)$) 
 by DC$_{\mathcal{S}(L)}$. \color{black}

\begin{proposition}\label{prop1}
The operations that are translationally-covariant relative to translations generated by $L$ are a proper subset of those that are dephasing-covariant relative to the eigenspaces of $L$,
\beq
\text{TC}_{L}\  \subset  \ \text{DC}_{\mathcal{S}(L)}.
\eeq
\end{proposition}

\begin{proof}
%We can see the distinction through the lens of any of the four characterizations we have provided. 
The subset relation can be understood easily within any of the characterizations of translationally-covariant and dephasing-covariant operations that we have provided. 
%while the fact that it is a proper subset is most easily seen using the definition in terms of the decomposition into diagonal and off-diagonal modes. 
%The most straightforward way of proving the inclusion relation---that any translationally-covariant  operation is also a dephasing-covariant operation---is to start 
For instance, starting with the expression for the translational-covariance of an operation $\mathcal{E}$, Eq.~(\ref{TI_opo}), if one integrates over $x$, one obtains the expression for the dephasing covariance of $\mathcal{E}$, Eq.~(\ref{DefDC}), where we have made use of Eq.~\eqref{int}.

%According to the Stinespring definition, the translationally-covariant unitary operations are a subset of the dephasing-covariant unitary operations.

%According to the mode decomposition.  [blah blah]

To show that the inclusion is strict, it suffices to show that there are dephasing-covariant operations which are not translationally-covariant.  Any example of a dephasing-covariant operation wherein one nonzero mode, $\omega$, is mapped to another, distinct, nonzero mode is sufficient. 

For example, consider the unitary that swaps a pair of states living in
different eigenspaces of $L$, 
%sectors of the preferred decomposition of the Hilbert space
and leaves the rest of the states unchanged. This operation in general will not be translationally-covariant while it is dephasing-covariant. 
% \color{red} [Provide more detail here] \color{black}
%(We discuss such examples further in Sec. \ref{Sec:examples})
\end{proof}

\begin{figure}[htb]
 \includegraphics[scale=0.35]{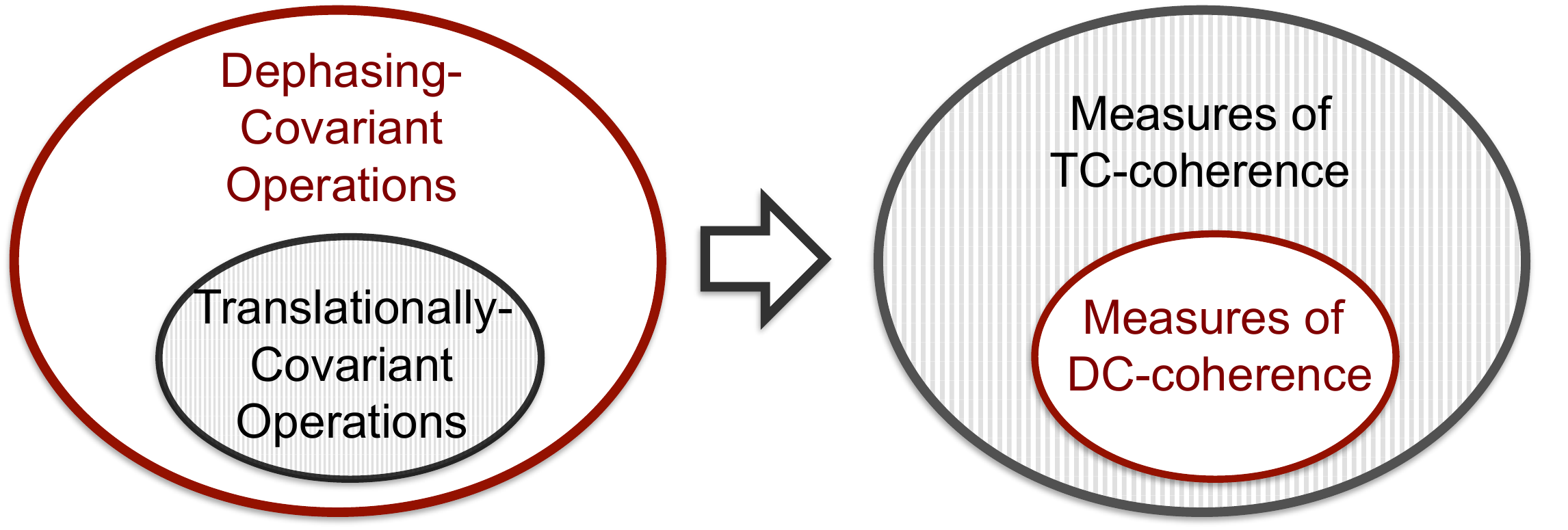}
 \caption{The relation between the dephasing-covariant and the translationally-covariant operations and the relation between the associated measures of coherence that these sets of operations define.  Both inclusions are shown to be strict.}
 \label{Venn1}
\end{figure}

\color{black}
Despite the strict inclusion of translationally-covariant operations in the set of dephasing-covariant operations, if we focus on the POVMs associated to measurements (i.e. the retrodictive aspect of the measurement) the two approaches pick out the same set, as we noted earlier.
\color{black}

\color{red}
%[Placeholder: discuss the relation between dephasing-covariant unitaries and translationally-covariant unitaries?  Our proof of strictness of the inclusion basically appeals to the distinction among unitaries, so it is probably worth spelling out earlier what the unitaries are.]
\color{black}

\subsubsection{Relation between measures of coherence}

%\color{red} [We need some simplifying terminology for ``measure of coherence in the translational-covariance approach'' and ``measure of coherence in the dephasing-covariance approach''.  The former could be ``measure of translational asymmetry'', but I'd rather stick to the language of coherence in this paper.  I propose ``measure of TC-coherence'' and ``measure of DC-coherence''.  Similarly, "coherence in the incoherence-preserving approach'' can be simply ``IP-coherence''.]\color{black}

%If a transformation from initial state $\rho$ to final state $\sigma$ is allowed under dephasing-covariant operations, then it is also allowed under  translationally-covariant operations because, by proposition \ref{prop1}, the latter set includes the former.  
%We will refer to measures of coherence that arise from a particular set $S$ of free operations as measures of $S$-coherence. 
Prop.~\ref{prop1} implies that if a transformation from initial state $\rho$ to final state $\sigma$ is allowed under $TC_L$ operations, then it is also allowed under $DC_{\mathcal{S}(L)}$  operations. This means that any measure of $DC_{\mathcal{S}(L)}$-coherence  is  also a measure of $TC_L$-coherence.  In fact, one can show that
%, but, as we will show, not every measure of $TC_L$-coherence is a measure of $DC_{\mathcal{S}(L)}$-coherence. 
%To summarize
\begin{proposition}\label{prop11}
%For an elementary system, 
Any measure of $DC_{\mathcal{S}(L)}$-coherence  is also a measure of $TC_L$-coherence, but not vice-versa. 
%For an elementary system, any measure of coherence in the dephasing-covariance approach (relative to the eigenspaces of an observable $L$) is also a measure of coherence in the translational-covariance approach (relative to translations generated by $L$), but not vice-versa. 
\end{proposition}

The strictness of the inclusion is demonstrated in Sec.~\ref{Sec:examples}.  This relation is illustrated in Fig.~\ref{Venn1}.

%\color{red} [Improve the below with discussion of modes]\color{black}

\section{Coherence via incoherence-preserving  and incoherent operations}\label{BCP}

In this section, we consider approaches to coherence wherein the free operations are incoherence-preserving or incoherent operations.

 \color{black} 

%\section{speakable coherence according to the BCP proposal}

%\subsection{Definition of free operations:  Incoherent operations}
\subsection{Free operations as incoherence-preserving operations}

%\color{red} [Revise the discussion below to acknowledge that there is a difference between BCP's ``incoherent operations" and incoherence-preserving operations] \color{black}

%In this section, we consider an approach to coherence wherein the free operations are those that are incoherence-preserving, defined as follows:
%We now turn to the BCP approach to coherence.  We begin by defining the set of incoherence-preserving operations, as follows:
%\begin{definition}
%A quantum operation $\mathcal{E}$ is said to be  \emph{incoherence-preserving} if it maps  incoherent states to incoherent states, 
% \begin{equation}\label{incoherent}
% \rho\in \mathcal{I} \ \Longrightarrow\  \mathcal{E}(\rho)\in\mathcal{I} . 
% \end{equation} 
% \end{definition}
 \begin{definition}
A quantum operation $\mathcal{E}$ is said to be  \emph{incoherence-preserving} if it maps  incoherent states on the input space to incoherent states on the output space, 
 \begin{equation}\label{incoherent}
 \rho\in \mathcal{I}_{\rm in} \ \Longrightarrow\  \mathcal{E}(\rho)\in\mathcal{I}_{\rm out}. 
 \end{equation} 
 \end{definition}
 
% BCP only identify the free operations set  do not include those wherein the output space is different from the input space.  Furthermore, 

%The BCP proposal for the set of free operations are those that admit of a Kraus decomposition for which each term is incoherence-preserving. 
% They term these the {\em incoherent operations}. 

%The preferred subspaces relative to which coherence is assessed are presumed to be 1-dimensional, so that we refer below to the preferred \emph{basis} rather than the preferred subspaces. The set of free operations, which BCP term \emph{incoherent operations}, are those that are incoherence-preserving, defined as follows:

%\color{red}
%The incoherence-preserving operations obviously satisfy properties \ref{property1} and \ref{property2} of Sec.~\ref{constraintsonfreeops}.
%Indeed, in this proposal, one simply allows \emph{all} operations that satisfy property \ref{property2}.  
%\color{black}
%Furthermore, if we consider the case of a trivial input space, so that $\mathcal{E}$ simply prepares a state on the output space, i.e. $\mathcal{E}(\cdot) = \sigma$, then Eq.~\eqref{incoherent} simply stipulates that $\sigma\in \mathcal{I}$, so property \ref{property3} holds as well.
 %The set of incoherent operations is closed under composition of operations, and convex combinations. 

Just as was the case with the dephasing-covariant operations, the incoherence-preserving operations can be characterized in terms of their interaction with the dephasing map:
% or in terms of their action on the block-diagonal and off-block-diagonal modes:
\begin{proposition}\label{charincops}
A quantum operation $\mathcal{E}$ is incoherence-preserving 
if and only if
\beq\label{Def2}
\mathcal{E}\circ\mathcal{D}=\mathcal{D}\circ \mathcal{E}\circ\mathcal{D}\ .
\eeq
%and  if and only if the off-diagonal component of the output state is only a function of the off-diagonal component of the input state. Thus, whenever $\mathcal{E}(\rho)=\sigma$, we have $\mathcal{E}(\rho^{\rm offd})=\sigma^{\rm offd}$.
\end{proposition}

We can also characterize incoherence-preserving operations in terms of their representations as matrices on the operator space $\mathcal{B}$, just as we did for dephasing-covariant operations.  

%The formal distinction between incoherence-preserving  operations and dephasing-covariant operations can be explained using their matrix representations: For a quantum operation $\mathcal{E}$ consider the matrix representation $\mathcal{E}_{ij}=\Tr(X^\dag_i\mathcal{E}(X_j))$, where $\{X_i\}_i$ is an orthonormal basis for the space of operators $\mathcal{B}$ (with respect to the Hilbert Schmidt decompostion).  Let $\mathcal{B}_\mathcal{I}$ be the subspace of the operator space $\mathcal{B}$ that is spanned by the set of incoherent states $\mathcal{I}$, and let  $\mathcal{B}^\perp_\mathcal{I}$ be the orthogonal subspace. In other words,  $\mathcal{B}^\perp_\mathcal{I}$  is  the subspace spanned by operators which have only off-diagonal elements relative to the standard basis.  Then, a quantum operation is dephasing-covariant iff its matrix representation has the following form in this basis,
%\begin{align}
%\text{DC operation}: \ \ \ \ \ \ \bordermatrix{~ & \mathcal{B}_\mathcal{I} & \mathcal{B}^\perp_\mathcal{I} \cr
%               \mathcal{B}_\mathcal{I}  & A  & 0 \cr
%               \mathcal{B}_\mathcal{I}^\perp & 0 &  B \cr}\ ,
%\end{align}
%where ${A}$ and ${B}$ are matrices. 

%On the other hand, a quantum 
We deduce from Eq.~\eqref{Def2} that an operation $\mathcal{E}$ is incoherence-preserving  if and only if its matrix representation has the following form relative to the decomposition $\mathcal{B} = \mathcal{B}^{\rm diag} \oplus \mathcal{B}^{\rm offd}$, 
\begin{align}
 \bordermatrix{~ & \mathcal{B}^{\rm diag} & \mathcal{B}^{\rm offd} \cr
               \mathcal{B}^{\rm diag}  & A  & C \cr
               \mathcal{B}^{\rm offd} & 0 &  B \cr}\ ,
\end{align}
where ${A}$, ${B}$ and $C$ are matrices.  A comparison with the analogous characterization of dephasing-covariant operations, Eq.~\eqref{twobytwoDC}, shows how incoherence-preserving operations do not preserve the diagonal and off-diagonal modes. 
%\color{cyan} 
%In terms of the diagonal and off-diagonal modes, Eq.~\eqref{Def2} merely implies that if $\mathcal{E}$ is incoherence-preserving, then $\mathcal{E}(\rho^{\rm diag}) \in \mathcal{B}^{\rm diag}$.
% operations are such that whenever $\mathcal{E}(\rho)=\sigma$, we have $\mathcal{E}(\rho^{\rm diag})=\sigma^{\rm diag}$.  The diagonal mode is preserved.  The off-diagonal mode, however, is generally {\em not} preserved by incoherence-preserving operations.
%the off-diagonal component of the output state is only a function of the off-diagonal component of the input state.
%\color{red} [Iman to reincorporate the $2\times2$ matrix picture.] \color{black}

We postpone our characterization of the incoherence-preserving {\em measurements} until Sec.~\ref{Sec:CriticismIP} because it forms the basis of one of our criticisms of this approach.

\begin{comment}
Just as was the case with the dephasing-covariant operations, the incoherence-preserving operations can be characterized in terms of their interaction with the dephasing map.
\begin{proposition}\label{dephasingcharincops}
A quantum operation $\mathcal{E}$ is incoherence-preserving 
if and only if
\beq\label{Def2}
\mathcal{E}\circ\mathcal{D}=\mathcal{D}\circ \mathcal{E}\circ\mathcal{D}\ .
\eeq
%In other words, $\mathcal{E}$ is incoherence-preserving  if and only if the off-diagonal component of the output state is only a function of the off-diagonal component of the input state. Thus, whenever $\mathcal{E}(\rho)=\sigma$, we have $\mathcal{E}(\rho^{\rm offd})=\sigma^{\rm offd}$.
\end{proposition}

We can also characterize the incoherence-preserving operations in terms of their action on the block-diagonal and off-block-diagonal modes as follows.
\begin{proposition}\label{modecharincops}
A quantum operation  $\mathcal{E}$ is incoherence-preserving if and only if the off-block-diagonal component of the output state is only a function of the off-block-diagonal component of the input state. Thus, whenever $\mathcal{E}(\rho)=\sigma$, we have $\mathcal{E}(\rho^{\rm offd})=\sigma^{\rm offd}$.
\end{proposition}
\end{comment}

\subsubsection{Incoherence-preserving unitary operations}\label{Sec:unitunit}
%\subsection{Unitary incoherent operations and unitary incoherence-preserving operations}\label{Sec:unitunit}

For simplicity, we start with the special case where the preferred subspaces are all 1-dimensional, where the unitary incoherence-preserving operations have a particularly simple form. 

Let $\mathcal{V}$  denote a unitary incoherence-preserving operation, and let $\{|l\rangle\langle l|\}_l$ denote the set of projectors onto the elements of the preferred basis. Consider the image of each $|l\rangle\langle l|$ under $\mathcal{V}$. Because a unitary operation preserves the rank of a state, the image must also be a projector onto a 1-dimensional subspace.   But given that $|l\rangle\langle l|$ is an incoherent state and $\mathcal{V}$ is incoherence-preserving, it follows that the image must be an incoherent pure state.  The only incoherent pure states are those in the set $\{|l\rangle\langle l|\}_l$, therefore $\mathcal{V}$ must map this set to itself, that is, for any $l$ it should hold that
\beq
\mathcal{V} (|l\rangle\langle l |)=|\pi(l)\rangle\langle  \pi(l) |
\eeq
where $\pi$ is a permutation over the set $\{l  \}_l$.

If $V$ is the unitary operator that defines the unitary incoherence-preserving operation $\mathcal{V}$  through $\mathcal{V}(\cdot)=V\cdot V^\dag$, the incoherence-preserving property implies that 
\beq\label{unit34}
V=  \sum_l e^{i\theta_l}|\pi(l)\rangle\langle  l |
\eeq
for a set of phases $\{e^{i\theta_l}\}_l$. So, any incoherence-preserving unitary operation can be characterized by a permutation $\pi$ of the preferred basis, and a set of phases $\{e^{i\theta_l}\}_l$.

The case where the preferred subspaces are not all 1-dimensional is slightly more complicated. 

 Let $\{ |l,\alpha_l\rangle \}_{\alpha_l}$  denote an arbitrary basis for the preferred subspace $\mathcal{H}_l$, so that $\{ |l,\alpha_l\rangle \}_{l,\alpha_l}$  is a basis of the entire Hilbert space.  Now consider the image of $\{ |l,\alpha_l\rangle \}_{\alpha_l}$ under the unitary $V$.  Although each vector $ |l,\alpha_l\rangle \in \mathcal{H}_l$ need not be mapped to another vector in $\mathcal{H}_l$, 
 %Now consider the image of $\{ |l,\alpha_l\rangle \langle l, \alpha_l | \}_{\alpha_l}$ under the unitary incoherence-preserving operation $\mathcal{V}$.  By unitarity, each such rank-1 projector must be mapped to another rank-1 projector, and although $\mathcal{V}$ need not preserve the preferred subspace in which it lies, 
it is still the case that for a given $l$, there must be some $l'$ such that for  {\em every} vector in $\mathcal{H}_l$,  its image is in $\mathcal{H}_{l'}$.  The reason is that if this were not the case, it would be possible to identify some vector in $\mathcal{H}_l$ 
%coherent superposition of the $\{ |l,\alpha_l\rangle \}_{\alpha_l}$
 that is mapped to a nontrivial coherent superposition of vectors lying in {\em different preferred subspaces}, and this would imply a violation of the incoherence-preserving property.  Note that the dimension of $\mathcal{H}_{l'}$ must be the same as the dimension of $\mathcal{H}_{l}$.

Therefore, if $\pi$ denotes a dimension-preserving permutation of the preferred subspaces, and $V_l$ denotes a unitary that acts arbitrarily {\em within} the $\mathcal{H}_l$ subspace and as identity on the complementary subspace, then the property of $\mathcal{V}$ being incoherence-preserving implies that
%It follows that we can write the constraint on $V$ as
\beq\label{unit35}
V=  \sum_{l,\alpha_l} V_{\pi(l)} |\pi(l), \alpha_{\pi(l)} \rangle \langle l,\alpha_{l} |.
\eeq

So, any incoherence-preserving unitary operation can be characterized by a dimension-preserving permutation $\pi$ among the preferred subspaces, and a set of unitary operators $\{ V_{l}\}_l$.
\color{black}

\subsection{Relation of incoherence-preserving operations to dephasing-covariant operations}\label{relationIPDC}

We begin by considering unitary operations.

\begin{proposition}\label{coincidenceunitaries}
The set of unitary dephasing-covariant operations relative to the preferred subpaces $\{ \mathcal{H}_l \}_l$ is equivalent to the set of unitary incoherence-preserving operations relative to these same subspaces. 
\end{proposition}

The proof is as follows.  Every dephasing-covariant operation is incoherence-preserving and therefore what must be demonstrated is that for unitary operations, being incoherence-preserving implies being dephasing-covariant. This simply follows from the general form of an incoherence-preserving unitary operation, Eq.~\eqref{unit35}.
%\color{red} [Iman: Provide proof based on demonstrating that my characterization of the most general IP unitary is indeed dephasing-covariant]

In general, however, the dephasing-covariant operations are a strict subset of the incoherence-preserving operations.
\color{black}

%In the next subsection, we shall see that this requirement is nontrivial, which is to say that the dephasing-covariant operations are a strict subset of the incoherence-preserving operations.

%The following example clarifies this difference:
To demonstrate this, we provide a simple example of an operation that is incoherence-preserving but not dephasing-covariant. 
 Consider a qubit and denote the preferred basis thereof by $\{|0\rangle,|1\rangle\}$. Let $|\pm\rangle = 2^{-1/2} ( |0\rangle \pm |1\rangle )$.  Consider the quantum operation defined by
\beq\label{examp}
%\mathcal{E}(\rho)=|0\rangle\langle 0| \langle+|\rho|+\rangle + |1\rangle\langle 1| \langle-|\rho|-\rangle,
\mathcal{E}(\rho)=|0\rangle\langle 0| {\rm Tr} (|+\rangle\langle+|\rho) + |1\rangle\langle 1| {\rm Tr} (|-\rangle\langle-|\rho),
\eeq
 $\mathcal{E}$ is clearly incoherence-preserving, because for all input states (and therefore, in particular, incoherent states), the output is always an incoherent state. On the other hand, one can easily show that it is not dephasing-covariant because $\mathcal{D} \circ \mathcal{E} = \mathcal{E}$, while $\mathcal{E} \circ \mathcal{D} \ne \mathcal{E}$.  This can be seen, for instance, by noting that while $ \mathcal{D}(|+\rangle\langle+|)= \mathcal{D}(|-\rangle\langle-|)$, operation  $\mathcal{E}$ maps  $|+\rangle$ and $|-\rangle$ to two different states.

This example is related to the fact  if one considers operations with trivial output spaces, i.e., destructive measurements, then there are incoherence-preserving operations which are not dephasing-covariant (See Sec.~\ref{Sec:CriticismIP}).

 \begin{comment}
 The latter fact is seen as follows.  First, note that
\begin{eqnarray}
 \mathcal{E} \circ \mathcal{D} (\cdot) &=&|0\rangle\langle 0| {\rm Tr} (|+\rangle\langle+|\mathcal{D}(\cdot)) + |1\rangle\langle 1| {\rm Tr} (|-\rangle\langle-|\mathcal{D}(\cdot)),\nonumber\\
 &=&|0\rangle\langle 0| {\rm Tr} (\mathcal{D}(|+\rangle\langle+|)\cdot) + |1\rangle\langle 1| {\rm Tr} (\mathcal{D}(|-\rangle\langle-|)\cdot),\nonumber\\
 &=& \frac{1}{2}  I {\rm Tr} (\cdot)
 %\nonumber\\
% &\ne& \mathcal{E}(\rho),
 \end{eqnarray}
 where in the second line we have used the fact that $\mathcal{D}$ is self-adjoint relative to the Hilbert-Schmidt inner product. 
 \end{comment}
 
  It then suffices to note that there are $\rho$ such that $\mathcal{E} \circ \mathcal{D}(\rho) \ne \mathcal{E}(\rho)$.  This occurs whenever $\rho$ is such that ${\rm Tr} (|+\rangle\langle+|\rho) \ne  {\rm Tr} (|-\rangle\langle-|\rho)$.

%BCP characterize the incoherent operations in terms of their Kraus decomposition.

%\subsection{Relation of BCP proposal to our proposal}

Suppose, as we have done before, that the operator $L$ has the preferred subspaces as its eigenspaces.  Denote the operations that are incoherence-preserving relative to the eigenspaces of $L$ by $IP_L$.   Then we have the following result.

\begin{proposition}\label{propIncDC}
The set of quantum operations that are dephasing-covariant relative to the eigenspace of $L$ are a strict subset of the set of quantum operations that are incoherence-preserving relative to the eigenspaces of $L$,
\beq
\text{DC}_{L}\  \subset  \ \text{IP}_L.
\eeq
\end{proposition}

%\begin{proposition}\label{propIncDC}
%The set of quantum operations that are dephasing-covariant with respect to a preferred basis, denoted $DC$, are a proper subset of the set of incoherent quantum operations relative to that preferred basis, denoted $Inc$,
%\beq
%\text{DC}\  \subset  \ \text{Inc}.
%\eeq
%\end{proposition}

\begin{figure}[htb]
 \includegraphics[scale=0.35]{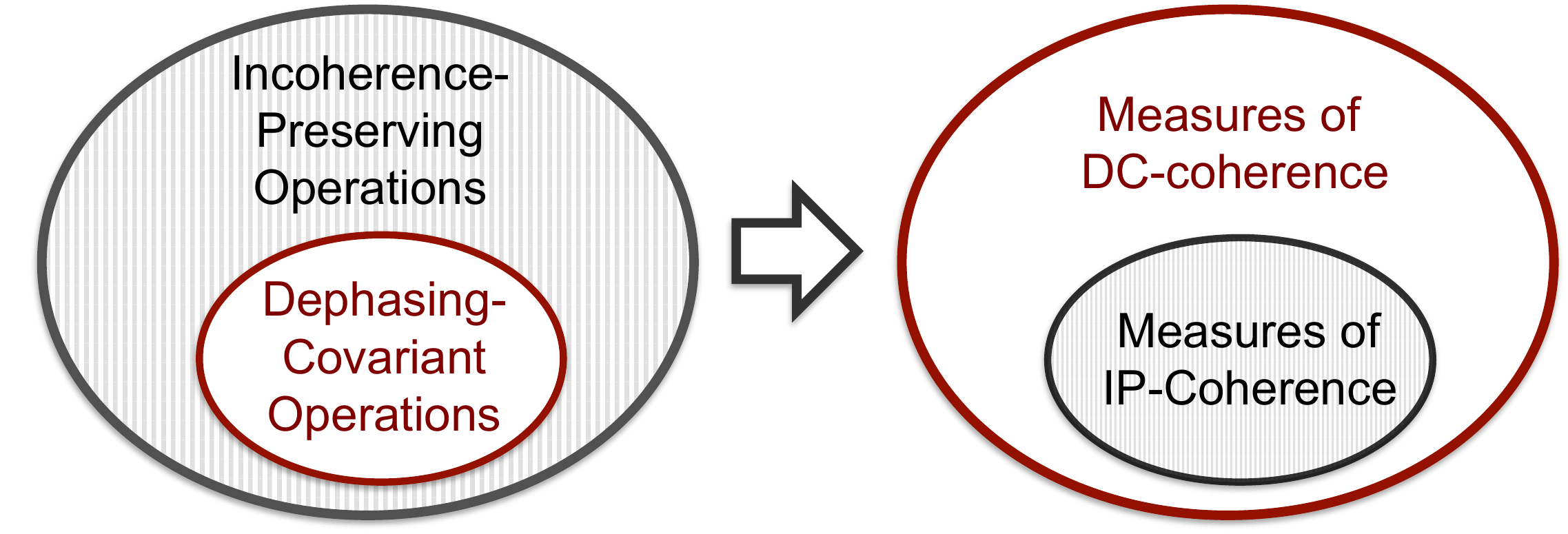}
 \caption{The relation between the incoherence-preserving and the dephasing-covariant operations and the relation between the associated measures of coherence that these sets of operations define.  It is not known whether the inclusion relation among the measures is strict.}
 \label{Venn2}
\end{figure}

\begin{proof}
%\color{red} [Lead with mode-based analysis] \color{black}
The discussion surrounding Eq.~\eqref{IPofDC} showed that any dephasing-covariant operation $\mathcal{E}$ is incoherence-perserving (and similar proofs could be given in terms of the alternative characterizations of the dephasing-covariant operations).

%For any dephasing-covariant operation $\mathcal{E}$, we have, by definition, $\mathcal{E}\circ\mathcal{D}=\mathcal{D}\circ \mathcal{E}$\ . This implies
%\beq 
%\mathcal{D}\circ\mathcal{E}\circ\mathcal{D}=\mathcal{E}\circ\mathcal{D}^2=\mathcal{E}\circ\mathcal{D},
%\eeq
%where we have used the idempotence of $\mathcal{D}$, which, given Eq.~\eqref{Def2}, implies that $\mathcal{E}$ is an incoherence-preserving operation. 

To see that the dephasing-covariant operations are a \emph{proper} subset of the incoherent operations, we provide an example of an operation that is incoherence-preserving but not dephasing-covariant.  
\end{proof}

\subsection{Free operations as incoherent operations}

The approach to coherence described above is certainly in the same spirit as the approach introduced by BCP.
 %\footnote{\label{footBCP} 
 Strictly speaking, however, BCP take the free operations to be the {\em incoherent operations}, defined as follows:
  \begin{definition}\label{footBCP} 
A quantum operation $\mathcal{E}$ is said to be  \emph{incoherent} if it admits of a Kraus decomposition where each term is incoherence-preserving, that is, if there is a decomposition with Kraus operators $\{K_n\}_n$ such that $ K_n \mathcal{I}_{\rm in} K^{\dag}_n \subset \mathcal{I}_{\rm out}$ for all $n$.
\end{definition}

In their article, BCP only considered coherence relative to a decomposition of the Hilbert space into 1-dimensional subspaces, while one may want to consider decompositions into subspaces of arbitrary dimension.  Also, BCP did not explicitly consider the possibility that the input and output spaces of $\mathcal{E}$ are different. The definition of incoherent operations we have just provided incorporates these two generalizations of the notion. 

Most importantly, however, most of our discussion will focus on the set of incoherence-preserving operations as the free set rather than the incoherent operations. 
%  Our definition \eqref{incoherent} again constitutes the natural generalization of their notion to this case. 
 As it turns out, this difference will not be relevant for any of our criticisms of the BCP approach.
%but they seem to have the index of the Kraus operator be the outcome of the measurement, whereas in the general case a single outcome might be associated with many Kraus operators. )

% BCP propose that the set of free operations defining the resource theory of coherence should be those that are incoherence-preserving.
\begin{comment}
We note that these operations would be best described as \emph{incoherence-preserving} rather than \emph{incoherent}, but we shall stick to the former terminology, as this has become standard.
\end{comment}

%Suppose, as we have done before, that the operator $L$ has the preferred subspaces as its eigenspaces.  
%Denote the operations that are incoherence-preserving relative to the eigenspaces of $L$ by $IP_L$.  
Denote the operations that are incoherence-preserving (dephasing-covariant) relative to a particular choice $\mathcal{S}\equiv \{ \mathcal{H}_l\}_l$ of the preferred subspaces by $IP_{\mathcal{S}}$  ($DC_{\mathcal{S}}$).  
%$\{ \mathcal{H}_l\}_l$ by $IP_{\{ \mathcal{H}_l\}_l}$  ($DC_{\{ \mathcal{H}_l\}_l}$).  
 Then we have the following result.
\begin{proposition}\label{propIncDC}
%The set of quantum operations that are dephasing-covariant relative to the eigenspace of $L$ are a strict subset of the set of quantum operations that are incoherence-preserving relative to the eigenspaces of $L$,
\beq
\text{DC}_{\mathcal{S}}  \subset  \ \text{IP}_{\mathcal{S}}.
\eeq
Therefore, any measure of IP$_{\mathcal{S}}$-coherence is also a measure of DC$_{\mathcal{S}}$-coherence (See Fig.\ref{Venn2}).
\end{proposition}\color{black}

\begin{figure}[htb]
 \includegraphics[scale=0.35]{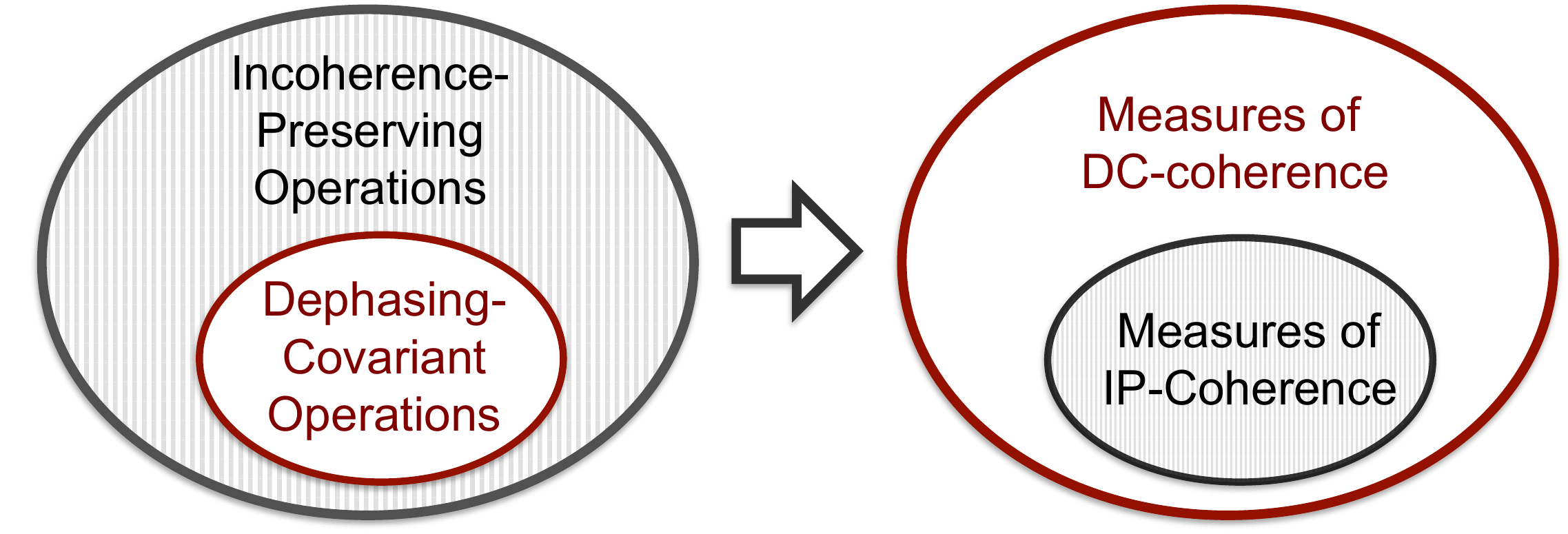}
 \caption{The relation between the incoherence-preserving and the dephasing-covariant operations and the relation between the associated measures of coherence that these sets of operations define.  It is not known whether the inclusion relation among the measures is strict.}
 \label{Venn2}
\end{figure}

Here the second statement follows from the fact that if $\rho \to \sigma$ is a transformation that is possible under dephasing-covariant operations, then it is also possible under incoherence-preserving  operations because the latter set includes the former. 
This, in turn, implies that any function on states that is nonincreasing under the latter set of operations is nonincreasing under the former set of operations.
%, hence any measure of IP$_L$-coherence is also a measure of DC$_L$-coherence.
% under the proposal wherein the free operations are the dephasing-covariant operations, i.e. any measure of dephasing noninvariance. 

\color{black}
%\begin{proposition}\label{prop11}
%For an elementary system, 
%Any measure of IP$_L$-coherence is also a measure of DC$_{L}$-coherence, but not vice-versa. 
%Any measure of IP$_L$-coherence is also a measure of DC$_{L}$-coherence.
%\end{proposition}

At present, it is not clear whether the vice-versa holds, that is, whether or not there exist measures of DC$_{\mathcal{S}}$-coherence that are not measures of IP$_{\mathcal{S}}$ coherence.  This question remains open because although we have examples of operations that are incoherence-preserving but not dephasing-covariant, we do not have examples of state transformations $\rho \to \sigma$ which are possible under incoherence-preserving operations but not dephasing-covariant ones.  That is, we have no proof that the additional operations in the incoherence-preserving set are in fact helpful in any state-conversion problem.  The question is also open if one considers measures of coherence under the incoherent operations rather than the incoherence-preserving operations.
\color{black}

%In Sec.~\ref{Sec:modes}, we demonstrate the strictness of the inclusion by providing an explicit example of  a measure of IP$_L$-coherence that is not a measure of DC$_L$-coherence.  The relation between the two types of measures is illustrated in Fig.~\ref{Venn2}.

%Again, the strictness of the inclusion also holds even if we consider measures of coherence relative to the incoherent operations of BCP rather than the incoherence-preserving operations.

\subsection{Free operations as incoherent operations}

The approach to coherence described above is certainly in the same spirit as the approach introduced by BCP.
 %\footnote{\label{footBCP} 
 Strictly speaking, however, BCP take the free operations to be the {\em incoherent operations}, defined as follows:
  \begin{definition}\label{footBCP} 
A quantum operation $\mathcal{E}$ is said to be  \emph{incoherent} if it admits of a Kraus decomposition where each term is incoherence-preserving, that is, if there is a decomposition with Kraus operators $\{K_n\}_n$ such that $ K_n \mathcal{I}_{\rm in} K^{\dag}_n \subset \mathcal{I}_{\rm out}$ for all $n$.
\end{definition}

In their article, BCP only considered coherence relative to a decomposition of the Hilbert space into 1-dimensional subspaces, while one may want to consider decompositions into subspaces of arbitrary dimension.  Also, BCP did not explicitly consider the possibility that the input and output spaces of $\mathcal{E}$ are different. The definition of incoherent operations we have just provided incorporates these two generalizations of the notion. 

Note that because unitary operations have only a single term in their Kraus decomposition, for unitary operations being incoherence-preserving coincides with being incoherent.  

%Note also that even if one restricts attention to the incoherent operations rather than the set of incoherence-preserving operations, the dephasing-covariant operations are still a strict subset thereof.  The reason is that our example of an operation that was incoherence-preserving but not dephasing-covariant, operation $\mathcal{E}$ in Eq.(\ref{examp}), is itself an incoherent operation, and therefore also establishes that there are incoherent operations that are not dephasing-covariant (See also \cite{yang2015optimal}).

%Consequently, every measure of coherence according to the BCP approach corresponds to a measure of DC-coherence and hence, by proposition \ref{prop11}, to a measure of TC-coherence.  The question arises, therefore, of whether the measures of coherence in the BCP approach that have been recently proposed correspond to measures of coherence that were already well-known in the context of the resource theory of asymmetry.  The answer is that they were, as we demonstrate in Sec.~\ref{Sec:examples}.

\subsection{Criticism of these approaches to defining coherence}\label{Sec:CriticismIP}
%\subsection{The unnaturalness of the set of free measurements according to the BCP proposal}

%The restriction to  incoherent operations
As noted in the introduction, a given choice of the set of free operations is only physically justified if it can be understood as arising from some natural restriction on experimental capabilities.  
%In the case of the translationally-covariant operations, Sec.~\ref{JustificationsTC} described several ways in which this set could be justified by a natural restriction.  In the case of dephasing-covariant operations, it was unclear whether such a restriction can be found.  
The situation where the free operations are the incoherence-preserving or incoherent operations is like that of the dephasing-covariant operations---it is unclear whether there is a natural restriction that picks out these sets. 
%The question is an important one and bears more investigation, as the usefulness of these approaches hinges on its answer.

Nonetheless, we describe two features of the incoherence-preserving or incoherent operations that seem pertinent to the question of whether they can arise from a natural restriction: how they constrain measurements and the nonexistence of a certain kind of Stinespring dilation.

%Our criticism of these approach, relative to the two alternatives we have presented in this article, is based on two facts about the former.  First, the set of measurements they take to be freely implementable is unnatural.  Second, they fail to satisfy a natural consistency requirement related to the existence of a Stinespring dilation. 
% nonunitary incoherent operations cannot be simulated using unitary incoherent operations and incoherent states.

\subsubsection{Incoherence-preserving and incoherent measurements}
%{The unnaturalness of the set of free measurements}

If the output space of a quantum operation is trivial, so that it corresponds to tracing with a measurement effect $E$ on the input space, that is,  $\mathcal{E}(\cdot) = {\rm Tr}(E \cdot)$, then the definition of incoherence-preserving in terms of the dephasing map, Eq.~\eqref{Def2}, reduces to a trivial condition that is satisfied by all effects $E$. 
%\begin{align}
%{\rm Tr}(E \mathcal{D}(\cdot)) &= {\rm Tr}(E \mathcal{D}(\cdot)).
%\end{align}
%But this condition is true for {\em all} effects $E$, 
Consequently, there is no constraint on the effects in this approach.    Similarly, in the case of the incoherent operations proposed by BCP, for any given POVM, the operation associated to a given outcome can always be chosen to prepare the output system in an incoherent state.  
It follows that the set of incoherent measurements includes {\em any} POVM.
\begin{proposition}\label{prop:meas}
The sets of POVMs associated to the set of incoherence-preserving measurements and the set of incoherent measurements are both the full set of POVMs.
%The set of POVMs considered to be free in the BCP approach (i.e. included among the incoherence-preserving operations) is the full set of  POVMs.
\end{proposition}

%\color{red}[ Should we say something about how this is still true even if one considers the incoherent operations?]\color{black}

%\color{red} [Should we say something about the transformative aspect of measurements and how BCP do put a constraint on this aspect?  I have already written some text about the difference between the predictive and transformative aspect of measurements.] \color{black}
%For a measurement to be in the class of incoherent operations, the instrument must be such that for every outcome, incoherent states are mapped to incoherent states.  This imposes a constraint on the predictive aspect of the measurement and hence on form of the instrument.  
In these approaches, therefore, there is no limitation on the retrodictive capacity of a free measurement. In particular, measurements in the free set are capable of detecting the presence of coherence in a state.  This is a rather counterintuitive feature for free measurements to have.  Furthermore, the fact that the free states are incoherent while the free effects are not implies that the proposal has an awkward asymmetry between prediction and retrodiction.  
%Intuitively, this should not be possible with a free measurement 
%constitutes a drawback of this approach.  
Recall that in the approaches based on translationally-covariant or dephasing-covariant operations, by contrast, the free effects are the incoherent effects. 
%Therefore, in a comparison between the two approaches to speakable coherence---via dephasing-covariant operations and via incoherence-preserving operations--intuitions about what POVMs ought to be free seem to favour the dephasing-covariant approach.
%Intuitions about which measurements ought to be free would therefore seem to favour dephasing-covariant operations over incoherence-preserving operations.
These considerations, in our view, suggest that this approach does not have a natural physical justification.  The criticism is not conclusive however---an explanation of circumstances in which just this sort of restriction arises may yet be forthcoming. 
\color{black}

%\color{red} [Rob, this reminds me a nice result of Lucian Hardy and Jonathan Walgate which shows that any pair of bipartite pure (entangled) states can be perfectly distinguished from each other, using LOCC. So, it's not correct to say that to distinguish different resource states we need to consume resource.]\color{black} 
%\color{blue} [Iman, we didn't say that.  We merely said that it is counterintuitive and not conclusive.]\color{black}  
 % Again, this result also holds for the incoherent operations of the BCP proposal (Definition \ref{footBCP}) because the set of unitary incoherent operations are equivalent to the set of unitary incoherence-preserving operations. 

%Note, furthermore, that the approach based on dephasing-covariant operations has time-symmetry between states and effects---both are restricted to incoherence-preserving operators.  In the BCP approach, on the other hand, only the states are restricted in form.  As such, the BCP approach involves an awkward lack of symmetry in its treatment of states and effects. 
%The BCP approach can also be critized on the grounds of having a certain time-asymmetry, namely, that while the free states are required to be diagonal in the preferred basis, the free observables are not.

%\textbf{The lack of a Stinespring dilation}
\subsubsection{Considerations regarding the existence of a free dilation}\label{dilationnotesIP}

%We begin by considering whether a free dilation is possible when one treats all physical systems equivalently, so that the auxiliary systems in the dilation are subject to the same restrictions as the systems are.
%As argued in the introduction, a strong desideratum on the set of free operations is that every operation can be achieved via a Stinespring dilation using only free states and free unitaries.  In the case of incoherence-preserving operations, the question is whether every such operation on a system can be obtained by coupling that system to an ancilla prepared in an incoherent state via an incoherence-preserving unitary operation on the composite.

It turns out that, if one assumes that the preferred subspaces for the system-ancilla composite are the tensor products of the preferred subspaces for the system and for the ancilla, i.e., if all systems are treated even-handedly, then incoherence-preserving and incoherent operations do not have free dilations.  For instance, the operation in Eq.~(\ref{examp}), which distinguishes states $|+\rangle$ and  $|-\rangle$, is both incoherent and incoherence-preserving and cannot be implemented in this way.

The lack of a free dilation can be understood as a consequence of the fact that when it comes to unitary operations, there is no distinction between the sets of dephasing-covariant, incoherent, or incoherence-preserving operations (This follows from proposition~\ref{coincidenceunitaries}, together with the fact that a unitary operation has a single term in its Kraus decomposition). It follows that one can substitute 
  ``incoherence-preserving (or incoherent) unitary operations'' for ``dephasing-covariant unitary operations'' in proposition \ref{Prop:StinespringForDC}, while preserving its validity.  Hence with incoherent states and incoherence-preserving (or incoherent) unitaries, one can only generate dephasing-covariant operations.  Because these are, by proposition~\ref{prop1}, a proper subset of the incoherent operations, it follows that we cannot implement every incoherence-preserving (incoherent) operation using incoherent states and incoherence-preserving (incoherent) unitaries. Again, we note that a physical justification might still be possible.

\subsubsection{Incoherent operations versus incoherence-preserving operations}

%We have left open the possibilty that there may be a d physical scenaro may lead to a restriction of the set of free operations 
%If a physical justification of a given set of free operations can be provided, then it can abjudicate disputes about what is the 

Which of various different sets of free operations is the appropriate one for defining a resource theory, for instance, whether to use the incoherence-preserving or the incoherent operatons to define coherence, is a question that can be settled by finding a physical scenario or restriction on experimental capabilities that picks out one or the other.   
%If the physical justification comes
If the set of operations are physically justified by an interaction between the system and an uncontrollable environment (on which one cannot implement any measurements), then incoherence-preserving operations are more natural than incoherent operations.
 %because, by assumption, one cannot implement measurements on the environment.  
 If, on the other hand, the physical justification comes from an interaction between the system and an apparatus with a classical read-out, then incoherent operations are more natural than incoherence-preserving operations.

Furthermore, if one seeks a physical justification in terms of a constraint on the dilation of an operation, then this too bears on the question of whether it is physically more reasonable to take the free operations to be the incoherent operations or the incoherence-perserving operations.    

%The question of whether a given set of operations can be understood in terms of a constraint on their dilations also bears on the question of whether it is physically more reasonable to take the free operations to be the incoherent operations or the incoherence-perserving operations.    

%This is because the best argument that can be given in favour of the incoherent operations appeals to the possibility of dilation. 
%The argument can be summarized as follows.  
Consider the following (flawed) argument in favour of using the incoherent operations.
Take the standard Stinespring dilation of an operation.  For any Kraus decomposition of that operation, it is possible to ensure that the effective map on the system is a single term in that Kraus decomposition by implementing the operation through its standard Stinespring dilation and then performing an appropriate measurement on the auxiliary system and post-selecting on a single outcome.  
%We will refer to such a scheme as {\em resolving} the operation into Kraus terms through its dilation.  
Given this possibility of realizing a single term in the Kraus decomposition, so the argument goes,  one should require each such term to be incoherence-preserving, rather than just requiring this of their sum. 
% Finally, the argument states that one should demand that each Kraus term in the resolution of the operation should be incoherence-preserving. 
%it is reasonable to demand that the effective map on the system under such post-selection should still be incoherence-preserving.  

However, this argument has appealed to the standard Stinespring dilation theorem which only guarantees 
that for every operation there is {\em some} unitary on a larger system that realizes it by dilation.  In the context of a resource theory, however, one cannot avail oneself of {\em any} unitary on the larger system because such a unitary might not be free.  Similar comments apply to the states and effects on the auxiliary system that appear in the dilation.  In a resource theory, if a free operation is implemented by dilation, then it must be implemented by a {\em free} dilation, which is a strict subset of all possible dilations.

Therefore, to settle this issue by appeal to dilations one must find a physical justification of either the incoherence-preserving or incoherent operations in terms of a restriction on the dilation resources, which is an unsolved problem, as we noted in the previous section.

\color{black}
Finally, we noted in the introduction that our proposal for the set of free operations in a theory of speakable coherence, the dephasing-covariant operations, is closely related to the proposal found in Ref.~\cite{yadin2015quantum}.
 %the proposal for the set of free operations found in Ref.~\cite{} is closely related to the one based on dephasing-covariant operations. 
  In fact, the set of free operations of Ref.~\cite{yadin2015quantum} stands to the set of dephasing-covariant operations as the set of incoherent operations stands to the set of incoherence-preserving operations.  As such, disputes about the relative merits of the fomer two sets are akin to those about the relative merits of the latter two---they will only be resolved when one or the other proposal is given a physical justification as all and only the operations that can be dilated using a restricted set of states, effects and unitaries.

% can one then turn to the question of what sorts of resolutions of the operation into more fine-grained operations can be achieved by post-selecting on the auxiliary system in the dilation.  
%Only if it is possible to understand the set of incoherence-preserving (or incoherent) operations as all and only those that admit of a dilation that is free in this broader sense can one then turn to the question of what sorts of resolutions of the operation into more fine-grained operations can be achieved by post-selecting on the auxiliary system in the dilation.  
\color{black}

\begin{comment}
%Rather, as we argued in Sec. ??, a given  set of free operations is only justified if each element admits of a dilation in terms of some restricted degree of control on the auxiliary system.  
%post-selecting on a single outcome, to realize a single term in the Kraus decomposition of the operation.  
 %However, one must show that this possibility remains for a {\em free dilation}.   
\color{red}
[The following is not clear]
 Given that no one has yet suggested how to understand incoherence-preserving operations as all and only those arising from some set of free dilations, it is not possible to settle the question of whether these dilations allow one to resolve the operation into Kraus terms. 
 %a restriction on experimental capabilities such that the free  dilations  of the incoherent operations has been provided, one cannot settle the question of whether one can post-select in the relevant fashion.  
  In the absence of any demonstration that such a resolution of the operation is possible,  there is no reason to suppose that  one should demand the incoherence-preserving property for every term in the Kraus decomposition.
\color{black}
\end{comment}

\section{Methods for deriving measures of coherence} \label{Sec:examples}
%\section{Measures of coherence according to the different proposals} \label{Sec:examples}
%\section{Examples of measures of asymmetry and coherence} \label{Sec:examples}

\color{red} 
%We are going to focus on how to find functions that are measures of TC-coherence and how to find functions that are measures of IP-coherence.  The latter class of functions are measures of coherence for any set of free operations that has the incoherence-preserving property.

%REVISE THIS SECTION! FOCUS ON HOW ALL PROPOSED MEASURES DO NOT DISTINGUISH THE SUBTLE CATEGORIES, THEY ARE JUST IP MEASURES.
\color{black}
%\red{[Should we talk about TC-coherence or unspeakable coherence?] }\color{cyan} It should be TC-coherence! Change the following. \color{black}
In this section, we consider measures of coherence for the various different sets of free operations described in the article, in particular, TC-coherence, DC-coherence and IP-coherence.  We have already noted, in Propositions \ref{prop11} and \ref{propIncDC}, that if one considers the same choice of preferred subspaces, then every measure of  IP-coherence is also a measure of DC-coherence which is also a measure of TC-coherence.  \color{black} Because a measure of TC-coherence is a measure of translational asymmetry, one can immediately obtain many interesting measures of coherence by simply appealing to the known measures of asymmetry. Indeed, we show that most of the recent proposed measures of coherence in the BCP proposal correspond to measures of asymmetry that have been previously studied in Refs.~\cite{Marvian_thesis} and \cite{Noether}. Because of the strictness of the inclusions in Proposition \ref{prop11}, however, only a subset of the measures of TC-coherence are measures of DC-coherence or IP-coherence.  We will highlight some examples of the strictness of the inclusion.

In addition, we present certain general techniques for deriving measures of coherence, adapted from ideas introduced in the context of asymmetry theory,
%and show how they can be generalized to derive not only measures of TC-coherence, but measures of DC-coherence as well.  
and we review some of the most important examples of measures of coherence.
 (Note, however, that the list we provide is not complete; there are many known measures of asymmetry that we do not review here.  See e.g. \cite{toloui2011constructing,toloui2012simulating, gour2008resource, Marvian_thesis}.)  

\color{black}

%\subsection{Measures of asymmetry from information monotones}
%\subsection{Measures of unspeakable coherence based on measures of information}
\subsection{Measures of coherence based on measures of information}

%In Sec.~\ref{infodual}
In Sec.~\ref{JustificationsTC}, we showed that the resource for phase estimation is TC-coherence.  We saw that there is a duality between the problem of state transformation in the resource theory of TC-coherence on one hand, and the problem of processing the classical information encoded in the phase shifted versions of the state.   This was formalized by proposition \ref{prop14}.   This duality can be leveraged to derive measures of  unspeakable coherence from measures of information.

We begin by recalling the definition of a measure of the information content of a quantum encoding of a classical message. 
%information encoded in an ensemble of quantum states.
\begin{definition}\label{measureofinformation}
A function $I$ from sets $\{ \rho(x) \}_x$ of quantum states to the reals is a measure of the  information about $x$ contained in the quantum encoding if \\
%\textbf{A}, \textbf{B} or \textbf{C} if it satisfies the followings \\
\noindent (i) For  any trace-preserving quantum operation $\mathcal{E}$ 
%and any ensemble of states $\{ \rho(x) \}_x$,
%$\{ p_k, \rho_k \}_k$,
%in the corresponding resource theory,
 it holds that $I(\{ \mathcal{E}(\rho(x))\}_x )\le I(\{ \rho(x)\}_x )$.\\
 % it holds that $I(\{ p_k, \mathcal{E}(\rho_k)\}_k )\le I(\{ p_k,\rho_k\}_k )$.\\
\noindent (ii) For any trivial encoding, the elements of which are indistinguishable, $I$ takes the value 0.
\end{definition}

If we define a function $f_I$ on states to be such that its value on a state is the measure of information $I$ on the set of states obtained by acting on the state with all elements of the translation group (i.e. the orbit under translations of that state),
\begin{equation}\label{TCcoherenceFromInfo}
f_I(\rho) = I( \{ \mathcal{U}_{L,x}( \rho )\}_{x\in \mathbb{R}} ),
\end{equation}
then by proposition~\ref{prop14} and the definitions of measures of TC-coherence and measures of information, we see that $f_I$ is a measure of TC-coherence if  $I$ is a measure of information.

In particular, one can obtain measures of TC-coherence from measures of the distinguishability of any pair of states in the translational orbit of $\rho$ \cite{Marvian_thesis, Noether}.

A similar sort of consideration allows us to infer measures of DC-coherence from measures of information.  
In particular, for any state $\rho$, one can encode 1 bit of classical information $b$ as $b\rightarrow \rho_b$ where $\rho_0=\rho$ and $\rho_1=\mathcal{D}(\rho)$. Then, it can be easily seen that for any measure of information $I$, 
\begin{equation}\label{DCcoherenceFromInfo}
f(\rho) = I( \{\rho, \mathcal{D}(\rho)\}),
\end{equation}
i.e., the amount of information about bit $b$ that can be transferred using this encoding is  a measure of DC-coherence. This follows from the fact that, by definition, if $\rho$ can be transformed to $\sigma$ by a dephasing-covariant trace-preserving quantum operation, then  the same quantum operation transforms $\mathcal{D}(\rho)$ to $\mathcal{D}(\sigma)$, and hence there is a trace-preserving quantum operation which transforms the $\rho$-based encoding of $b$ to the $\sigma$-based encoding of $b$.
%one encoding of $b$ to the other. Therefore, the latter encoding cannot transfer more information about the bit $b$ 
But this in turn implies that the $\sigma$-based encoding cannot have more information about $b$ than the $\rho$-based encoding.  (This can be thought of as the analogue of the easy direction of the duality in proposition \ref{prop14}). 

\color{black}

We now consider various measures of coherence that can be derived from measures of information. 
%via proposition \ref{prop14} and the recipe described below it.\color{red} [Provide this recipe here?]\color{black}

%\color{blue} [We can just use the dephasing map here]
 
%In some of the  following examples we use the \emph{uniform twriling} quantum channel $\mathcal{N}$, defined as

%[Placeholder: Find a way to make the different numbered items more visually distinct on the page.
%Mention connection to the Fisher information?]

 \color{black}
(i) 
If one uses the duality described by Eq.~\eqref{TCcoherenceFromInfo} with the Holevo quantity  as the measure of information, then, following the argument of Ref.~\cite{Noether} and assuming a uniform probability density over the translations, one can prove that the function 
 \begin{align}
\Gamma(\rho)\equiv S(\mathcal{D}({\rho}))-S(\rho)\ .
\end{align}
where $S$ is the von Neumann entropy, is a measure of TC-coherence.  

Meanwhile, if one uses the duality described by Eq.~\eqref{DCcoherenceFromInfo} with the quantum relative entropy, $S(\rho||\sigma)\equiv \text{tr}(\rho\log \rho)-\text{tr}(\rho\log \sigma)$,  as the measure of information, then one can prove that the function 
 \begin{align}
\Gamma'(\rho)&= S\left(\rho||\mathcal{D}({\rho})\right)\ ,
\end{align}
is a measure of DC-coherence.  

Finally, the function
\begin{align}
\Gamma''(\rho)=\min_{\sigma\in \mathcal{I}}  S\left(\rho||\sigma\right),
%&= S\left(\rho||\mathcal{D}({\rho})\right)\ ,
\end{align}
the minimum relative entropy distance of $\rho$ to the set of incoherent states, is clearly nonincreasing under incoherence-preserving operations and is therefore a measure of IP-coherence.   It is also nonincreasing under incoherent operations~\cite{Coh_Plenio}.
%, as was first noted in Ref.~\cite{Coh_Plenio}, and hence is also a measure of coherence in the BCP approach, where it has been dubbed the \emph{relative entropy of coherence}. 

It turns out that the  three measures are all equivalent, that is, 
 \begin{align}
\min_{\sigma\in \mathcal{I}}  S\left(\rho||\sigma\right) = S\left(\rho||\mathcal{D}({\rho})\right) = S(\mathcal{D}({\rho}))-S(\rho)\ ,
\end{align}
a fact that has been noted by many authors~\cite{aberg, GMS09, Coh_Plenio}.
%A simple proof can be found  in Ref.~\cite{GMS09}  (see proposition 2).  
 So this is an example of a function that is a measure of coherence in all of the approaches we have considered.

%This function was called the {\em Holevo asymmetry measure} in Ref.~\cite{Noether}.  
In the context of asymmetry theory, this function was first introduced by Vaccaro et al., who called it simply {\em the asymmetry}~\cite{vac2008}.   It was  further studied as a measure of asymmetry in Ref.~\cite{GMS09} (see proposition 2) and it was first derived from the Holevo information in Ref.~\cite{Noether}, where it was called the {\em Holevo asymmetry measure}. Ref.~\cite{Noether} also proposed that it and other measures of asymmetry could be used to quantify coherence.
% wher this method of derivation was first described.  
 BCP noted in Ref.~\cite{Coh_Plenio}  that this function was monotonically nonincreasing under incoherent operations and hence a measure of coherence in their approach, 
%, as was first noted in Ref.~\cite{Coh_Plenio}, and hence is also a measure of coherence in the BCP approach, 
where it has been dubbed the \emph{relative entropy of coherence}. 
%BCP dubbed this measure the \emph{relative entropy of coherence}.  
Several years prior both to BCP's work and the work which studied it as a measure of translational asymmetry, 
%  studying this function as a measure of coherence, and the observation that measures of asymmetry yield measures of coherence
%to this work, however, and indeed prior to the 
%Indeed, several years before BCP \cite{Coh_Plenio},
 this function was proposed as a measure of coherence by \.{A}berg in a paper entitled ``Quantifying superposition"~\cite{aberg}.

\color{black}
(ii) 
If we start from Eq.~\eqref{TCcoherenceFromInfo} using the Holevo quantity, but where the probability distribution over translations is allowed to be arbitrary rather than uniform, it is possible to show that the following function is also a measure of TC-coherence:
\begin{equation}
\Gamma_p(\rho)=S\left(\mathcal{D}_p(\rho) \right)-S(\rho)\ ,
\end{equation}
where $p$ is an arbitrary probability density on the real line and $\mathcal{D}_p$ is a ``weighted twirling operation'' defined by
\begin{equation}\label{weight-twr}
\mathcal{D}_p(\cdot)=
\lim_{x_0\rightarrow \infty} \frac{1}{2x_0}\ \int_{-x_0}^{x_0}  dx\ p(x)\\  e^{-i x L}(\cdot)  e^{i x L}\ .
\end{equation}
This translates into the language of coherence a measure of asymmetry identified in Ref~\cite{Noether}.

\begin{comment}
In Ref.~\cite{Noether},   we proposed a generalization of  the function $\Gamma$ as 
\begin{equation}
\Gamma_p(\rho)=S\left(\mathcal{D}_p(\rho) \right)-S(\rho)
\end{equation}
where $p$ is an arbitrary probability distribution over the reals
%interval $(0,2\pi]$
 and $\mathcal{D}_p$ is a ``weighted twirling operation'' defined by
\begin{equation}\label{weight-twr}
\mathcal{D}_p(\cdot)=
%\int_{0}^{2\pi} d\theta\ p(\theta)\  e^{-i \theta N} (\cdot) e^{i \theta N}\ .
\lim_{x_0\rightarrow \infty} \frac{1}{2x_0}\ \int_{-x_0}^{x_0}  dx\ p(x)\\  e^{-i x L}(\cdot)  e^{i x L}\ .
\end{equation}

This is also obtained from Eq.~\eqref{TCcoherenceFromInfo} using the Holevo quantity, but where the probability distribution over translations is allowed to be arbitrary. \color{black}
\end{comment}

Note that for any symmetric state $\rho$ and any arbitrary probability density $p$, $\Gamma_{p}(\rho)=0$.  In Ref.~\cite{Noether}, it is shown that using a nonuniform probability density can be useful in some physical examples.

It turns out that for  a general probability density $p$, the function $\Gamma_p$ is not a measure of DC-coherence and hence not a measure of IP-coherence either.
%To see this, consider the case where $\theta$ is with probability 1/2 equal to zero and with probability 1/2 equal to $\pi$. In this case, the function $\Gamma_p$ is zero for the state $(|0\rangle+|2\rangle)/\sqrt{2}$, while it is non-zero for the state $(|0\rangle+|1\rangle)/\sqrt{2}$. But these two states are equivalent resources from the point of view of dephasing-covariant operations, and hence they should have the same measure of coherence, thereby showing that this $\Gamma_p$ cannot be a measure of coherence in the dephasing-covariance approach.

%Because $\Gamma_p$ for general $p$ need not be a measure of coherence in the dephasing-covariance approach, proposition~\ref{propIncDC} assures us that it need not be a measure of coherence in the BCP approach either. 

(iii)
Using Eq.~\eqref{DCcoherenceFromInfo} while taking as our measure of information the Holevo quantity for a quantum encoding of a classical bit but where the bit values have unequal prior probabilities, we can derive a measure of DC-coherence.  Specifically, for any $q\in [0,1]$, the function
\begin{equation}
\Gamma_q(\rho)=S\left(q \rho + (1-q) \mathcal{D}(\rho) \right)- qS(\rho) - (1-q) S(\mathcal{D}(\rho)),
\end{equation}
is a measure of DC-coherence.

(iv)  
Starting from the monotonicty of the  \emph{relative Renyi entropy}  under information processing~\cite{hayashi2006quantum}, one can use 
%proposition \ref{prop14} and 
Eq.~\eqref{TCcoherenceFromInfo} to show that the function
%derive the monotonicity under translationally-covariant operations of the function \bes
\begin{align}
S_{L,s}(\rho)&\equiv \text{tr}(\rho L^{2})-\text{tr}(\rho^{s}L \rho^{(1-s)}L)\\ &= -\frac{1}{2} \text{tr}\left([\rho^s,L][\rho^{1-s}, L] \right)
\end{align}
%\ees
for $0<s<1$ is a measure of TC-coherence.  The argument is provided, in the context of asymmetry theory, 
%s was first considered as a measure of asymmetry
 in Refs.~\cite{Noether} and \cite{Marvian_thesis}.
%In \cite{Noether} and \cite{Marvian_thesis}, we also used  proposition \ref{prop14} and Eq.~\eqref{TCcoherenceFromInfo} to  show that the function
%\bes
%\begin{align}
%S_{N,s}(\rho)&\equiv \text{tr}(\rho N^{2})-\text{tr}(\rho^{s}N \rho^{(1-s)}N)\\ &= -\frac{1}{2} \text{tr}\left([\rho^s,N][\rho^{1-s}, N] \right)
%\end{align}
%\ees
%is a measure of asymmetry for $0<s<1$. The monotonicty of this function under symmetric operations follows from proposition \ref{prop14} together with the monotonicty of   \emph{relative Renyi entropy} ~\cite{hayashi2006quantum} under information processing.

Interestingly, this quantity has been introduced before by Wigner and Yanase for $s=1/2$ ~\cite{wigner1963information} (and generalized by Dyson to arbitrary $s$ in $(0,1)$) and since then some of its interesting properties have been studied. However, its monotonicity under symmetric dynamics, and hence its interpretation as a measure of asymmetry, was not recognized in the past.
% (See Sec. \ref{Sec:History} in the appendix  for further discussion about the history and the properties of this function).

It has been claimed by Girolami  \cite{Girolami}  that this function is  a measure of coherence according to the definition of BCP
%Baumgratz et. al.
 \cite{Coh_Plenio},  that is, he claimed that it is non-increasing under incoherent operations. However, as is noted in \cite{du2015wigner} and \cite{marvian2015quantum}, this claim is incorrect. This can be  seen, for instance, by recognizing that in the case of pure states this function is equal to the variance of the observable $L$, but variance obviously is  not invariant under permutations of the eigenspaces of $L$. In other words, this function is only a measure of unspeakable coherence, not of speakable coherence.

Note that for any incoherent state $\rho$, it holds that $S_{L,s}(\rho)=0$.  Furthermore, for pure states, the Wigner-Yanase-Dyson skew information reduces to the variance of the observable $L$, that is, 
\begin{equation}
S_{L,s}(|\psi\rangle\langle\psi|)=  \langle\psi| L^{2}|\psi \rangle-\langle\psi| L|\psi \rangle^{2}.
\end{equation}
%Given that an incoherent mixture over the eigenspaces of $N$  has vanishing asymmetry according to this measure, while a superposition over these eigenspaces that is totally coherent has asymmetry $S_{N,s}$ equal to the variance over $N$, this asymmetry measure 
For a general mixed state, a nonzero variance over $L$ does not attest to there being coherence between the $L$ eigenspaces because it can also be explained by an incoherent mixture of the latter. The function $S_{L,x}$, on the other hand, seems to succeed in quantifying the amount of variance over $L$ that is coherent, which one might call the ``coherent spread'' over the eigenspaces of $L$. 
It is also worth mentioning that recently, Ref.~\cite{Girolami} has proposed  a method for measuring this quantity.

Interestingly, it has been noted that the function which is the average over $s$ of the Wigner-Yasane-Dyson skew information for index $s$, $\int_0^1 ds\  S_{L,s}(\rho)$, has a natural interpretation as the \emph{quantum fluctuations} of the observable $L$, i.e., as the difference between the total fluctuations $\langle\delta^2 L\rangle$ and the (classical) thermal fluctuations \cite{Frerot2015}. Furthermore, it has been shown that this quantity can be calculated in several interesting examples of many-body systems \cite{Frerot2015}.

(v)  
If the relative Renyi entropy is used in Eq.~\eqref{DCcoherenceFromInfo}, we can prove that the function
\begin{align}
S_{s}(\rho)&\equiv  \frac{1}{s-1} \log \left[ {\rm tr} \left( \rho^s [\mathcal{D}(\rho)]^{1-s} \right) \right]
\end{align}
for $0<s<1$ is a measure of DC-coherence. 

(vi) Following an argument presented in Ref.~\cite{Noether}, we can use Eq.~\eqref{TCcoherenceFromInfo} to show that the function
%Using proposition \ref{prop14} and the recipe for constructing measures of asymmetry, in Ref.~\cite{Noether}   we showed that the function
\begin{equation}
F_{L}(\rho)\equiv\|[\rho,L]\|_1
\end{equation}
where $\|\cdot\|_1$ is the trace norm (i.e., the sum of the singular values) is a measure of TC-coherence.   This measure formalizes the intuition that the coherence of a state with respect to the eigenspaces of $L$ can be quantified by the extent to which the state fails to commute with $L$.   The state $\rho$ has coherence relative to the eigenspaces of $L$ if and only if $[\rho,L]\ne 0$ so in retrospect one would naturally expect that \emph{some} operator norm of the commutator $[\rho,L]$ should be a measure of TC-coherence.  This intuition does not, however, tell us \emph{which} operator norm to use.  Our result shows that it is the trace norm that does the job~\footnote{Note that for $s=1/2$, we have $S_{L,s=1/2}=\|[\rho^{1/2}, L]\|_2/2$ where $\|\cdot\|_2$ is the Frobenius norm, that is, the sum of the squares of the singular values. So, both $\|[\rho^{1/2}, L]\|_2$ and $\|[\rho, L]\|_1$ are measures of asymmetry.}.

The function $F_{L}$  reduces to a simple expression for pure states: it is proportional to the square root of the variance of the observable $L$, that is,
\begin{equation}
F_{L}(|\psi\rangle\langle\psi|)=2 \left( \langle\psi| L^{2}|\psi \rangle-\langle\psi| L|\psi \rangle^{2} \right)^{1/2}.
%\sqrt{\text{Var}_{L}(\psi)}.
\end{equation}
Again, we see that a mixture over the eigenspaces of $L$ has vanishing $F_{L}$, while a coherent superposition over these eigenspaces has $F_L$ that depends only on the variance over $L$.  Consequently, this coherence measure, like $S_{L,s}$, in some sense quantifies the coherent spread over the eigenspaces of $L$. 

We see that, when restricted to pure states, the function $F_{L}$ is a monotonic function of $S_{L,s}$.  Given that the latter is neither a measure of DC-coherence nor a measure of IP-coherence, as argued above, it follows that the former is not either.
%This function is not a measure of coherence in the dephasing-covariance approach \color{red} because it violates the criterion in proposition \ref{Prop:criterion}.   It then follows from proposition~\ref{prop11} that it is not a measure of coherence in the BCP approach either. \color{black}

(vi) 
%Using  proposition \ref{prop14} and the recipe below it,
Arguments in Refs.~\cite{Marvian_thesis} and \cite{Noether} show that  the function 
\begin{equation}
R_p(\rho)\equiv\|\rho-\mathcal{D}_p(\rho)\|_1 
\label{measurev}
\end{equation}
is a measure of TC-coherence for an arbitrary probability distribution $p$ on the reals. 

In the special case where $p$ is the uniform distribution, the function becomes 
\begin{equation}
R(\rho)\equiv\|\rho-\mathcal{D}(\rho)\|_1.
\label{measurevuniform}
\end{equation}
%i.e. where $R_p(\rho)=\|\rho-\mathcal{D}(\rho)\|_1 $, 
The latter is a measure of DC-coherence, as we show in the next section.  However, for a general distribution $p$, the function $R_p$ can increase under dephasing-covariant operations, and hence it is not a measure of DC-coherence. It then follows from proposition~\ref{prop11} that $R_p$ for general distribution $p$ is not a measure of IP-coherence either.

\subsection{Measures of coherence based on mode decompositions}\label{Sec:modes}

In Sec.~\ref{Modes}, we introduced the concept of the \emph{mode decomposition} of states and operations, first introduced in Ref.~ \cite{Modes} as a useful method in the resource theory of asymmetry.  In the language of mode decompositions, the translationally-covariant operations are those such that a given mode component of the input state is mapped to the corresponding mode component of the output state (proposition \ref{defnmodes}).
This implies, in particular, that one can only generate a given output if the input contains all of the necessary modes.

%The central idea is to study the decomposition of states, quantum operations and  measurements into different modes, which we call \emph{modes of asymmetry}. Under symmetric quantum operations, a given mode of the input is mapped to the corresponding mode of the output, implying that one can only generate a given output if the input contains all of the necessary modes. By defining monotones that quantify the asymmetry in a particular mode, we also derive quantitative constraints on the resources of asymmetry that are required to simulate a given asymmetric operation. This approach  provides a new insight into the structure of the resource  theory of asymmetry and provides general explanations for the special observations that have been made before. 

%This method has been used in \cite{lostaglio2015quantum} and \cite{lostaglio2015description} to study coherence in the context of quantum thermodynamics. 

\begin{comment}
\end{comment}

Based on this observation, Ref.~\cite{Modes} noted that one can define a family of asymmetry measures which quantify the amount of asymmetry in each mode.  By considering translational symmetry, we obtain measures of TC-coherence. In particular, using the monotonicity of the trace-norm under trace-preserving completely positive maps, we find that for each $\omega \neq 0$, the function
\begin{equation}\label{MeasureModek}
\left\|\mathcal{P}^{(\omega)}(\rho)\right\|_1
%=\left\|\sum_n \Pi_{n+k}(\cdot)\Pi_{n}\right\|_1
\end{equation}
is a measure of TC-coherence (the $\omega=0$ case yields a constant function and so is uninteresting). Indeed, we find 
%\color{red} [we should provide the proof of this] \color{black} 
that any linear function of modes can lead to a different measure of TC-coherence. In other words, for any set of complex numbers $\{c^{(\omega)}\}$, the function 
\begin{equation}\label{MeasureLinearModes}
\left\|\sum_{\omega\in\Omega} c^{(\omega)}\mathcal{P}^{(\omega)}(\rho) \right\|_1 ,
\end{equation}
is a measure of TC-coherence, where $\Omega$ is the set of modes corresponding to the generator $L$. 

%An interesting special case is where we choose $c^{(0)}=0$  and $c^{(\omega)}=1$ for all $\omega\neq 0$. In this case,
%\beq
%\sum_{\omega\in\Omega} c^{(\omega)}\mathcal{P}^{(\omega)}(\rho)=\rho-\mathcal{D}(\rho)=\rho^{\rm offd}\ .
%\eeq
%Therefore, in this special case, the measure of TC-coherence in Eq.~(\ref{MeasureLinearModes}) is equal to  the function $R(\rho)$ of Eq.~\eqref{measurevuniform} and consequently is also a measure of DC-coherence and IP-coherence. 

%[It would be better to show that it is a measure of DC-coherence using modes.]
Proposition~\ref{modecharDCops} established that dephasing-covariant operations preserve the diagonal and off-diagonal modes of a state. 
%A quantum operation  $\mathcal{E}$ is dephasing-covariant relative to a preferred set of subspaces if and only if it preserves the diagonal and off-diagonal modes.  Formally, the condition is that  whenever $\mathcal{E}(\rho)=\sigma$, we have $\mathcal{E}(\rho^{\rm diag})=\sigma^{\rm diag}$  and $\mathcal{E}(\rho^{\rm offd})=\sigma^{\rm offd}$ .
 It follows that 
%But then, by the monotonicity of the trace-norm under trace-preserving completely positive maps, it follows that the function
\begin{equation}\label{Measureprojectoffd}
\left\|\rho^{\rm offd} \right\|_1 = \left\|[\mathcal{I}_\text{id} - \mathcal{D}](\rho)\right\|_1
%= \left\|\rho - \mathcal{D}(\rho)\right\|_1
\end{equation}
is nonincreasing under dephasing-covariant operations and hence is a measure of DC-coherence.  Eq.~\eqref{Measureprojectoffd} is just 
 %$\left\|\rho^{\rm offd} \right\|_1 = \left\|\rho - \mathcal{D}(\rho)\right\|_1$, we see that this is simply the measure 
 $R(\rho)$ of Eq.~\eqref{measurevuniform}, and so these considerations have only provided an independent way of seeing that it is a measure of DC-coherence. 
 %hence also a measure of IP-coherence. 

Recall that the space of off-diagonal operators is equal to the direct sum of the spaces of mode-$\omega$ operators for $\omega \ne 0$, $\mathcal{B}^{\rm offd} = \bigoplus_{\omega \ne 0} \mathcal{B}_{\omega}$ (Eq.~\eqref{relateBoffdtoBk}), which implies that the superoperator that projects on the one space also projects onto the other, that is, 
\beq
%\mathcal{P}^{\rm offd} 
\mathcal{I}_\text{id} - \mathcal{D}= \sum_{\omega \ne 0} \mathcal{P}_\omega .
\eeq
It follows that the function \eqref{Measureprojectoffd} is a special case of the function \eqref{MeasureLinearModes} where we choose $c^{(0)}=0$  and $c^{(\omega)}=1$ for all $\omega\neq 0$, thereby confirming that this particular measure of DC-coherence is also measure of TC-coherence, as Proposition \ref{prop11} requires. 

%Using mode decompositions, we can furthermore confirm that not every  measure of TC-coherence is a measure of DC-coherence.The function \eqref{MeasureModek} for some particular $\omega$ is an example.  It can increase under dephasing-covariant operations because these can move weight from other mode components into the mode-$\omega$ component. 

 \color{black}

However,  measures of TC-coherence based on Eq.~(\ref{MeasureLinearModes}) will not, in general, be measures of DC-coherence or IP-coherence. In particular, the function \eqref{MeasureModek} for some particular $\omega$ is an example.  It can increase under dephasing-covariant or incoherence-preserving or incoherent operations because these can move weight from other mode components into the mode-$\omega$ component.
% So, measures of TC-coherence based on the mode decomposition idea, in general, will not be measures of speakable coherence. 
This is yet another proof of the strictness of the inclusion in proposition \ref{prop11}.  

Interestingly, measures of TC-coherence of the form of Eq.~(\ref{MeasureModek}) are closely related to a method which is regularly used in NMR for characterizing the coherence of states. Here, the relevant observable $L$ is the magnetic moment in the direction of the quantization axis. The modes corresponding to this observable, i.e., the differences of its  eigenvalues, are integers.  Then, using NMR techniques, one can experimentally measure functions 
\begin{equation}
\left\|\mathcal{P}^{(k)}(\rho)\right\|_2=\sqrt{\Tr(\mathcal{P}^{(k)}(\rho) {\mathcal{P}^{(-k)}(\rho)})}\ ,
%=\left\|\sum_n \Pi_{n+k}(\cdot)\Pi_{n}\right\|_1
\end{equation}
for integer $k$, where $\|\cdot \|_2$ is the Frobenius norm \cite{cappellaro2007simulations, cho2006decay, cappellaro2014implementation}. Strictly speaking, these functions are not measures of TC-coherence, i.e. they can increase by translationally-covariant operations such as partial trace. However, these functions provide useful lower and upper bound on $\left\|\mathcal{P}^{(k)}(\rho)\right\|_1$, which {\em are} measures of TC-coherence, namely,
\begin{equation} \label{modeNMR}
\left\|\mathcal{P}^{(k)}(\rho)\right\|_2 \le \left\|\mathcal{P}^{(k)}(\rho)\right\|_1\le \sqrt{d}  \left\|\mathcal{P}^{(k)}(\rho)\right\|_2\ ,
\end{equation}
where $d$ is the dimension of the Hilbert space.

\color{black}
\section{Concluding Remarks}\label{Discussion}

We have shown that the translationally-covariant operations define a useful resource theory of unspeakable coherence.  The constraint of translationally covariance is seen to arise naturally in many physical scenarios, each motivated by a different application of unspeakable coherence. In the case of speakable coherence, we have explored two sorts of approaches, one based on dephasing-covariant operations and the other based on operations that are incoherence-preserving (the BCP approach is a variant of the latter where a free operation is one that has a Kraus decomposition each term of which is incoherence-preserving).  It is currently unclear whether there are physical scenarios that pick out one of these sets of operations as the freely-implementable ones.  We have, however, constrained the shape of a putative physical justification.

%We would not be surprised if it turned out, in fact,
A possibility worth considering is that speakable coherence, unlike its unspeakable counterpart, {\em cannot} be usefully understood as a resource.   Perhaps the resource that powers tasks involving speakable information is not, in fact, a resource of {\em coherence}, but rather a different property of quantum states~\footnote{This is analogous to how it is asymmetry rather than entanglement that is the resource powering quantum metrology~\cite{Noether}.}.   Even if this different property {\em implied} having some coherence in the state, it might be that coherence was merely {\em necessary} but not sufficient for acheiving the task. In this case, it would be incorrect to identify coherence as the resource powering the task.

\begin{comment}
%We would not be surprised if it turned out, in fact,
It is possible that a speakable notion of coherence cannot be usefully understood as a resource.   Once one eliminates all of those tasks for which the relevant resource is {\em unspeakable coherence}, it would seem that what remains are information-theoretic tasks for which the nature of the degree of freedom encoding the information is not relevant---so-called ``fungible'' or ``speakable'' information [cite Bartlett, Rudolph, Spekkens RMP].  The question, therefore, is whether the resource that powers various quantum information-processing taks is, in fact, a resource of {\em coherence}.   It might well be that the resource for many such tasks is a different property of quantum states.\footnote{This is analogous to how it is asymmetry rather than entanglement that is the resource powering quantum metrology~\cite{Noether}.
%; for certain composite systems, a state being asymmetric implies that it is entangled, entanglement is sometimes a necessary but not a sufficient condition for quantum metrology ,  entanglement is sometimes a necessary condition for asymmetry
}   Even if this different property {\em implied} having some coherence in the state, it might be that coherence was merely {\em necessary} but not sufficient for acheiving the task. In this case, it would be incorrect to identify coherence as the resource powering the task. 
\end{comment}

This is indeed the case for at least one model of quantum computation, namely, the state injection model \cite{bravyi2005universal}.  Here, the circuit consists entirely of Clifford gates---i.e., those that take the set of Stabilizer states to itself---and one allows injection of nonStabilizer states.  The injection of nonStabilizer states is critical for achieving universal quantum computation because, as the Gottesman-Knill theorem shows, a Clifford circuit can be efficiently simulated classically \cite{nielsen2000quantum}.  Note that a Clifford circuit acting only on Stabilizer states is efficiently simulatable {\em even though} the states throughout the computation have coherence relative to the computational basis.   Clearly, then, coherence is not sufficient for achieving universal quantum computation in the state injection model.  
Furthermore, for the case where the systems have dimension corresponding to an odd prime, it has been shown that a {\em necessary} condition on the injected states for achieving universal quantum computation is that the circuit should fail to admit of a noncontextual hidden variable model \cite{howard2014contextuality}.  A Clifford circuit acting only on Stabilizer states admits of such a model via Gross's discrete Wigner representation \cite{gross2006hudson, spekkens2008negativity}. As such, the failure of noncontextuality is a much more stringent requirement than the presence of coherence and, unlike coherence, it is a viable candidate for the resource that powers quantum computational advantages in the state injection model.

The prospects for speakable coherence as a resource are better for cryptographic tasks.
For instance, the BB84 quantum key distribution protocol \cite{bennett1984proceedings} requires non-orthogonal states and therefore  requires the preparation of states that have coherence relative to the preferred basis (regardless of one's choice of preferred basis).  Furthermore, the BB84 protocol uses only stabilizer states and measurements, so that the latter are sufficient for the protocol, unlike the situation for universal quantum computation in the state injection model. 

%Indeed, it is conceivable that there are many cryptographic and communication tasks for which coherence is the resource.  Therefore, 

 It seems, therefore, that speakable coherence may be a resource for some quantum information-processing tasks and not for others. Greater clarity on the applications of speakable coherence would further the project of finding which sets of free operations that can define speakable coherence and which are physically justified. 
%Clarifying just how to understand speakable coherence 
%Therefore, despite the caveats above, coherence might still be the relevant resource for some quantum information-processing tasks. 
%there is currently no demonstration that coherence is {\em not} the relevant resource for any quantum information-processing task, 
%It seems worthwhile, therefore, to consider the problem of understanding speakable coherence as a resource. 

%\color{red}
%[Where should we cite this work by Vedral:  \cite{yadin2014new}?]\color{black}
%\color{black}

\section{Acknowledgements}
We acknowledge helpful discussions with Gilad Gour, Eric Chitambar, Paola Cappellaro,
Tommaso Roscilde,  Gerardo Adesso, Alex Streltsov, Julio I. de Vicente, and Martin Plenio. 
%a useful correspondence with the authors of Ref.~\cite{Coh_Plenio}.   
Research at Perimeter Institute is supported in part by the Government of Canada through NSERC and by the Province of Ontario through MRI. IM acknowledges support from grants ARO W911NF-12-1-0541 and NSF CCF-1254119. 

\emph{Note Added}: During the preparation of this article, we became aware of independent work by Gour and Chitambar,  which also studies the physical relevance of incoherent operations  and the class of dephasing-covariant operations \cite{chitambar2016critical}.
 \color{black}
\bibliography{Ref_v11}

\end{document}